\newif\ifdraft\draftfalse  
\newif\ifanon\anonfalse     
\newif\iffull\fullfalse   
\newif\iflongrefs\longrefsfalse 
\newif\ifbackref\backreffalse 
\newif\ifsooner\soonerfalse
\newif\iflater\laterfalse
\newif\ifcamera\cameratrue    
\newif\ifcheckpagebudget\checkpagebudgetfalse
\newcommand{\xxx}{}
\makeatletter \@input{texdirectives.tex} \makeatother

\documentclass[authoryear,twocolumn,
\ifcamera\else preprint,nocopyrightspace\fi]{sigplanconf}

\usepackage{afterpage}
\usepackage{color}
\usepackage{amssymb,amsmath,amsfonts}
\usepackage{extarrows}
\usepackage{mathptmx}
\usepackage{version}
\usepackage{xspace}
\usepackage{graphicx}
\usepackage{listings}
\usepackage{wrapfig}
\usepackage{ifpdf}
\usepackage{semantic}
\usepackage{natbib}
\newcommand\citepos[1]{\citeauthor{#1}'s\ (\citeyear{#1})}
\usepackage{amsthm}
\usepackage{stmaryrd}
\usepackage{subfigure}
\usepackage{multirow}
\usepackage{proof}
\usepackage{mathpartir}
\ifcamera
\usepackage{flushend} 
\fi
\definecolor{darkblue}{rgb}{0.0,0.0,0.3}

\usepackage[
    pdftex,%
    pdfpagelabels\ifcamera\else,%
    colorlinks,%
    linkcolor=darkblue,%
    citecolor=darkblue,%
    filecolor=darkblue,%
    urlcolor=darkblue\fi%
]{hyperref}

\usepackage{verbatim}
\usepackage{url}
\usepackage{natbib}
\bibpunct();A{},
\let\cite=\citep

\usepackage{color}
\definecolor{dkblue}{rgb}{0,0.1,0.5}
\definecolor{dkgreen}{rgb}{0,0.4,0}
\definecolor{dkred}{rgb}{0.6,0,0}
\definecolor{linkColor}{rgb}{0,0,0.5}

\usepackage{url}
\usepackage{listings}
\usepackage{version}

\usepackage{microtype} 

\excludeversion{HIDE}

\newcommand\maybecolor[1]{\color{#1}}
\definecolor{lightblue}{rgb}{0.2,0.2,1.0}
\definecolor{dkred}{rgb}{0.5, 0.0,0.0}
\definecolor{lightgrey}{rgb}{0.8,0.8,0.8}
\definecolor{dkblue}{rgb}{0,0.1,0.5}
\definecolor{dkgreen}{rgb}{0,0.4,0}
\definecolor{dkred}{rgb}{0.6,0,0}
\definecolor{linkColor}{rgb}{0,0,0.5}
\definecolor{lightgray}{rgb}{.9,.9,.9}
\definecolor{darkgray}{rgb}{.4,.4,.4}
\definecolor{purple}{rgb}{0.65, 0.12, 0.82}

\let\ls\lstinline

\newcommand{\kw}[1]{\ensuremath{\mathsf{#1}}}
\def\lstcode#1#2#3{\lstinputlisting[linerange=#2-#3]{Code/#1}}
\def\lstfrag#1/#2.{\lstcode{#1}{#2Begin}{#2End}}

\newcommand{\atforbindingcolon}{\mathcode`@="003A}
\atforbindingcolon

\definecolor{gray}{rgb}{0.8,0.8,0.8}
\lstnewenvironment{rcfcode}
  {\lstset{frame=single,frameround=tttt,rulecolor=\color{gray},xrightmargin=7pt,xleftmargin=7pt}}
  {}

\makeatletter
  \newcommand{\addToLabel}[1]{%
    \protected@edef\@currentlabel{\@currentlabel#1}%
  }
\makeatother

\usepackage{aliascnt}

\newaliascnt{theoremI}{propositionI}
\newtheorem{theoremI}[theoremI]{Theorem}
\aliascntresetthe{theoremI}

\newaliascnt{lemmaI}{propositionI}
\newtheorem{lemmaI}[lemmaI]{Lemma}
\aliascntresetthe{lemmaI}

\newaliascnt{definitionI}{propositionI}

\aliascntresetthe{definitionI}

\newaliascnt{corollary}{propositionI}
\newtheorem{corollary}[corollary]{Corollary}
\aliascntresetthe{corollary}

\newenvironment{lemma}[1][]{%
  \lemmaI[#1]%
  \ifx\relax#1\relax\else
  \fi
}{%
  \endlemmaI
}

\newenvironment{theorem}[1][]{%
  \theoremI[#1]%
  \ifx\relax#1\relax\else
  \fi
}{%
  \endtheoremI
}

{\begin{trivlist}\item[]{\normalsize\bf Restatement of #1}\hspace*{4mm}\it}%
  {\end{trivlist}}

\newcommand{\hbra}{
  \hbox to \columnwidth{\vrule width0.3mm height 1.8mm depth-0.3mm
    \leaders\hrule height1.8mm depth-1.5mm\hfill
    \vrule width0.3mm height 1.8mm depth-0.3mm}}
\newcommand{\hket}{
  \hbox to \columnwidth{\vrule width0.3mm height1.5mm
    \leaders\hrule height0.3mm\hfill
    \vrule width0.3mm height1.5mm}}

\newcommand{\ratio}{.35}

 \lstnewenvironment{displaylisting}[1]
   {~\\\noindent\textbf{#1}\\[-.8ex]\hbra\vspace{-2ex}}
   {\vspace{-2ex}\hket\vspace{1ex}}

\newcounter{rule}

\newcommand{\True}{\kw{true}}   

\newcommand{\sig}[1]{\mathcal{S}}

\newcommand\nquad{\!\!\!\!}
\newcommand\nqquad{\nquad\nquad}

\iffull

\else

\fi

\newcommand\lam[2]{\lambda #1.#2}

\newcommand\tarr[2]{\ensuremath{#1{\longrightarrow}~#2}} 

\def\Snospace~{\S{}}

\newcommand*{\ii}[1]{\ensuremath{\mbox{\textit{#1}}}}

\newcommand\fstar{F$^\star$\xspace}

\newcommand\emf{{\sc emf}$^\star$\xspace}
\newcommand\deflang{{\sc dm}\xspace} 
\newcommand\emfST{{\sc emf}$^\star_{\text {\sc st}}$\xspace} 

\newcommand{\un}[1]{\ensuremath{\underline{#1}}}
\newcommand{\F}[2]{\mathrm{F}_{#1}~#2}
\newcommand{\G}[3]{\mathrm{G}^{#1}_{#2}(#3)}
\newcommand{\returnT}[1]{\mathrm{\bf return}_{\teff}~#1}
\newcommand{\bindT}[3]{\mathrm{\bf bind}_{\teff}~#1~\mathrm{\bf to}~#2~\mathrm{\bf in}~#3}

\newcommand{\sub}{\mathrm{s}}
\newcommand{\sG}{\sub_\Gamma}

\newcommand{\neu}{\ensuremath{n}}
\newcommand{\teff}{\ensuremath{\tau}}
\newcommand{\narr}{\xrightarrow{\neu}}
\renewcommand{\tarr}{\xrightarrow{\teff}}
\newcommand{\earr}{\xrightarrow{\epsilon}}
\newcommand{\dneg}[1]{(#1 -> \typez) -> \typez}
\newcommand{\DmG}{\Delta\mid\Gamma}
\newcommand{\fst}[1]{\mathrm{\bf fst}(#1)}
\newcommand{\snd}[1]{\mathrm{\bf snd}(#1)}
\newcommand{\inl}[1]{\mathrm{\bf inl}(#1)}
\newcommand{\inr}[1]{\mathrm{\bf inr}(#1)}
\newcommand{\mycases}[3]{\mathrm{\bf case}~#1~\mathrm{\bf inl}~#2;~\mathrm{\bf inr}~#3}
\newcommand{\eqdef}{=_{\mathrm{\bf def}}}
\newcommand{\bang}{\,!\,}

\newcommand{\cps}[1]{\ensuremath{#1^{\star}}}          
\newcommand{\cpst}[0]{$\star$-translation\xspace}  
\newcommand{\cpsts}[0]{$\star$-translations\xspace}  

\newcommand{\stg}{\mathrel{\lesssim}}

\definecolor{dkblue}{rgb}{0,0.1,0.5}
\definecolor{dkgreen}{rgb}{0,0.4,0}
\definecolor{dkred}{rgb}{0.6,0,0}
\definecolor{dkpurple}{rgb}{0.7,0,1.0}
\definecolor{olive}{rgb}{0.4, 0.4, 0.0}
\definecolor{teal}{rgb}{0.0,0.4,0.4}
\definecolor{azure}{rgb}{0.0, 0.5, 1.0}

\newcommand{\comm}[3]{\ifcheckpagebudget\else\ifdraft{\maybecolor{#1}[#2: #3]}\fi\fi}
\newcommand{\nik}[1]{\comm{dkpurple}{Nik}{#1}}
\newcommand{\ch}[1]{\comm{teal}{CH}{#1}}
\newcommand{\gm}[1]{\comm{dkgreen}{GM}{#1}} 
\newcommand{\gdp}[1]{\comm{blue}{GDP}{#1}} 
\newcommand{\aseem}[1]{\comm{magenta}{Aseem}{#1}}

\newcommand*{\EG}{e.g.,\xspace}
\newcommand*{\IE}{i.e.,\xspace}
\newcommand*{\ETAL}{et al.\xspace}
\newcommand*{\ETC}{etc.\xspace}

\newcommand\surl[1]{{\small\url{#1}}}
\newcommand\ssurl[1]{{\scriptsize\url{#1}}}

\usepackage[firstpage]{draftwatermark}
\SetWatermarkText{\hspace*{8in}\raisebox{6.5in}{\includegraphics[scale=0.1]{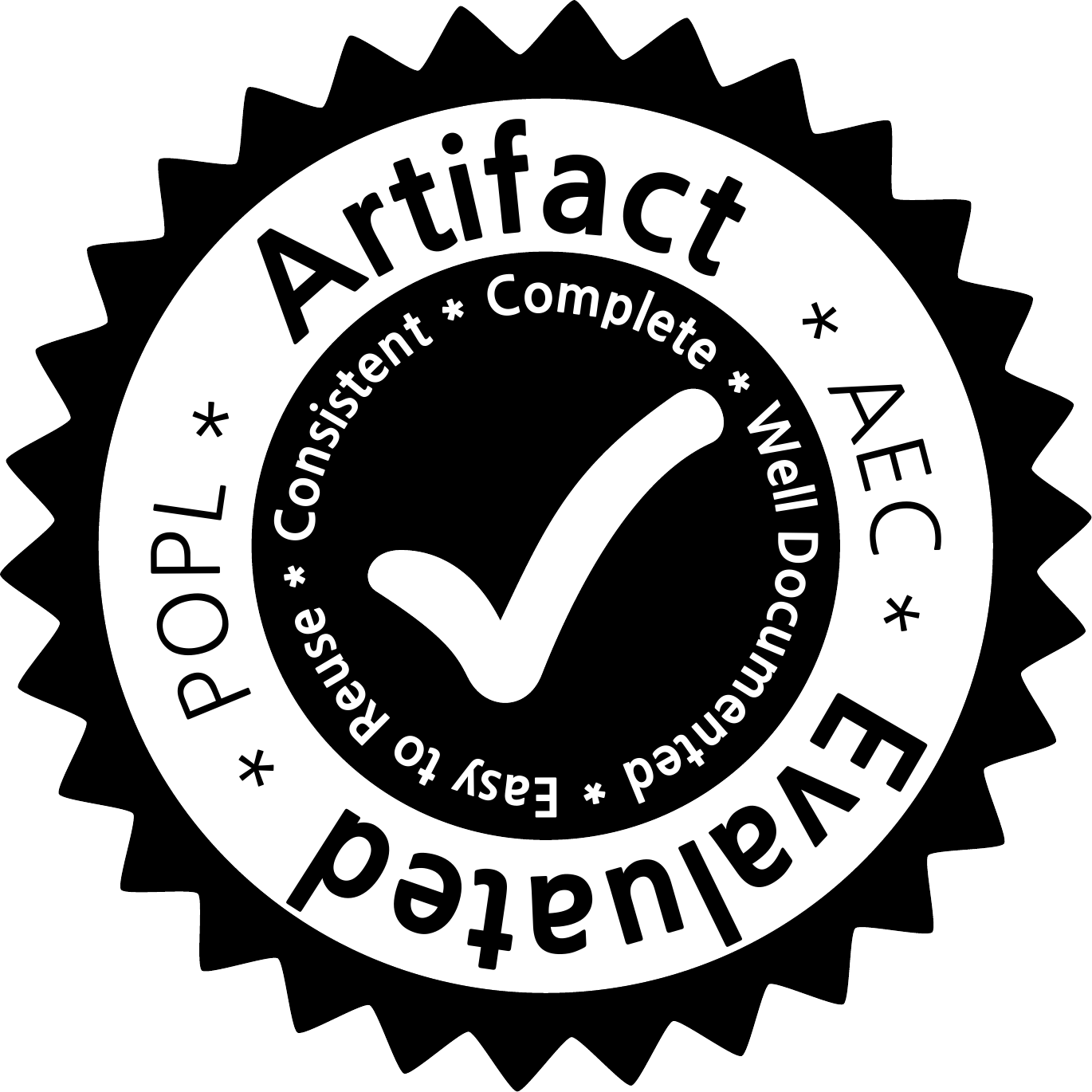}}}
\SetWatermarkAngle{0}

\begin{document}
\toappear{}
\titlebanner{}        
\preprintfooter{Draft\qquad}   

\title{ {
\ifcamera
\fontsize{18}{0}\selectfont
\else
\Huge
\fi
Dijkstra Monads for Free}\ifanon\vspace{-2cm}\fi}

\xxx{}


\authorinfo{
\ifanon\else
            Danel Ahman$^{1,2}$ \quad
            C\u{a}t\u{a}lin Hri\c{t}cu$^{1,3}$ \quad
            Kenji Maillard$^{1,3,4}$ \quad
            Guido Mart\'inez$^{3,5}$ \\[0.5em]
            Gordon Plotkin$^{1,2}$ \quad
            Jonathan Protzenko$^1$ \quad
            Aseem Rastogi\ifcamera$^6$\else$^1$\fi \quad
            Nikhil Swamy$^1$
\fi}
           {
\ifanon\else
\vspace{0.2cm}
            $^1$Microsoft Research\ifcamera, USA\fi\qquad
            $^2$University of Edinburgh\ifcamera, UK\fi\qquad
            $^3$Inria Paris\ifcamera, France\fi\ifcamera\\[0.5em]\else\qquad\fi
            $^4$ENS Paris\ifcamera, France\fi\qquad
            $^5$\ifcamera UNR, Argentina\else Rosario National University\fi\qquad
            \ifcamera$^6$Microsoft Research, India\fi
\fi}{}

\maketitle

\begin{abstract}
{\em Dijkstra monads} enable a dependent type theory to be enhanced
with support for specifying and verifying effectful code via weakest
preconditions.
Together with their closely related counterparts, {\em Hoare monads},
they provide the basis on which verification tools like \fstar, Hoare
Type Theory (HTT), and Ynot are built.

We show that Dijkstra monads can be derived ``for free''
by applying a continuation-passing style (CPS) translation to the
standard monadic definitions of the underlying computational effects.
Automatically deriving Dijkstra monads in this way provides a
correct-by-construction and efficient way of reasoning about
user-defined effects in dependent type theories.

We demonstrate these ideas in \emf, a new dependently typed
calculus, validating it via both formal proof and a prototype
implementation within \fstar. Besides equipping \fstar with a
more uniform and extensible effect system, \emf enables
a novel mixture of intrinsic and extrinsic proofs within \fstar.
\end{abstract}


\category{D.3.1}{Programming Languages}{Formal Definitions and Theory}[Semantics]
\category{F.3.1}{Logics and Meanings of Programs}{Specifying and Verifying and Reasoning about Programs}[Mechanical verification]
\keywords{\hspace{-1.5ex} verification; proof assistants; effectful programming; dependent types}

\section{Introduction}
\label{sec:intro}



In \citepos{Dijkstra75} weakest precondition semantics, stateful
computations transform postconditions, relating results and final
states to preconditions on input states. One can express 
such semantics
via 
a monad of predicate transformers, a so-called ``Dijkstra
monad''~\citep{fstar-pldi13, Jacobs15}. For instance, in the case of state,
the following monad arises:






{\lstset{keepspaces=true}
\begin{lstlisting}
       WP_ST a = post a -> pre        where post a = (a * state) -> Type
                                                                      pre = state -> Type

       return_WP_ST x post s0 = post (x, s0)
       bind_WP_ST f g post s0 = f (fun (x, s1) -> g x post s1) s0
\end{lstlisting}
}

\noindent The {\em weakest precondition (WP)} of a pure term \ls$e$ is
computed to be \ls$return_WP_ST e$,
and the WP of the sequential composition \ls$let x = e1 in e2$ is
computed to be \ls$bind_WP_ST wp1 (fun x. wp2)$, where \ls$wp1$
and \ls$wp2$ are the WPs of \ls$e1$ and \ls$e2$ respectively.

Building on previous work by \citet{nmb08htt},
\citet{fstar-pldi13, mumon}   
designed and implemented \fstar, a
dependently typed programming language whose type system can express 
WPs for higher-order, effectful programs 
via Dijkstra monads.

While this technique of specifying and verifying programs has been
relatively successful, there is 
still 
room for
improvement. Notably, in the version of \fstar described by~\citet{mumon}
  specifying Dijkstra monads in \fstar is 
a tedious, manual
process, requiring delicate meta-theoretic arguments to establish the
soundness of a user-provided predicate-transformer semantics with
respect to the semantics of effectful programs. These typically
require proofs of various correctness and admissibility conditions,
including the correspondence to the operational semantics and the
monad laws.
Furthermore, only a handful of primitively
supported effects are provided by the previous version of \fstar, and extending it with
user-defined effects is not possible.\ch{Why are we using present
  for something that's past?}

%



Rather than being given manually, we show that these predicate
transformers can be automatically derived by CPS'ing purely
functional definitions of monadic effects (with answer
type \ls$Type$).
For instance, rather than defining \ls$WP_ST$, one can simply compute it by
CPS'ing the familiar ST monad
(i.e., \ls$state -> a * state$), deriving 
\begin{lstlisting}
    WP_ST a = ((a * state) -> Type) -> state -> Type   $\qquad\mbox{(unfolded)}$
\end{lstlisting}

We apply this 
technique of deriving Dijkstra monads to \fstar. Our goal
is to make \fstar's effect system easier to configure and extensible
beyond 
its primitive effects. 
To do so,  we proceed as
follows:


\paragraph*{A monadic metalanguage}
We introduce \deflang, a simply typed, pure, monadic
metalanguage in which one can, in the spirit of \citet{Wadler92},
define a variety of monadic effects, ranging from state and
exceptions, to continuations. 

\paragraph*{A core dependent type theory with monadic reflection} To formally
study our improvements to \fstar, we define a new dependently typed
core calculus, \emf (for Explicitly Monadic \fstar) that features an
extensible effect system. \emf is loosely based on the Calculus of
Constructions~\cite{coquand1988} with (among other features): (1) a
predicative hierarchy of non-cumulative universes; (2) a
weakest-precondition calculus for pure programs; (3) refinement types;
and (4) a facility for representing user-defined effects using the
monadic reflection and reification of~\citet{Filinski94}, adapted to
the dependently typed setting. New effects can be introduced into the
language by defining them in terms of the built-in pure constructs,
related to each other via monad morphisms;  
each such effect
obtains a suitable weakest precondition calculus derived from the
underlying pure WPs. We prove the calculus strongly normalizing and
the WP calculus sound for total correctness verification, for both
pure and effectful programs.

\paragraph*{A CPS translation}
We 
give a type-directed CPS
translation from \deflang to \emf.
%
%
This can be used to extend \emf with
a new effect. One starts by defining a monadic effect (say \ls$ST$)
in \deflang.  Next, via the 
translation, 
one obtains the Dijkstra variant of that effect (\ls$WP_ST$) as a monotone, conjunctive predicate \emf transformer monad. 
Finally, a second translation from \deflang produces expression-level
terms representing monadic computations in \emf. A logical
relations proof shows that monadic computations are correctly specified by their
predicate transformers. We 
give examples of
these translations  
for monadic effects, such as 
state,
exceptions, information-flow control, continuations, and some
combinations thereof.

\paragraph*{Intrinsic and extrinsic proofs in \emf}
Effectful programs in \emf can be proven correct using one or both of
two different reasoning styles.
First, using the WP calculus, programs can be proven
intrinsically, by decorating their definitions with specifications
that must be proven to be at least as strong as their WPs. We refer
to this as the \emph{intrinsic} style, already familiar to users
of \fstar, and other tools like HTT~\citep{nmb08htt},
Dafny~\citep{leino10dafny}, and Why3~\cite{filliatre13why3}.

Second, through monadic reification, \emf allows terminating
effectful programs to be revealed as
their underlying pure implementations. Once reified,
one can reason about them
via the computational behavior of
their definitions.\gdp{I don't understand this last sentence}
As such, one may define effectful programs with
relatively uninformative types, and prove properties about them as
needed, via reification. This \emph{extrinsic} style of proving is
familiar to users of systems like Coq or Isabelle, where it is
routinely employed to reason about pure functions; using monads this
style extends smoothly to terminating effectful programs.
As in Coq an Isabelle, this extrinsic style only works for terminating
code; in this paper we do not consider divergent computations and
only discuss divergence as future work (\autoref{sec:conclusion})
%
%
%

\paragraph*{Primitive effects in a call-by-value semantics} We see \emf 
as a meta-language
%
%
in which to analyze and describe the semantics of terms in an
object language, \emfST, a call-by-value programming language with primitive
state. In the spirit of~\citet{Moggi89}, we show that \emf{} programs
that treat their \ls$ST$ effect abstractly soundly model \emfST
reductions---technically, we prove a
simulation between \emfST and \emf.
As such, our work is a strict improvement on the prior
support for primitive effects in \fstar: despite programming and
proving programs in a pure setting, stateful programs can still be
compiled to run efficiently in the primitively effectful \emfST, while
programs with other user-defined effects (\EG information-flow control) can, unlike before, be executed via their
pure encodings.

\paragraph*{A prototype implementation for \fstar}
We have adapted \fstar to benefit from the theory developed in this
paper, using a subset of \fstar itself as an implementation
of \deflang, and viewing \emfST as a model of its existing extraction
mechanism to OCaml.
Programmers can now configure \fstar's effect system using simple
monadic definitions, use \fstar to prove these definitions
correct, and then use our CPS transformation to derive the
Dijkstra monads required to configure \fstar's existing type-checker.
%
%
To benefit from the new extrinsic proving capabilities, we also
extended \fstar with two new typing rules, and changed its normalizer,
to handle monadic reflection and reification.
%

Several examples show how 
our work allows \fstar to be easily
extended beyond the primitive effects 
already supported, without
compromising its efficient primitive effect compilation strategy;
and how 
the new extrinsic proof style places reasoning about terminating effectful
programs in \fstar on an equal footing with its support for reasoning about
pure programs.

\subsection{Summary of Contributions}

The central contribution of our work is designing three closely
related lambda calculi, studying their metatheory and the connections
between them, and applying them to provide a formal and
practical foundation for a user-extensible effect system
for \fstar. Specifically,

\begin{enumerate}
\item \emf: A new dependent type theory with
      user-extensible, monadic effects; monadic reflection and
      reification; WPs; and refinement types. We prove that \emf is strongly
      normalizing and that its WPs are sound for total
      correctness (\autoref{sec:emf}).

\item \deflang: A simply typed language to define the expression-level monads
      that we use to extend \emf with effects. We define a CPS transformation of \deflang terms to
      derive Dijkstra monads from expression-level monads,
      as well as an elaboration of \deflang terms
       to \emf. Moreover, elaborated terms are proven
      to be in relation with their WPs (\autoref{sec:dm4f}).
      This is the first formal characterization of the relation
      between WPs and CPS at arbitrary order.
 
\item \emfST: A call-by-value language with primitive state,
      whose reductions are simulated by well-typed \emf
      terms (\autoref{sec:stateful}).

\item An implementation of these ideas
      within \fstar (\autoref{sec:emf-impl}, \autoref{sec:dm4f-impl})
      and several examples of free Dijkstra monads for
      user-defined effects (\autoref{sec:examples}). We
      highlight, in particular, the new ability to reason
      extrinsically about effectful terms.
\end{enumerate}

The \ifanon anonymous\fi{} auxiliary materials
\ifanon\else(\url{https://www.fstar-lang.org/papers/dm4free})\fi{}
contain appendices with complete definitions and proofs for the formal results
in \autoref{sec:emf},
\autoref{sec:dm4f}\iffull\;(Appendix \ref{appendix} below)\else\fi,
and \autoref{sec:stateful}.
%
%
%
%
\ifanon
The non-anonymous materials contain the extensions to \fstar{}
from \autoref{sec:emf-impl} and \autoref{sec:dm4f-impl} and the
examples from \autoref{sec:examples}.
\else
The \fstar{} source code (\url{https://github.com/FStarLang/FStar})
now includes the extensions from \autoref{sec:emf-impl} and
\autoref{sec:dm4f-impl} and the examples from \autoref{sec:examples}.
\fi
%

\ifcheckpagebudget\clearpage\fi
\section{Illustrative Examples}
\label{sec:examples}

We illustrate our main ideas using several examples from \fstar,
contrasting with the state of affairs in \fstar prior to our work.
We start by presenting the core WP calculus for pure programs
(\autoref{sec:ex-pure}), then show how state and exceptions can be
added to it
(\autoref{sec:ex-state}, \autoref{sec:ex-cps}, \autoref{sec:ex-reify}
and \autoref{sec:ex-combining}). Thereafter, we present several
additional examples, including modeling dynamically allocated
references (\autoref{sec:state-with-references}), reasoning about
primitive state (\autoref{sec:primitive-heap}), information-flow
control (\autoref{sec:ex-ifc}), and continuations
(\autoref{sec:ex-cont})---sections~\autoref{sec:emf}
and~\autoref{sec:dm4f} may be read mostly independently of these
additional examples.

\paragraph*{Notation:} The syntax
\ls@fun (b$_1$) ... (b$_n$) -> t@ introduces a 
lambda abstraction, where \ls@b$_i$@ ranges over binding
occurrences \ls$x:t$ declaring a variable \ls$x$ at type \ls$t$.  The
type \ls@b$_1$ -> ... -> b$_n$ -> c@ is the type of a curried
function, where \ls$c$ is a computation type---we emphasize the lack
of enclosing parentheses on the \ls@b$_i$@. We write just the type in \ls$b$
when the name is irrelevant, and \ls$t -> t'$ for
\ls$t -> Tot t'$.


\subsection{WPs for Pure Programs}
\label{sec:ex-pure}

Reasoning about purely functional programs is a relatively
well-understood activity: the type theories underlying systems like
Coq, Agda, and \fstar are already well-suited to the task.
Consider proving that pure term \ls$sqr = fun(x:int) -> x * x$ always returns
a non-negative integer.
A natural strategy is an {\em extrinsic} proof, which involves giving
\ls$sqr$ a simple type such as \ls$int -> Tot int$, the type of total
functions on integers, and then proving a lemma
\ls$forall x. sqr x >= 0$. In the case of \fstar, the proof of the lemma 
involves, first, a little computation to turn the goal into
\ls$forall x. x*x >= 0$, and then 
reasoning in the theory of integer arithmetic of the
Z3 SMT solver~\citep{MouraB08} to discharge the proof.

An alternative {\em intrinsic} proof style in \fstar involves
giving \ls$sqr$ type
\ls$x:int -> Pure int (fun post -> forall y. y>=0 ==> post y)$,
a dependent function type of the form \ls@x:t -> $c$@, where the
formal parameter \ls$x:t$ is in scope in the \emph{computation type}
$c$ to the right of the arrow. Computation types $c$ are either \ls$Tot t$
(for some type \ls$t$) or of the
form \ls$M t wp$, where \ls$M$ is an effect label, \ls$t$ is the
result type of the computation, and \ls$wp$ is a predicate transformer
specifying the semantics of the computation.
The computation type we give to \ls$sqr$ is of the form
\ls$Pure t wp$, the type of
\ls$t$-returning pure computations described by the predicate
transformer \ls$wp: (t -> Type) -> Type$, a function taking
postconditions on the result (predicates of type
\ls$t -> Type$), to preconditions. These predicate
transformers form a Dijkstra monad.
In this case,
the \ls$wp$ states that to prove any property \ls$post$ of \ls$sqr x$,
it suffices to prove \ls$post y$, for all non-negative \ls$y$---as
such, it states our goal that \ls$sqr x$ is non-negative. To
prove \ls$sqr$ can be given this type, \fstar infers a weakest
precondition for \ls$sqr x$, namely \ls$fun post -> post (x * x)$ and
aims to prove that the predicate transformer we specified is at least
as strong as the weakest one it inferred:
\ls$forall post. (forall y. y>=0 ==> post y) ==> post (x*x)$,
which is discharged automatically by Z3.
For pure programs, this intrinsic proof style may seem like
overkill and, indeed, it often is. But, as we will see, this mechanism for
reasoning about pure terms via WPs is a basic capability that we can
leverage for reasoning about terms with more complex, effectful
semantics. 

\subsection{Adding WPs for State}
\label{sec:ex-state}

Consider proving that 
\ls$incr _ = let x = get() in put (x + 1)$ produces an output
state greater than its input state. Since this program has the state
effect, a proof by extrinsic reasoning is not completely straightforward,
because reducing an effectful computation within a logic may not be
meaningful. Instead, tools like Ynot~\cite{ynot},
HTT~\cite{nmb08htt}, and \fstar only support the intrinsic proof style.
In the case of \fstar, this involves the use of
a computation type \ls$STATE0 t wp$,
where \ls$wp: WP_ST t$
and for our simple example we take
\ls$WP_ST t = ((t * int) -> Type) -> int -> Type$,
\IE the Dijkstra state monad from
\autoref{sec:intro} with \ls$state=int$.

Using the \ls$STATE0$
computation type in \fstar, one can specify for \ls$incr$ the type
\ls$unit -> STATE0 unit (fun post s0 -> forall s1. s1 > s0 ==> post ((), s1))$.
%
%
That is, to prove any postcondition \ls$post$ of \ls$incr$, it
suffices to prove
\ls$post ((), s1)$ for every \ls$s1$ greater than \ls$s0$, the
initial state---this is the statement of our goal.
The proof in \fstar currently involves:


\begin{enumerate}
\item As discussed already in \autoref{sec:intro}, one must
define \ls$WP_ST t$, its return and bind combinators, proving that
these specifications are sound with respect to the operational
semantics of state.

\item The primitive effectful actions, 
\ls$get$ and \ls$put$ are
assumed to have the types below---again, these types must be proven
sound with respect to the operational semantics of \fstar.
\begin{lstlisting}
get : unit -> STATE0 int (fun post s0 -> post (s0, s0))
put : x:int -> STATE0 unit (fun post _ -> post ((), x))
\end{lstlisting}

\item Following the rule for sequential composition sketched
in \autoref{sec:intro}, \fstar uses the specifications of
\ls$get$ and \ls$put$ to compute
\ls$bind_ST_WP wp_get (fun x -> wp_put (x + 1))$ as the WP of
\ls$incr$, which reduces to \ls$fun post s0 -> post ((), s0 + 1)$.

\item The final step requires proving that the computed WP is
at least as weak as the specified goal, which boils down to showing that
\ls$s0 + 1 > s0$, which \fstar and Z3 handle automatically.
\end{enumerate}

The first two steps above correspond to adding a new effect
to \fstar. The cost of this is amortized by the much more frequent and
relatively automatic steps 3 and 4. However, adding a new effect
to \fstar is currently an expert activity, carried out mainly by the
language designers themselves. This is in large part because the first
two steps above are both tedious and highly technical: a dangerous
mixture that can go wrong very easily.

Our primary goal is to simplify those first two steps,
allowing effects to be added to \fstar more easily and with fewer
meta-level arguments to trust. Besides, although \fstar supports
customization of its effect system, it only allows programmers to
specify refinements of a fixed set of existing effects inherited from ML,
namely,
state, exceptions, and divergence. For example, an \fstar programmer
can refine the state effect into three sub-effects for reading,
writing, and allocation; but, she cannot add a new effect like
alternative combinations of state and exceptions, non-determinism,
continuations, etc.  We aim for a more flexible, trustworthy mechanism
for extending \fstar beyond the primitive effects it currently
supports. Furthermore, we wish to place reasoning about terminating
effectful programs on an equal footing with pure ones,
supporting mixtures of intrinsic and extrinsic proofs for both.

\subsection{CPS'ing Monads to Dijkstra Monads}
\label{sec:ex-cps}

Instead of manually specifying \ls$WP_ST$, we program a traditional \ls$ST$
monad and derive \ls$WP_ST$ using a CPS
transform. In \autoref{sec:deflang} we formally present
\deflang, a simply typed language in which to define monadic effects.
\deflang itself contains a single primitive identity monad \ls$tau$,
which (as will be explained shortly) is used to control the CPS
transform. We have implemented \deflang as a subset of \fstar, and for
the informal presentation here we use the concrete syntax of our
implementation. What follows is an unsurprising definition of a state
monad \ls$st a$, the type of total functions from \ls$s$ to identity
computations returning a pair \ls$(a * s)$.

\begin{lstlisting}
let st a = s -> tau (a * s)
let return (x:a) : st a = fun s0 -> x, s0
let bind (f:st a) (g:a -> st b) : st b = fun s0 -> let x,s1 = f s0 in g x s1
let get () : st a = fun s0 -> s0, s0
let put (x:s) : st unit = fun _ -> (), x
\end{lstlisting}

\noindent This being a subset of \fstar, we can use it to prove that this
definition is indeed a monad: proofs of the three monad laws
for \ls$st$ are discharged automatically by \fstar below (\ls$feq$ is
extensional equality on functions, and \ls$assert p$ requests \fstar to prove
\ls$p$ statically). Other identities relating combinations of
\ls$get$ and \ls$put$ can be proven similarly.

\begin{lstlisting}
let right_unit_st (f:st 'a) = assert (feq (bind f return) f)
let left_unit_st (x:'a) (f:('a -> st 'b)) = assert (feq (bind (return x) f) (f x))
let assoc_st (f:st 'a) (g:('a -> st 'b)) (h:('b -> st 'c))
   = assert (feq (bind f (fun x -> bind (g x) h)) (bind (bind f g) h))
\end{lstlisting}

\noindent We then follow a two-step recipe to add an effect like
\ls$st$
to \fstar:


\paragraph{Step 1}
To derive the Dijkstra monad variant of \ls$st$, we apply a selective
CPS transformation called the \cpst (\autoref{sec:cps}); first, on type
\ls$st a$; then, on the various monadic operations. CPS'ing only those arrows
that have \ls$tau$-computation co-domains, we obtain:

\smallskip
\begin{tabular}{lcl}
$\cps{\kw{(st\ a)}}$ & = & \ls$s -> ((a * s) -> Type) -> Type$ \\
$\cps{\kw{return}}$ & = & \ls$fun x s0 post -> post (x, s0)$ \\
$\cps{\kw{bind}}$   & = & \ls$fun f g s0 post -> f s0 (fun (x,s1) -> g x s1 post)$ \\
$\cps{\kw{get}}$    & = & \ls$fun () s0 post -> post (s0, s0)$ \\
$\cps{\kw{put}}$    & = & \ls$fun x _ post -> post ((), x)$
\end{tabular}
\smallskip

\noindent Except for a reordering of arguments, the terms above are identical
to the analogous definitions for \ls$WP_ST$.{}
We prove that the \cpst preserves equality: so, having shown the monad
laws for \ls$st a$, we automatically obtain the monad laws for
$\cps{\kw{(st\ a)}}$. We also prove that every predicate transformer
produced by the \cpst is monotone (it maps weaker postconditions to
weaker preconditions) and conjunctive (it distributes over conjunctions
and universals, \IE infinite conjunctions, on the postcondition).

\paragraph{Step 2}
The \cpst yields a predicate transformer semantics for a new monadic
effect, however, we still need a way to extend \fstar with the
computational behavior of the new effect. For this, we define a second
translation, which elaborates the definitions of the new monad and its
associated actions to \ls$Pure$ computations in \fstar.
A first rough approximation of what we prove is that for a
well-typed \deflang computation $\kw{e} : \tau~\kw{t}$, its elaboration
$\un{\kw{e}}$ has type $\kw{Pure}~\un{\kw{t}}~\cps{\kw{e}}$
in \emf.

\newcommand\unst[2]{\ensuremath{\un{\mathsf{st}}~\mathsf{#1}~\mathsf{#2}}}

The first-order cases are particularly simple: for example,
$\un{\kw{return}} = \kw{return}$ has type
\ls@x:a -> Pure a (return$^\star$ x)@ in \emf; and 
$\un{\kw{get}} = \kw{get}$ has type
\ls@u:unit -> Pure s (get$^\star$ u)@ in \emf.
For a higher-order example, we sketch the elaboration of \ls$bind$ below,
writing \unst{t}{wp} for \ls@s0:s -> Pure t (wp s0)@: \\

\begin{tabular}{lcl}
$\un{\mathsf{bind}}$ \nquad&\nquad:\nquad&\nquad  \ls@wpf:(st a)$^\star$ -> f:$\unst{a}{wpf}$@\\
\nquad&\nquad\ls$->$\nquad&\nquad \ls@wpg:(a -> (st b)$^\star$) -> g:(x:a -> $\unst{b}{wpg x}$)@ \\
\nquad&\nquad\ls$->$\nquad&\nquad $\unst{b}{(bind^\star~wpf~wpg)}$ \\
\nquad&\nquad\ls$= $\nquad&\nquad \ls@fun wpf f wpg g s0 -> let x, s1 = f s0 in g x s1@ \\
\end{tabular}\\

Intuitively, a function in \deflang (like \ls$bind$) that abstracts
over computations (\ls$f$ and \ls$g$) is elaborated to a function
($\un{\mathsf{bind}}$) in \emf that abstracts both over those
computations (\ls$f$ and \ls$g$ again, but at their elaborated types)
as well as the WP specifications of those computations (\ls$wpf$
and \ls$wpg$). The result type of $\un{\mathsf{bind}}$ shows that it
returns a computation whose specification matches \ls@bind$^\star$@,
i.e., the result of the CPS'ing \cpst.


In other words, the WPs computed by \fstar for monads implemented
as \ls$Pure$ programs correspond exactly to what one gets by CPS'ing
the monads. At first, this struck us as just a happy coincidence,
although, of course, we now know that it must be so. We see our proof
of this fact as providing a precise characterization of the 
close connection between and WPs and CPS transformations.
%

\subsection{Reify and Reflect, for Abstraction and Proving}
\label{sec:ex-reify}

Unlike prior \fstar formalizations which included divergence,
primitive exception and state effects, the only primitive monad
in \emf is for \ls$Pure$ computations. Except for divergence, we can
encode other effects using their pure representations; we leave
divergence for future work.
%
Although the translations from \deflang yield pure definitions of
monads in \fstar, programming directly against those pure
implementations is undesirable, since this may break abstractions. For
instance, consider an integer-state monad whose state is expected to
monotonically increase: revealing its representation as a pure term
makes it hard to enforce this invariant. We rely
on \citepos{Filinski94} monadic reflection for
controlling abstraction.

Continuing our example, introducing the state effect in \fstar
produces a new computation type \ls@ST (a:Type) (wp: (st a)$^\star$)@
and two coercions:

\begin{lstlisting}
reify : ST a wp -> s0:s -> Pure (a * s) (wp s0)
reflect : (s0:s -> Pure (a * s) (wp s0)) -> ST a wp
\end{lstlisting}

\noindent The \ls$reify$ coercion reveals the representation of an
\ls$ST$ computation
as a \ls$Pure$ function, while \ls$reflect$ encapsulates a \ls$Pure$ function
as a stateful computation. As we will see in subsequent sections
(\autoref{sec:ex-combining} and \autoref{sec:state-with-references}), in
some cases to preserve abstractions, one or both of these coercions
will need to be removed, or restricted in various ways
(\autoref{sec:primitive-heap}).

To introduce the actions from \deflang as effectful actions 
in \fstar, we
\ls$reflect$ the pure terms produced by the elaboration
from \deflang to \emf, obtaining actions for the newly introduced
computation type. For example, after reflection the actions \ls$get$
and \ls$put$ appear within \fstar at the types below:

\begin{lstlisting}
get : unit -> ST s (get$^\star$ ())
put : s1:s -> ST unit (put$^\star$ s1)
\end{lstlisting}

\ch{In Section 3 we call these \ls$ST.get$ and \ls$ST.put$. Could
  consider doing that here (and in 2.2) too?}

\ch{The above implicit reification is not specific only to actions;
  \ls$ST.return$, \ls$ST.bind$, \ls$ST.lift$, they all work like this.
  And in fact when you write let below that's actually a \ls$ST.bind$, right?
  Might want to make that explicit or is it too much of a detail?}

As in \autoref{sec:ex-state}, we can still program stateful functions and
prove them intrinsically, by providing detailed specifications to
augment their definitions---of course, the first two steps of the
process there are now automatic. However, we now have a means of doing
extrinsic proofs by reifying stateful programs, as shown below
(taking \ls$s=int$).

\begin{lstlisting}
let StNull a = ST a (fun s0 post -> forall x. post x)
let incr _ : StNull unit = let n = get() in put (n + 1)
let incr_increases (s0:s) = assert (snd (reify (incr()) s0) = s0 + 1)
\end{lstlisting}

\noindent The \ls$StNull unit$ annotation on the second line above
gives a weak specification for \ls$incr$.
However, later, when a particular property
of \ls$incr$ is required, we can recover it by reasoning extrinsically
about the reification of \ls$incr()$ as a pure term.

\subsection{Combining Monads: State and Exceptions, in Two Ways}
\label{sec:ex-combining}

To add more effects to \fstar, one can simply repeat the methodology
outlined above. For instance, one can use \deflang to define
\ls$exn a = unit -> tau (option a)$ in the obvious way (the \ls$unit$
is necessary, cf. \autoref{sec:deflang}), our automated two-step recipe
extends \fstar with an effect for terminating programs that may raise
exceptions. Of course, we would like to combine the effects to equip
stateful programs with exceptions and, here, we come to a familiar
fork in the road.

State and exceptions can be combined in two mutually incompatible
ways. In \deflang, we can define both \ls$stexn a = s -> tau ((option a) * s)$ and
\ls$exnst a = s -> tau (option (a * s))$. The former is more familiar to
most programmers: raising an exception preserves the state; the latter
discards the state when an exception is raised, which though less
common, is also useful. We focus first on \ls$exnst$ and then discuss
a variant of \ls$stexn$.

\paragraph*{Relating \ls$st$ and \ls$exnst$} Translating \ls$st$ (as before)
and \ls$exnst$ to \fstar gives us two unrelated effects \ls$ST$
and \ls$ExnST$. To promote \ls$ST$ computations to \ls$ExnST$, we
define a \ls$lift$ relating \ls$st$ to \ls$exnst$, their pure
representations in \deflang, and prove that it is a monad morphism.

\begin{lstlisting}
let lift (f:st a) : exnst a = fun s0 -> Some (f s0)
let lift_is_an_st_exnst_morphism =
 assert (forall x. feq (lift (ST.return x)) (ExnST.return x));
 assert (forall f g. feq (lift (ST.bind f g)) (ExtST.bind (lift f) (fun x -> lift (g x))))
\end{lstlisting}

Applying our two-step translation to
\ls$lift$, we obtain in \fstar a computation-type coercion
from 
\ls@ST a wp@ to \ls@ExnST a (lift$^\star$ wp)@.
\gm{Mention that this lift preserves monotonicity/conjunctivity?}
Through this coercion, and through \fstar's existing inference
algorithm~\citep{swamy11coco, mumon}, \ls$ST$ computations are
implicitly promoted to \ls$ExnST$ computations whenever needed. In
particular, the \ls$ST$ actions, \ls$get$ and \ls$put$, are implicitly
available with \ls$ExnST$. All that remains is to define an additional
action, \ls$raise = fun () s0 -> None$, which gets elaborated and
reflected to \fstar at the type
\ls$unit -> ExnST a (fun _ p -> p None)$.
\ch{This is a bit backwards,
  raise has to be defined in DM with exnst (or even exn), so before this point.
  I guess this is only about the elaboration and reification part,
  but then again that's not specific only to actions.}
\ch{Explaining actions in terms of reflection seems quite silly
  for one more reason: some effects do not have reflect. See stexnC below,
  although there this is explained in terms of hiding reflect not removing
  it. Why not take the same stance in the implementation then?}

\ls$ExnST$ programs in \fstar can be verified intrinsically
and extrinsically. For an intrinsic proof, we show \ls$div_intrinsic$
below, which raises an exception on a divide-by-zero. To prove it, we
make use of an abbreviation \ls$ExnSt a pre post$, which lets us write
specifications using pre- and postconditions instead of predicate
transformers.

\begin{lstlisting}
let ExnSt a pre post =
  ExnST a (fun s0 p -> pre s0 /\ forall x. post s0 x ==> p x)
let div_intrinsic i j : ExnSt int
  (requires (fun _ -> True))
  (ensures (fun s0 x -> match x with 
		     | None -> j=0 
		     | Some (z, s1) -> s0 = s1 /\ j <> 0 /\ z = i / j))
  = if j=0 then raise () else i / j
\end{lstlisting}

Alternatively, for an extrinsic proof, we give a weak specification for
\ls$div_extrinsic$ and verify it by reasoning about
its reified definition separately. This time, we add a call
to \ls$incr$ in the \ls$ST$ effect in case of a
division-by-zero. \fstar's type inference lifts \ls$incr$
to \ls$ExnST$ as required by the context. However, as the proof shows,
the \ls$incr$ has no effect, since the raise that follows it discards
the state.

\begin{lstlisting}
let ExnStNull a = ExnST a (fun s0 post -> forall x. post x)
let div_extrinsic i j : ExnStNull int = if j=0 then (incr(); raise ()) else i / j
let lemma_div_extrinsic i j =
  assert (match reify (div_extrinsic i j) 0 with
          | None -> j = 0
	  | Some (z, 0) -> j <> 0 /\ z = i / j)
\end{lstlisting}

Using \ls$reify$ and \ls$reflect$ we can also build exception
handlers, following ideas of Filinski~\cite{Filinski99}. For example,
in \ls$try_div$ below, we use a handler and (under-)specify that it
never raises an exception.

\begin{lstlisting}
let try_div i j : ExnSt int
    (requires (fun _ -> True))
    (ensures (fun _ x -> Option.isSome x))
  = reflect (fun s0 -> match reify (div_intrinsic i j) s0 with
                  | None -> Some (0, s0)
                  | x -> x)
\end{lstlisting}

More systematically, we can first program
a~\citet{BentonKennedy2001Exceptional} exception handler in \deflang,
namely, as a term of type
\begin{center}\ls$exnst a -> (unit -> exnst b) -> (a -> exnst b) -> exnst b$\end{center}
 and then translate it to \fstar, thereby obtaining a weakest
precondition rule for it for free. More generally, adapting
the algebraic effect
handlers of~\citet{PlotkinP09} to user-defined monads \ls$m$,
handlers can be programmed in \deflang as terms of type
\begin{center} \ls$m a -> (m b -> b) -> (a -> b) -> b$ \end{center}
and then imported to \fstar. We
leave a more thorough investigation of 
such effect handlers for Dijkstra monads to the future.

\paragraph*{An exception-counting state monad: \ls$stexnC$} For another
combination of state and exceptions, we define \ls$stexnC$, which in
addition to combining state and exceptions (in the familiar way), also
introduces an additional integer output that counts the number of
exceptions that are raised.
%
%
In \deflang, we write:
%

\begin{lstlisting}
let stexnC a = s -> tau (option a * (s * int))
let return (x:a) = fun s -> Some x, (s, 0)
let bind (m:stexnC a) (f:a -> stexnC b) = fun s0 -> let r0 = m s0 in
    match r0 with
    | None, (s1, c1) -> None, (s1, c1)
    | Some r, (s1, c1) -> let res, (s, c2) = f r s1
                       in res (s, c1 + c2)
let raise () : stexnC a = fun s -> None, (s, 1)
let lift (f:st a) : stexnC a = fun s -> let x, s1 = f s in Some x, (s1, 0)
\end{lstlisting}

\noindent Notice that \ls$raise$ returns an exception count of 1.
This count is added up in \ls$bind$, \`{a} la
writer monad. Adding \ls$stexnC$
to \fstar proceeds as before. But, we need to be a bit careful with
how we use reflection. In particular, an implicit invariant
of \ls$stexnC$ is that the exception count field in the result
is non-negative
and actually counts the number of raised exceptions. If a programmer
is allowed to reflect any
\ls$s -> Pure (option a * (s * int))  wp$ into an \ls$stexnC$
computation, then this invariant can be broken. Programmers can
rely on \fstar's module
system to simply forbid the use of \ls$stexnC.reflect$ in client
modules. Depending on the situation, the module providing the effect may
still reveal a restricted version of the reflect operator to a client,
e.g., we may only provide \ls$reflect_nonneg$ to clients, which
only supports reflecting computations whose exception count is not
negative. Of course, this only guarantees that the counter
over-approximates the number of exceptions raised, which may or may
not be acceptable.

\begin{lstlisting}
let reflect_nonneg (f: s -> Pure (option a * (s * int)) wp)
    : stexnC a (fun s0 post ->
      wp s0 (fun (r, (s1, n)) -> post (r, (s1, n)) /\ n >= 0))
    = reflect f
\end{lstlisting}

The standard combination of state and exceptions (i.e., \ls$stexn$)
was already provided primitively in \fstar. The other two combinations
shown here were not previously supported, since \fstar only allowed
OCaml effects. In the following more advanced subsections, we
present a heap model featuring dynamic allocation
(\autoref{sec:state-with-references}), the reconciliation of primitive
state and extrinsic reasoning via reify and reflect
(\autoref{sec:primitive-heap}), and encodings of two other
user-defined effects (a dynamic information-flow control monitor in
\autoref{sec:ex-ifc} and continuations in \autoref{sec:ex-cont}).



\subsection{State with References and Dynamic Allocation}
\label{sec:state-with-references}

The state monads that we have seen so far provide global state, with
\ls$get$ and \ls$put$ as the only actions. Using just these actions, we can
encode references and dynamic allocation by choosing a suitable
representation for the global state. There are many choices for this
representation with various tradeoffs, but a (simplified) model of
memory that we use in \fstar is the type \ls$heap$ shown
below:\footnote{Although expressible in \fstar, types like \ls$heap$
are not expressible in the \emf calculus of \S\ref{sec:emf} since
it lacks support for features like inductive types and universe
polymorphism.}

\begin{lstlisting}
type pre_heap = {
  next_addr: nat;
  mem   : nat -> Tot (option (a:Type & a))
}
type heap = h:pre_heap{forall (n:nat). n >= h.next_addr
                                ==> h.mem n==None}
\end{lstlisting}

A \ls$pre_heap$ is a pair of \ls$next_addr$, the next free memory
location and a memory \ls$mem$ mapping locations to possibly allocated
values. (``\ls@a:Type & a@'' is a dependent pair type of some
type \ls@a:Type@ and a value at that type). A \ls$heap$ is
a \ls$pre_heap$ with an invariant (stated as a refinement type) that
nothing is allocated beyond \ls$next_addr$.

By taking \ls$s=heap$ in the \ls$ST$ monad of the previous section, we can
program derived actions for allocation, reading, writing and
deallocating references---we show just \ls$alloc$ below; deallocation
is similar, while reading and writing require their references to be
allocated in the current state. First, however, we define an
abbreviation \ls$St a pre post$, which lets us write specifications
using pre- and postconditions instead of predicate transformers, which
can be more convenient---the \fstar keywords, \ls$requires$
and \ls$ensures$ are only there for readability and have no semantic
content.

\begin{lstlisting}
let St a pre post = ST a (fun h0 p ->
    pre h0 /\ (* pre: a predicate on the input state *)
    forall x h1. post h0 x h1 (* post relates result, initial and final states *)
       ==> p (x, h1))

abstract let ref (a:Type) = nat (* other modules cannot treat ref as nat *)

let alloc (a:Type) (init:a) : St (ref a)
 (requires (fun h -> True)) (* can allocate, assuming infinite mem. *)
 (ensures (fun h0 r h1 ->
   h0.mem r == None /\                            (* the ref r is fresh *)
   h1.mem r == Some (| a, init |) /\              (* initialized to init *)
   (forall s. r$\neq$s ==> h0.mem s == h1.mem s))) (* other refs not modified *)
    = let h0 = get () in                           (* get the current heap *)
      let r  = h0.next_addr in                     (* allocate at next_addr *)
      let h1 = {
        next_addr=h0.next_addr + 1;                (* bump and update mem *)
        mem = (fun r$^\prime$ -> if r = r$^\prime$ then Some (| a , x |) else h0.mem r$^\prime$)
      } in
      put h1; r (* put the new state and return the ref *)
\end{lstlisting}

\paragraph*{Forbidding recursion through the store} The
reader may wonder if adding mutable references would allow stateful
programs to diverge by recursing through the memory. This is forbidden
due to universe constraints. The type \ls$Type$ in \fstar includes an
implicit level drawn from a predicative, countable hierarchy of
universes. Written explicitly, the type \ls$heap$ lives in
universe \ls@Type$_{i + 1}$@ since it contains a map whose co-domain
is in \ls@Type$_i$@, for some universe level $i$. As such, while one
can allocate references like \ls$ref nat$ or
\ls$ref (nat -> Tot nat)$, importantly, \ls$ref (a -> ST b wp)$ is forbidden,
since the universe of \ls$a -> ST b wp$ is the universe of its representation
\ls@a -> h:heap -> Pure (b * heap) (wp h)@, which is
\ls@Type$_{i+1}$@. Thus, our \ls$heap$ model forbids
storing stateful functions altogether. More fine-grained
encodings are possible too, e.g., stratifying the heap into fragments
and storing stateful functions that can only read from lower strata.

\subsection{Relating \ls$heap$ to a Primitive Heap}
\label{sec:primitive-heap}

While one can execute programs using the \ls$ST$ monad instantiated
with \ls$heap$ as its state, in practice, for efficiency, the \fstar
compiler provides primitive support for state via its extraction
facility to OCaml. In such a setting, one needs a leap of faith
to believe that our model of the \ls$heap$ is
faithful to the concrete implementation of the OCaml heap, e.g., the
abstraction of \ls$ref a$ is important to ensure that \fstar programs
are parametric in the representation of OCaml references.\ch{Can't
  references be compared for equality in F*?}

More germane to this paper, compiling \ls$ST$ programs primitively
requires that they do not rely concretely on the representation
of \ls$ST a wp$ as \ls$h:heap -> Pure (a * heap) (wp h)$, since the
OCaml heap cannot be reified to a value. In \S\ref{sec:stateful}, we
show that source programs that are free of \ls$reflect$ and \ls$reify$
can indeed be safely compiled using primitive state
(\autoref{thm:simulation}). Without \ls$reflect$ and \ls$reify$, one
may rely on intrinsic proof to show generally useful properties of
programs. For example, one may use intrinsic proofs to show
that \ls$ST$ computations never use references after they are
deallocated, since reading and writing require their references to be
allocated in the current state. Note, we write \ls@r $\in$ h@ for
\ls$h.mem r == Some _$, indicating that \ls$r$ is allocated in \ls$h$.

\begin{lstlisting}
let incr (r:ref int) : St unit (requires (fun h -> r $\in$ h))
                         (ensures (fun h0 s h1 -> r $\in$ h1))
    = r := !r + 1
\end{lstlisting}

Of course, one would still like to show that \ls$incr$ increments its
reference. In the rest of this section, we show how we can safely
restore \ls$reflect$ and \ls$reify$ in the presence of primitive
state.

\paragraph*{Restoring reify and reflect for extrinsic proofs}
\label{sec:ghost-reify}

\fstar programs using primitive state are forbidden from 
using \ls$reify$ and \ls$reflect$ only in the \emph{executable} part
of a program---fragments of a program that are computationally
irrelevant (aka ``ghost'' code) are erased by the \fstar compiler and
are free to use these operators. As such, within specifications and
proofs, \ls$ST$ programs can be reasoned about extrinsically via
reification and reflection (say, for functional correctness), while
making use of intrinsically proven properties like memory safety.


To restrict their use, as described in \autoref{sec:ex-combining}, we
rely on \fstar's module system to hide both the \ls$reify$
and \ls$reflect$ operators from clients of a module \ls$FStar.State$
defining the \ls$ST$ effect. Instead, we expose to clients
only \ls$ghost_reify$, a function equivalent to \ls$reify$, but at the
signature shown below. Notice that the function's co-domain is marked
with the \ls$Ghost$ effect, meaning that it can only be used within
specifications (e.g., WPs and assertions)---any other use will be
flagged as a typing error by \fstar.

\begin{lstlisting}
  ghost_reify: (x:a -> ST b (wp x))
            -> Ghost (x:a -> s0:s -> Pure (b * s) (wp x s0))
\end{lstlisting}

The \ls$FStar.State$ module also provides 
\ls$refine_St$, a total function that allows a client to strengthen
the postcondition of an effectful function \ls$f$ to additionally
record that the returned value of \ls$f$ on any argument and input state
\ls$h0$ corresponds to the computational behavior of the (ghostly)
reification of \ls$f$. This allows a client to relate \ls$f$ to its
reification while remaining in a computationally relevant context.

\begin{lstlisting}
    let refine_St (f :(x:a -> St b (pre x) (post x)))
      : Tot (x:a -> St b (pre x) (fun h0 z h1 -> post x h0 z h1 /\
                                ghost_reify f x h0 == z, h1)
      = fun x -> STATE.reflect (reify (f x))
\end{lstlisting}

Reasoning using \ls$ghost_reify$ instead of \ls$reify$, clients can
still prove
\ls$incr_increases$ as in \autoref{sec:ex-reify}, making use of
\ls$incr$'s intrinsic specification to show that if the
reference \ls$r$ is allocated before calling \ls$incr$ it will still
be allocated afterwards.
\vfill

\begin{lstlisting}
let incr_increases (r:ref int) (h0:heap{r $\in$ h0}) : Ghost unit =
    let Some x0 = h0.mem r in
    let _, h1 = ghost_reify incr r h0 in
    let Some x1 = h1.mem r in
    assert (x1 = x0 + 1)
\end{lstlisting}

Further, in computationally relevant (non-ghost) code, \ls$refine_St$
allows us to reason using the concrete definition of \ls$incr$:

\begin{lstlisting}
     let r = ST.alloc 42 in
     let n0 = !r in
     refine_St incr r;
     let n1 = !r in
     assert(n1 == n0 + 1)
\end{lstlisting}

The intrinsic specification of \ls$incr$ does not constrain the final
value of \ls$r$, so calling \ls$incr$ directly here would not be
enough for proving the final assertion. By tagging the call-site
with \ls$refine_St$, we strengthen the specification of \ls$incr$
extrinsically, allowing the proof to complete as
in \ls$incr_increases$.







\subsection{Information Flow Control}
\label{sec:ex-ifc}

Information-flow control~\citep{SabelfeldMyers06IFC} is
a paradigm in which a program is deemed secure when one can prove that
its behavior observable to an adversary is independent of the secrets
the program may manipulate, i.e., it is \emph{non-interferent}.
Monadic reification allows us to prove
non-interference properties directly, by relating multiple runs of an
effectful program~\citep{benton04relational}. For example, take the
simple stateful program below:

\begin{lstlisting}
let ifc h = if h then (incr(); let y = get() in decr(); y) else get() + 1
\end{lstlisting}

It is easy to prove this program non-interferent via the extrinsic,
relational proof below, which states that regardless of its secret
input (\ls$h0, h1$), \ls$ifc$ when run in the same public initial
state (\ls$s0$) produces identical public outputs.
This generic extrinsic proof style is in
contrast to~\citet{BartheFGSSB14}, whose r\fstar is a custom
extension to \fstar supporting only intrinsic relational proofs.

\begin{lstlisting}
let ni_ifc = assert (forall h0 h1 s0. reify (ifc h0) s0 = reify (ifc h1) s0)
\end{lstlisting}

Aside from such relational proofs, with user-defined effects, it is
also possible to define monadic, dynamic information-flow control
monitors in \deflang, deferring non-interference checks to runtime,
and to reason about monitored programs in \fstar. Here's a
simplified example, inspired by the floating label approach of
LIO~\citep{stefan11LIO}.
For simplicity, we take the underlying monad to be \ls$exnst$, where
the state is a security label from a two-point lattice that
represents the secrecy of data that a computation may have observed so far.
\begin{lstlisting}
type label = Low | High
let difc a = label -> tau (option (a * label))
\end{lstlisting}

Once added to \fstar, we can provide two primitive actions to
interface with the outside world, where \ls$DIFC$ is the effect
corresponding to \ls$difc$. Importantly, writing to a public channel
using \ls$write Low$ when the current label is \ls$High$ causes a
dynamic failure signaling a potential \ls$Leak$ of secret information.

\begin{lstlisting}
let join l1 l2 = match l1, l2 with | _, High | High, _ -> High | _ -> Low
val read : l:label -> DIFC bool (fun l0 p -> forall b. p (Some (b, join l0 l)))
let flows l1 l2 = match l1, l2 with | High, Low -> false | _ -> true
val write : l:label -> bool -> DIFC unit (fun l0 p ->
  if flows l0 l then p (Some ((), l0)) else  p None)
\end{lstlisting}

As before, it is important
to not allow untrusted client code to reflect on \ls$DIFC$, since that
may allow it to declassify arbitrary secrets. Arguing that \ls$DIFC$
soundly enforces a form of termination-insensitive non-interference
requires a meta-level argument, much like that of~\citet{stefan11LIO}.

We can now write programs like the one below, and rely on the dynamic
checks to ensure they are secure.

\begin{lstlisting}
let b1, b2 = read Low, read Low in write Low (b1 && b2)
let b3 = read High in write High (b1 || b3); write Low (xor b3 b3)
\end{lstlisting}

In this case, we can also prove that the program fails with
a \ls$None$ at the last \ls$write Low$. In contrast to the relational
proof sketched earlier, dynamic information-flow control is conservative: even
though the last write reveals no information on the low channel, the
monitor still raises an error.

\subsection{CPS'ing the Continuation Monad}
\label{sec:ex-cont}

As a final example before our formal presentation, we ask the
irresistible question of whether we can get a Dijkstra monad for free
for the continuation monad itself---indeed, we can.

We start by defining the standard continuation monad, \ls$cont$,
in \deflang. Being a subset of \fstar, we can prove that it is indeed
a monad, automatically.

\begin{lstlisting}
let cont a = (a -> tau ans) -> tau ans
let return x = fun k -> k x
let bind f g k = f (fun x -> g x k)
(* cont is a monad *)
let right_unit_cont (f:cont 'a) = assert (bind f return == f);
let left_unit_cont (x:'a) (f:('a -> cont 'b)) = assert (bind (return x) f == f x)
let assoc_cont (f:cont 'a) (g:('a -> cont 'b)) (h:('b -> cont 'c)) =
   assert (bind f (fun x -> bind (g x) h) == bind (bind f g) h)
\end{lstlisting}

Following our two-step recipe, we derive the Dijkstra variant
of \ls$cont$, but first we define some abbreviations to keep the
notation manageable. The type \ls$kwp a$ is the type of a predicate
transformer specifying a continuation \ls$a -> tau ans$; and \ls$kans$
is the type of a predicate transformer of the computation that yields the
final answer.

\smallskip
\begin{tabular}{lclcl}
 \ls$kwp a$    & = & \ls$a -> kans$                  & = & \ls@(a -> tau ans)$^\star$@\\
 \ls$kans$     & = & \ls$(ans -> Type) -> Type$      & = & \ls@(tau ans)$^\star$@ \\
\end{tabular}
\smallskip

Using these abbreviations, we show the \cpst of \ls$cont$, \ls$return$
and \ls$bind$. Instead of being just a predicate transformer,
$\cps{\kw{(cont\ a)}}$ is a predicate-transformer transformer.

\smallskip
\begin{tabular}{lcl}
$\cps{\kw{(cont\ a)}}$ & = & \ls$kwp a -> kans$\\
$\cps{\kw{return}}$   & = & \ls$fun (x:a) (wp_k:kwp a) -> wp_k x$ \\
$\cps{\kw{bind}}$     & = & \ls$fun f g (wp_k:kwp b) -> f (fun (x:a) -> g x wp_k)$ \\
\end{tabular}
\smallskip

For step 2, we show the elaboration of \ls$return$ and \ls$bind$
to \fstar, using the abbreviation \ls$kt a wp$ for the type of the
elaborated term $\un{k}$, where the \deflang term $k$ is a
continuation of type
\ls$a -> tau ans$ and \ls@wp=$k^\star$@.
As illustrated in \autoref{sec:ex-cps}, elaborating higher-order
functions from \deflang to \fstar introduces additional arguments
corresponding to the predicate transformers of abstracted
computations.

\newcommand\elabcont{\ensuremath{\un{\mathsf{cont}}}}
\smallskip
\begin{tabular}{lcl}
\nquad\nqquad\ls$kt a wp$              \nquad&\nquad =       \nquad&\nquad \ls$x:a -> Pure ans (wp x)$    \\
\nquad\nqquad$\un{\kw{return}}$  \nquad&\nquad :       \nquad&\nquad \ls@x:a -> wpk:kwp a -> k:kt a wpk -> Pure ans (return$^\star$ x wpk)@\\
                           \nquad&\nquad =       \nquad&\nquad \ls$fun x wpk k -> k x$\\
\nquad\nqquad$\un{\kw{bind}}$    \nquad&\nquad :       \nquad&\nquad \ls@wpf:(cont a)$^\star$@ \\
                           \nquad&\nquad \ls$->$ \nquad&\nquad \ls@f:(wpk:kwp a -> k:kt a wpk -> Pure ans (wpf wpk))@ \\
                           \nquad&\nquad \ls$->$ \nquad&\nquad \ls@wpg:(a -> (cont b)$^\star$)@ \\
                           \nquad&\nquad \ls$->$ \nquad&\nquad \ls@g:(x:a -> wpk:kwp b -> k:kt b wpk -> Pure ans (wpg x wpk))@ \\
                           \nquad&\nquad \ls$->$ \nquad&\nquad \ls@wpk:kwp b@ \\
                           \nquad&\nquad \ls$->$ \nquad&\nquad \ls@k:kt b wpk@ \\
                           \nquad&\nquad \ls$->$ \nquad&\nquad \ls@Pure ans (bind$^\star$ wpf wpg wpk)@ \\
                           \nquad&\nquad =       \nquad&\nquad \ls$fun wpf f wpg g wpk k -> f (fun x -> wpg x wpk) (fun x -> g x wpk k)$
\end{tabular}
\smallskip

\noindent In the case of $\un{\kw{return}}$, we have one
additional argument for the predicate transformer of the
continuation \ls$k$---the type of the result shows how
$\un{\kw{return}}$ relates to $\cps{\kw{return}}$. The elaboration
$\un{\kw{bind}}$ involves many such additional parameters, but the
main point to take away is that its specification is given in terms of
$\cps{\kw{bind}}$, which is applied to the predicate transformers
\ls$wpf, wpg, wpk$, while $\un{\kw{bind}}$ was applied to the
computations \ls$f, g, k$.
In both cases, the definitions of $\un{\kw{return}}$ and
$\un{\kw{bind}}$ match their pre-images in \deflang aside from
abstracting over and passing around the additional WP arguments.

To better see the monadic structure in the types of $\un{\kw{return}}$
and $\un{\kw{bind}}$ we repeat these types, but this time writing
\ls@$\elabcont$ a wp@ for the type
\ls$wpk:kwp a -> k:kt a wpk -> Pure ans (wp wpk)$:

\smallskip
\begin{tabular}{lcl}
\nquad\nqquad$\un{\kw{return}}$  \nquad&\nquad :       \nquad&\nquad \ls@x:a -> $\elabcont$ a (return$^\star$ x)@ \\
\nquad\nqquad$\un{\kw{bind}}$    \nquad&\nquad :       \nquad&\nquad \ls@wpf:(cont a)$^\star$ -> f:$\elabcont$ a wpf@ \\
                                 \nquad&\nquad \ls$->$ \nquad&\nquad \ls@wpg:(a -> (cont b)$^\star$) -> g:(x:a -> $\elabcont$ b (wpg x))@ \\
                                 \nquad&\nquad \ls$->$ \nquad&\nquad \ls@$\elabcont$ b (bind$^\star$ wpf wpg)@
\end{tabular}

\ifcheckpagebudget\clearpage\fi
\section[Explicitly Monadic Fstar]{Explicitly Monadic \fstar}
\label{sec:emf}

\ch{Section 3: make it clearer what the limitations are:
- Div, Ghost, inductives (would be prerequisite for IO),
  universe polymorphism}




We begin our formal development by presenting \emf, an explicitly
typed, monadic core calculus intended to serve as a model
of \fstar. As seen above, the \fstar implementation includes an
inference algorithm \cite{mumon} so that source programs may omit all
explicit uses of the monadic \ls$return$, \ls$bind$ and \ls$lift$
operators. We do not revisit that inference algorithm here.
Furthermore, \emf lacks \fstar's support for divergent and ghost
computations, fixed points and their termination check, inductive
types, and universe polymorphism. We leave extending \emf to
accommodate all these features as future work, together with a formal
proof that after inference, \fstar terms can be elaborated into \emf
(along the lines of the elaboration of \citet{swamy11coco}).

\newcommand\refine[2]{\ensuremath{#1\{#2\}}}
\renewcommand\lam[2]{\ensuremath{\lambda #1. #2}}
\newcommand\product[2]{\ensuremath{#1 -> #2}}
\newcommand\casefull[5]{\ensuremath{\kw{case}_{#5} (#1~\kw{as}~#2)~#3~#4}}
\newcommand\run{\ensuremath{\kw{run}}}
\newcommand\reify{\ensuremath{\kw{reify}}}
\newcommand\freflect{\ensuremath{\kw{reflect}}}
\newcommand\mreturn[1][M]{\ensuremath{#1.\kw{return}}}
\newcommand\freturn[1][F]{\ensuremath{#1.\kw{return}}}
\newcommand\mbind[1][M]{\ensuremath{#1.\kw{bind}}}
\newcommand\fbind[1][F]{\ensuremath{#1.\kw{bind}}}
\newcommand\mlift[2][M]{\ensuremath{#1.\kw{lift}_{#2}}}
\newcommand\fact[1][F]{\ensuremath{#1.\kw{act}}}
\newcommand\mTot{\ensuremath{\kw{Tot}}}
\newcommand\mPure{\ensuremath{\kw{Pure}}}
\newcommand\mCont{\ensuremath{\kw{Cont}}}
\newcommand\purereturn[1][\mPure]{\ensuremath{#1.\kw{return}}}
\newcommand\purebinddmf[1][\mPure]{\ensuremath{#1.\kw{bind}}}

\newcommand{\pureretex}[2]{\purereturn~{#1}~{#2}}
\newcommand{\pureretin}[2]{\purereturn~{#2}}
\newcommand{\purebindex}[6]{\purebinddmf~{#1}~{#2}~{#3}~{#4}~{#5}~{#6}}
\newcommand{\purebindin}[6]{\purebinddmf~{#5}~{#6}}
\newcommand{\purerett}{\mathrm{\bf Pure.return}}
\newcommand{\purebindt}{\mathrm{\bf Pure.bind}}

\let\purebind\purebindex
\let\pureret\pureretex
%
\newcommand\Pure[2]{\mPure~{#1}\;(#2)}
\newcommand\caseimp[3]{\ensuremath{\kw{case} (#1)~{#2}~{#3}}} 
\newcommand\minl[1]{\mInl~#1}
\newcommand\minr[1]{\mInr~#1}
\newcommand\mfst[1]{\mFst~#1}
\newcommand\msnd[1]{\mSnd~#1}

\newcommand\mInl{\ensuremath{\kw{inl}}}
\newcommand\mInr{\ensuremath{\kw{inr}}}
\newcommand\mFst{\ensuremath{\kw{fst}}}
\newcommand\mSnd{\ensuremath{\kw{snd}}}

\newcommand\typez{\ensuremath{\kw{Type}_{0}}}
\newcommand\typei[1]{\ensuremath{\kw{Type}_{#1}}}
\newcommand\steps{\ensuremath{\longrightarrow}}
\newcommand\frepr[3][F]{\ii{F.\repr}~#2~#3}

\renewcommand*{\wp}[0]{\ii{wp}}
\newcommand*{\repr}[0]{\ii{repr}}

\newcommand\mliftE{\ensuremath{\un{\ii{M.lift}}_{\ii{M}'}}}
\newcommand\mliftWP{\ensuremath{\cps{\ii{M.lift}}_{\ii{M}'}}}
\begin{figure}
\[\begin{array}{lcl}
\multicolumn{3}{l}{\nquad\mbox{\textit{Terms}}}\\
 e, t, \wp, \phi & ::= & x
                  \mid  T
                  \mid  \refine{x@t}{\phi}
                  \mid  \lam{x@t}{e}
                  \mid  \product{x@t}{c}
                  \mid   e_1~e_2 \\ &
                  \mid & \casefull{e}{y}{x.e_1}{x.e_2}{t} 
                  \mid   \run~e
                  \mid   \reify~e \\&
                  \mid & \freflect~e
                  \mid   \mlift{M'}~t~\wp~e
                  \mid   \fact~\bar{e}\\&
                  \mid & \mreturn~t~e
                  \mid   \mbind~t_1~t_2~\wp_1~e_1~\wp_2~{x.e_2}
\\[1pt]
\multicolumn{3}{l}{\nquad\mbox{\textit{Computation types}}}\\
 c              & ::=  & \mTot~t
                 \mid    M~t~\wp~~\text{where}~M \in \{\mPure,~F\} \\
\end{array}
\]
\[\begin{array}{lcllll}
\multicolumn{3}{l}{\nquad\mbox{\textit{Signatures of monadic effects and lifts}}}\\
S               &\nquad::=\nquad& D \mid S, D \mid S, L\\
D               &\nquad::=\nquad& F \left\{\begin{array}{lllclll}
                    \repr              & = & t & ; &  \ii{wp\_type}      & = & t  \\
                    \un{\ii{return}}   & = & e & ; &  \cps{\ii{return}}  & = & \wp \\
                    \un{\ii{bind}}     & = & e & ; &  \cps{\ii{bind}}    & = & \wp \\
                    \un{\ii{act}}_j    & = & e & ; &  \cps{\ii{act}_j}   & = & \product{\overline{x_j@t_j}}{c_j}\\
                    \end{array}\right\}\\
L               &\nquad::=\nquad& \{~~\mliftE\;= e ; ~\mliftWP = wp~~\}
\end{array}\]
\caption{Syntax of \emf}
\label{fig:emf-syntax}
\end{figure}

\subsection{Syntax}
\label{sec:emf-syntax}

\autoref{fig:emf-syntax} shows the \emf syntax. We highlight several
key features.

\paragraph*{Expressions, types, WPs, and formulae} are all represented
uniformly as terms; however, to evoke their different uses, we
often write $e$ for expressions, $t$ for types, $\wp$ for WPs, and
$\phi$ for logical formulae.  Terms include variables ($x, y,
a, b, w$ etc.); refinement types $\refine{x@t}{\phi}$;
$\lambda$ abstractions;
dependent products with computation-type co-domains, $\product{x@t}c$
(with the sugar described in \S\ref{sec:examples}); and
applications. Constants $T$ include
\typei{i}, the $i$th level from a countable hierarchy of predicative
universes.\footnote{We have yet to model \fstar{}'s universe polymorphism,
making the universes in \emf less useful than the ones in \fstar. Lacking
universe polymorphism, we restrict computation to have results
in \ls@Type$_0$@. A simple remediation would be replicate the monad
definitions across the universe levels.} We also include constants for
non-dependent pairs and disjoint unions; the former are eliminated using
\mFst~and~\mSnd~(also constants), while the latter are eliminated
using \casefull{e}{y}{x.e_1}{x.e_2}{t}, which is standard dependent pattern
matching with an explicit return type $t$ and a name for the scrutinee
$y$, provided only when the dependency is necessary.

\paragraph*{Computation types $(c)$} include 
$\mTot~t$, the type of total $t$-returning terms, and
$M~t~\wp$, the type of a
computation with effect $M$, return type $t$, and behavior specified by the
predicate transformer $\wp$. Let $M$ range over the $\mPure$ effect as
well as user-defined effects $F$. 

\paragraph{Explicit monadic returns, binds, actions, lifts, reify, and reflect.}
$\mreturn$ and $\mbind$ are the monad operations for the effect $M$,
with explicit arguments for the types and predicate transformers.
$\mlift{M'}~t~\wp~e$ lifts the $e:M~t~\wp$ to $M'$.
A fully applied $F$ action is written $\fact~{\bar{e}}$.
The \ls$reify$ and \ls$reflect$ operators are for monadic reflection,
and \ls$run$ coerces a \ls$Pure$ computation to \ls$Tot$.

\newcommand\frepR[1][F]{\ensuremath{\ii{S.#1.repr}}}
\newcommand\fretE[1][F]{\ensuremath{\ii{S.#1.\un{return}}}}
\newcommand\fretWP[1][F]{\ensuremath{\cps{\ii{S.#1.return}}}}
\newcommand\fbindE[1][F]{\ensuremath{\ii{S.#1.\un{bind}}}}
\newcommand\fbindWP[1][F]{\ensuremath{\cps{\ii{S.#1.bind}}}}
\newcommand\mretWP[1][M]{\ensuremath{\cps{\ii{S.#1.return}}}}
\newcommand\mbindWP[1][M]{\ensuremath{\cps{\ii{S.#1.bind}}}}

\paragraph{Signatures for user-defined effects}
\emf is parameterized by a signature $S$. 
A user-defined effect $F~t~\wp$ is specified using $D$, the
result of translating a \deflang monad.
A definition $D$ is a record containing several fields:
\repr~is the type of an $F$ computation reified as a pure term,
\ii{wp\_type} is the type of the $\wp$ argument to $F$;
\un{\ii{return}}, \un{\ii{bind}}, and \un{\ii{act}}$_j$ are \emf
expressions, and \cps{\ii{return}}, \cps{\ii{bind}},
and \cps{\ii{act}_j} are \emf WPs (\ii{act}$_{j}$ is the $j^{th}$
action of F). We use \fretE~to denote the lookup
of the \un{\ii{return}} field from $F$'s definition in the signature
$S$, and similar notation for the other fields.

\iffull
For example, for the \ls|ST| monad from \autoref{sec:ex-cps}, we
have{\footnote{We use $\mathsf{sans~serif}$ font for the actual field values.}}:
\[\begin{array}{llll}
\ii{ST} \{
   & \ii{wp\_type}      & = & \lam{a}{s -> (a * s -> \kw{Type}_0) -> \kw{Type}_0}   \\
   & \repr              & = & \lam{a~w}{\product{s_0@s}{\kw{Pure}~(a * s)~(w~s_0)}} \\
   & \un{\ii{return}}   & = & \lam{a}{\un{\mathsf{return}}}\\
   & \cps{\ii{return}}  & = & \lam{a}{\cps{\mathsf{return}}} \\
   & \un{\ii{bind}}     & = & \lam{a~b}{\un{\mathsf{bind}}}\\
   & \cps{\ii{bind}}    & = & \lam{a~b}{\cps{\mathsf{bind}}} \\
   & \un{\ii{get}}      & = & {\un{\mathsf{get}}}\\
   & \cps{\ii{get}}     & = & {\cps{\mathsf{get}}}\\
   & \un{\ii{put}}      & = & {\un{\mathsf{put}}}\\
   & \cps{\ii{put}}     & = & {\cps{\mathsf{put}}} \qquad\}\\
\end{array}\]
where (as described in \autoref{sec:ex-cps})
$\un{\kw{return}} : \product{a}{\product{x@a}{\ii{repr}~a~(\cps{\ii{return}}~a~x)}}$;
and similarly for $\un{\kw{bind}}$, $\un{\kw{get}}$, and $\un{\kw{put}}$.
\fi
%

%
%
%

In addition to the monad definitions $D$, the
signature $S$ contains the definitions of lifts that contain
an \emf expression and an \emf WP. We use notations
$S.\mliftE$~and~$S.\mliftWP$ to look these up in $S$.
Finally, the signature always includes a fixed partial definition for
the $\mPure$ monad, only containing the following definitions:
\[
\begin{array}{l@{\hspace{3pt}}c@{\hspace{3pt}}l}
\ii{Pure} \,\{\, \ii{wp\_type} & = & \lambda a{:}\typez.~ (a -> \typez) -> \typez\\
~~~~~~~~~~~\ii{return}^\star & = & \lambda a{:}\typez.~
                       \lambda x{:}a.~
                       \lambda p{:}(a -> \typez).~ p~x\\
~~~~~~~~~~~\ii{bind}^\star 
      & = & \lambda a.~
            \lambda b.~
            \lambda w_1.~
            \lambda w_2.~
            \lambda p.~ w_1~(\lambda x.~(w_2~x)~p) \quad \}
\end{array}
\]
The other fields are not defined, since $\mPure$ is handled
primitively in the \emf dynamic semantics (\autoref{sec:emf-dynamic-semantics}).

The well-formedness conditions on the signature $S$ (shown in the
auxiliary material) check that the fields in definitions $D$
and the lifts in $L$ are well-typed as per their corresponding WPs.
In addition, each effect definition can make use of the previously
defined effects, enabling a form of layering. However, in this paper,
we mainly focus on combining effects using the lift operations.






\begin{figure}

\[\setnamespace{0pt}\setpremisesend{5pt}
\nqquad\begin{array}{c}















\inferrule*[lab=T-Return]
{
  S; \Gamma |- e : \mTot~t
}
{
  S; \Gamma |- \mreturn~t~e : M~t~(\mretWP~t~e)
}

\qquad

\inferrule*[lab=T-Refine]
{
  S; \Gamma |- t : Type_i \\\\
  S; \Gamma, x@t |- \phi : Type_j
}
{
  S; \Gamma |- \refine{x@t}{\phi} : Type_i
}

\\\\

\inferrule*[lab=T-Bind] 
{
  S; \Gamma |- t_2 : \typez\qquad
  S; \Gamma |- \wp_2 : \product{x@t_1}{S.M.\ii{wp\_type}~t_2} \\
  S; \Gamma |- e_1 : M~t_1~\wp_1 \quad
  S; \Gamma, x@t_1 |- e_2 : M~t_2~(\wp_2~x)
}
{
  S; \Gamma |- \mbind~t_1~t_2~\wp_1~e_1~\wp_2~{x.e_2}
  : M~t_2~(\mbindWP~t_1~t_2~\wp_1~\wp_2)
}

\\\\

\inferrule*[lab={\hspace{0.4cm}T-Lift}]
{
  S; \Gamma |- e : M~t~\wp
}
{
  S; \Gamma |- \mlift{M'}~t~\wp~e : M'~t~(S.\mliftWP~\wp)
}

\hspace{0.1cm}

\inferrule*[lab=T-Act]
{
  S.F.\cps{\ii{act}} = \product{\overline{x@t}}{c} \\\\
  \forall i.~S; \Gamma |- e_i : t_i
}
{
  S; \Gamma |- \fact~\bar{e} : c[\bar{e}/\bar{x}]
}

\\\\

\inferrule*[lab={\hspace{0.2cm}T-Reify}]
{
  S; \Gamma |- e : F~t~\wp
}
{
  S; \Gamma |- \reify~e : \mTot~(S.F.\repr~t~\wp)
}

\quad

\inferrule*[lab=T-Reflect]
{
  S; \Gamma |- e : \mTot~(S.F.\repr~t~\wp)
}
{
  S; \Gamma |- \freflect~e : F~t~\wp
}




\\\\

\inferrule*[lab=T-Run]
{
  S; \Gamma |- e : \mPure~t~\wp \quad
  S; \Gamma |= \exists p. \wp~p
}
{
  S; \Gamma |- \run~e : \mTot~t
}

\hspace{0.1cm}

\inferrule*[lab=T-Sub]
{
  S; \Gamma |- e : c' \quad
  S; \Gamma |- c' <: c
}
{
  S; \Gamma |- e : c
}

\\\\

\inferrule*[lab=C-Pure]
{
  S; \Gamma |- t : \typez\\\\
  S; \Gamma |- \wp : \product{(\product{t}{\typez})}{\typez}
}
{
  S; \Gamma |- \kw{Pure}~t~\wp : \typez
}

\qquad

\inferrule*[lab=C-F]
{
  S; \Gamma |- S.F.\repr~t~\wp : \typez
}
{
  S; \Gamma |- F~t~\wp : \typez
}

\end{array}
\]
\caption{Selected typing rules for \emf}
\label{fig:emf-typing}
\end{figure}

\begin{figure}

\[\setnamespace{0pt}\setpremisesend{5pt}
\nqquad\begin{array}{c}

\inferrule*[lab=S-Tot]
{
  S; \Gamma |- t' <: t
}
{
  S; \Gamma |- \mTot~t' <: \mTot~t
}

\hspace{0.3cm}

\inferrule*[lab=S-Pure]
{
  S; \Gamma |- t' <: t \hspace{0.2cm}
  S; \Gamma |= \forall p. \wp~p \Rightarrow \wp'~p
}
{
  S; \Gamma |- \mPure~t'~\wp' <: \mPure~t~\wp
}

\\\\

\inferrule*[lab=S-F]
{
  S; \Gamma |- S.F.\repr~t'~\wp' <: S.F.\repr~t~\wp
}
{
  S; \Gamma |- F~t'~\wp' <: F~t~\wp
}

\hspace{0.4cm}

\inferrule*[lab=S-Prod]
{
  S; \Gamma |- t <: t' \\\\
  S; \Gamma, x:t |- c' <: c
}
{
  S; \Gamma |- \product{x@t'}{c'} <: \product{x@t}{c}
}

\\\\

\inferrule*[lab=S-RefineL]
{
}
{
  S; \Gamma |- \refine{x@t}{\phi} <: t
}

\hspace{0.1cm}

\inferrule*[lab=S-RefineR]
{
  S; \Gamma, x:t |= \phi
}
{
  S; \Gamma |- t <: \refine{x@t}{\phi}
}

\hspace{0.1cm}

\inferrule*[lab=S-Conv]
{
  S |- t' \steps^{\ast} t~\vee~S |- t \steps^{\ast} t'
}
{
  S; \Gamma |- t' <: t
}

\end{array}
\]
\caption{Selected subtyping rules for \emf}
\label{fig:emf-subtyping}
\end{figure}

\subsection{Static Semantics}



The expression typing judgment in \emf has the form $S; \Gamma |- e : c$,
where $\Gamma$ is the list of bindings $x:t$ as usual. Selected rules for
the judgment are shown in \autoref{fig:emf-typing}. In the rules, we sometimes
write $S; \Gamma |- e : t$ as an abbreviation for $S; \Gamma |- e : \mTot~t$.


\paragraph{Monadic returns, binds, lifts, and actions.}
Rules~{\sc{T-Return}}, {\sc{T-Bind}}, and {\sc{T-Lift}} simply
use the corresponding $\wp$ specification from the signature for $M$
to compute the final $\wp$. For example, in the case of the \ls|ST| monad
from \autoref{sec:ex-cps}, $S.ST.\cps{\ii{return}}~t
= \lam{x@t}{\lam{s_0@s}{\lam{post}{post~(x, s_0)}}}$.
Rule~{\sc{T-Act}} is similar; it looks up the type of the action from
the signature, and then behaves like the standard function application
rule.

\paragraph{Monadic reflection and reification.}
Rules~{\sc{T-Reify}} and {\sc{T-Reflect}} are dual, coercing between a
computation type and its underlying pure representation.
Rule~{\sc{T-Run}} coerces $e$ from type $\mPure~t~\wp$ to
$\mTot~t$. However, since the $\mTot$ type is unconditionally total,
the second premise of the rule checks that the $\wp$ is satisfiable.

\paragraph{Refinements, computations types, and proof irrelevance.}
\emf's refinement and computation types include a form of proof irrelevance.
In {\sc{T-Refine}}, the universe of $\refine{x@t}{\phi}$ is
determined by the universe of $t$ alone, since a witness for the
proposition $\phi$ is never materialized.
Refinement formulas $\phi$ and $\wp$s are manipulated using an
entailment relation, $S;\Gamma |= \phi$, for a proof-irrelevant,
classical logic where all the connectives are
``squashed''~\citep{Nogin02}, e.g.,
\ls$p /\ q$ and \ls$p ==> q$ from \S\ref{sec:examples}, are encoded
as \ls$x:unit{p * q}$ and \ls$x:unit{p -> q}$, and reside in $\typez$.
Similar to {\sc{T-Refine}}, in {\sc{C-Pure}}, the universe of a
computation type is determined only by the result type. Since the
$\wp$ is proof irrelevant, the use of $\typez$ in the type of $\wp$ is
quite natural, because its proof content is always squashed. For
user-defined monads $F$, the rule {\sc{C-F}} delegates to their
underlying representation $S.F.\repr$.


\paragraph{Subsumption and subtyping judgment.}
{\sc{T-Sub}} is a subsumption rule for computations, which makes use
of the two judgments $S; \Gamma |- c <: c'$ and
$S; \Gamma |- t <: t'$, shown (selectively) in \autoref{fig:emf-subtyping}.
Rule~{\sc{S-Pure}} checks that $t' <: t$, and
makes use of the $S;\Gamma |= \phi$ relation to check that
$\wp$ is stronger than $\wp'$, i.e. for all postconditions, the
precondition computed by $\wp$ implies the precondition computed
by $\wp'$.

Similar to {\sc{C-F}}, the rule {\sc{S-F}} delegates the check to the
underlying representation of $F$. Rule~{\sc{S-Prod}} is
the standard dependent function subtyping. Rule~{\sc{S-RefineL}} permits
dropping the refinement from the subtype, and rule~{\sc{S-RefineR}} allows
subtyping to a refinement type, if we can prove the formula $\phi$ for
an arbitrary $x$. 
Finally, rule~{\sc{S-Conv}} states
that the beta-convertible types are subtypes of each other ($S |- t \steps t'$
is the small-step evaluation judgment, introduced in the next section).



\begin{figure}

\[\setnamespace{0pt}\setpremisesend{5pt}
\nqquad\begin{array}{c}

\inferrule*[lab=R-App]
{
}
{
  S |- (\lam{x@t}{e})~e' \steps e[e'/x]
}

\hspace{0.3cm}



\inferrule*[lab=R-Run]
{
}
{
  S |- \run~(\purereturn~t~e) \steps e
}

\\\\

\inferrule*[lab=R-PureBind]
{
}
{
  S |- \purebinddmf~t_1~t_2~\wp_1~ (\purereturn~t~e_1)~\wp_2~{x.e_2} \steps e_2[e_1/x]
}

\\\\

\inferrule*[lab=R-ReifyRet]
{
}
{
  S |- \reify~(\freturn~t~e) \steps \fretE~t~e
}

\inferrule*[lab=R-ReifyReflect]
{
}
{
  S |- \reify~(\freflect~e) \steps e
}

\\\\

\inferrule*[lab=R-ReifyBind]
{
  e' = \fbindE~t_1~t_2~\wp_1~(\reify~e_1)~\wp_2~{x.(\reify~e_2)}
}
{
  S |- \reify~(\fbind~t_1~t_2~\wp_1~e_1~\wp_2~{x.e_2}) \steps e'
}

\\\\

\inferrule*[lab=R-ReifyAct]
{
}
{
  S |- \reify (\fact~\bar{e}) \steps S.F.\un{\ii{act}}~\bar{e}
}

\\\\

\inferrule*[lab=R-ReifyLift]
{
}
{
  S |- \reify (\mlift{M'}~t~\wp~e) \steps S.\mliftE~t~\wp~(\reify~e)
}

\end{array}
\]
\caption{Dynamic semantics of \emf (selected reduction rules)}
\label{fig:emf-dynamic-semantics}
\end{figure}

\subsection{\emf Dynamic Semantics}
\label{sec:emf-dynamic-semantics}

We now turn to the dynamic semantics of \emf, which is formalized as 
a strong small-step reduction relation.
Evaluation context are defined as
follows:
\[
\begin{array}{lcl}
E & ::=  & \bullet \mid \lam{x@t}{E} \mid E~e \mid e~E \mid \run~{E} \mid \reify~E \mid \freflect~E\\
  & \mid & \mbind~t_1~t_2~\wp_1~E~\wp_2~{x.e_2} \mid \mreturn~t~{E}\\
  & \mid & \mlift{M'}~t~\wp~E \mid \fact~{\bar{e}~E~\bar{e}'} \mid \casefull{E}{\_}{x.e_1}{x.e_2}{t}\\
  & \mid &  \casefull{e}{\_}{x.E_1}{x.e_2}{t} \mid \casefull{e}{\_}{x.e_1}{x.E_2}{t}
\end{array}
\]

The judgment has the form $S |- e \steps e'$. We show some selected rules
in \autoref{fig:emf-dynamic-semantics}. The main ideas of the judgment
are: (a) the $\mTot$ terms reduce primitively in using a strong
reduction semantics, (b) $\purebinddmf$ is also given a primitive
semantics, however (c) to $\beta$-reduce other monadic operations
(binds, returns, actions, and lifts), they need to be reified first,
which then makes progress using their underlying implementation in the
signature.

\paragraph{Order of evaluation.} Since the effectful terms reduce via
reification, the semantics does not impose any evaluation order on the
effects---reification yields $\mTot$ terms ({\sc{T-Reify}}), that
reduce using the strong reduction semantics. However, the more
familiar sequencing semantics of effects can be recovered by
controlled uses of reify that do not break the abstraction of effects
arbitrarily. Indeed, we formalize this notion in
Section~\ref{sec:stateful}, and prove that by sequencing the effects
as usual using \kw{bind}, and then reifying and reducing the entire
effectful term, one gets the expected strict evaluation semantics
(Theorem~\ref{thm:simulation}).

\paragraph{Semantics for $\mPure$ terms.} Rule~{\sc{R-PureBind}} reduces
similarly to the usual $\beta$-reduction. For $\run~e$, the semantics first evaluates $e$
to $\purereturn~t~e'$, and then $\run$ removes the $\purereturn$ and steps
to the underlying total computation $e'$ via {\sc{R-Run}}.

\paragraph{Semantics for monadic returns and binds.} Rule~{\sc{R-ReifyBind}} looks
up the underlying
implementation $\fbindE$
in the signature, and applies it to $e_1$ and $e_2$ but
after reifying them so that their effects are handled properly.
In a similar manner, rule~{\sc{R-ReifyRet}} looks up the underlying
implementation $\fretE$ and applies it to $e$. Note that in this case,
we don't need to reify $e$ (as we did in bind), because $e$ is already
a $\mTot$ term.

\paragraph{Semantics for monadic lifts and actions.} Rules {\sc{R-ReifyAct}} and
{\sc{R-ReifyLift}} also lookup the underlying implementations of the lifts and
actions in the signature and use them. Rule {\sc{R-ReifyLift}} in addition
reifies the computation $e$. For lifts, the arguments $\bar{e}$ are already
$\mTot$.

\subsection{\emf Metatheory}

We prove several metatheoretical results for \emf. First, we prove
strong normalization for \emf via a translation to
the calculus of inductive constructions (CiC) \cite{cic}.

\begin{theorem}[Strong normalization]
If $S; \Gamma |- e : c$ and CiC is strongly normalizing, then $e$ is
strongly normalizing.
\begin{proof}(sketch)
The proof proceeds by defining a translation from \emf to CiC, erasing
refinements and WPs, inlining the pure implementations of each monad,
and removing the reify and reflect operators. We show that this
translation is a type-preserving, forward simulation. If CiC is strongly
normalizing, then \emf must also be, since otherwise an infinite
reduction sequence in \emf could not be matched by CiC, contradicting
the forward simulation.
\end{proof}
\end{theorem}

\begin{theorem}[Subject Reduction]
If~$S; \Gamma |- e : c$ and $S |- e \steps e'$, then $S; \Gamma |- e' : c$.
\end{theorem}

This allows us to derive a total correctness property for
the \mPure~monad saying that \run-ing a \mPure~computation produces a
value which satisfies all the postconditions that are consistent with
the \wp~of the \mPure~computation.

\begin{corollary}[Total Correctness of $\mPure$]\label{thm:pure-total-correctness}
If~$S; \cdot |- e : \mPure~t~\wp$, then $\forall p.~S; \cdot |- p : t \rightarrow \typez$ and
$S; \cdot |= \wp~p$, we have $S |- \run~e \steps^{\ast} v$ such that $S; \cdot |= p~v$.
\end{corollary}

For the user-defined monads $F$, we can derive their total correctness
property by appealing to the total correctness of
the \mPure~monad. For instance, for the \ls|ST| monad from
\autoref{sec:ex-cps}, we can derive the following corollary simply by
using the typing of \reify~and \autoref{thm:pure-total-correctness}.

\begin{corollary}[Total Correctness of $ST$]
If~$S; \cdot |- e : ST~t~wp$, then $\forall p, s_0.~S; \cdot |- s_0 :
s,~S; \cdot |- p : t \times s \rightarrow \typez$
and $S; \cdot |= wp~s_0~p$, then $S |- \run~((\reify~e)~s_0) \steps^{\ast} v$ such that
$S; \cdot |= p~v$.
\end{corollary}

\subsection[Implementation in Fstar]{Implementation in \fstar{}}
\label{sec:emf-impl}

The implementation of \fstar was relatively easy to adapt to \emf. In
fact, \emf and \deflang and the translation between them were designed
to match \fstar's existing type system, as much as possible. We
describe the main changes that were made.

\paragraph*{User-defined non-primitive effects} are, of course,
the main new feature. Effect configurations closely match the $D$ form
from~\autoref{fig:emf-syntax}, the main delta being that non-primitive
effects include pure implementations or $M.\un{\ii{bind}}$,
$M.\un{\ii{return}}$, $M.\un{\ii{lift}}_{M'}$ etc.

\paragraph*{Handling \ls$reify$ and \ls$reflect$} in the type-checker
involved implementing the two relatively simple rules for them 
in~\autoref{fig:emf-typing}. A more significant change was made
to \fstar's normalization machinery, extending it to support rules
that trigger evaluation for reified, effectful programs. In contrast,
before our changes, \fstar would never reduce effectful terms. The
change to the normalizer is exploited by \fstar's encoding of proof
obligations to an SMT solver---it now encodes the semantics of
effectful terms to the solver, after using the normalizer to partially
evaluate a reified effectful term to its pure form.

\ch{Commented this out until we implement extraction
\paragraph*{Running programs with user-defined effects} is achieved by
extracting it to OCaml, as is usual for \fstar, except we now inline
the definitions of the underlying pure terms.\ch{This seems misleading,
  we haven't actually implemented extraction of non-primitive effects
  yet: \url{https://github.com/FStarLang/FStar/issues/753}}
Effects marked as primitive are extracted as before while making use
of primitive effects in OCaml---this is formally justified
in~\autoref{sec:stateful}.
}

\ifcheckpagebudget\clearpage\fi
\section{Dijkstra Monads for Free}
\label{sec:dm4f}


This section formally presents \deflang, a language for defining
effects by giving monads with their actions and lifts between them.
Via a pair of translations, we export such definitions to \emf as
effect configurations. The first translation of a term $e$, a CPS,
written $\cps{e}$ produces a predicate-transformer from \deflang term;
the second one is an {\em elaboration}, $\un{e}$, which produces
an \emf implementation of a \deflang term.
The main result shows that for any \deflang term the result of
the \cpst is in a suitable logical relation to the
elaboration of the term, and thus a valid specification for this
elaboration.
We also show that the \cpst always produces monotonic and
conjunctive predicates, properties that should always hold for WPs.
Finally, we show that the \cpst preserves all equalities in
\deflang, and thus translates \deflang monads into \emf Dijkstra
monads.



\subsection{Source: \deflang Effect Definition Language}
\label{sec:deflang}


The source language \deflang is a simply-typed lambda calculus
augmented with an abstract monad \teff, as in \autoref{sec:ex-cps}.
The language is essentially that of \citet{Filinski94} with certain restrictions
on allowed types to ensure the correctness of elaboration.

%
There are two effect symbols: $\neu$ (non-effectful) and $\teff$.
The typing judgment is split accordingly, and $\epsilon$ ranges over both of them.
%
Every monadic term needs to be bound via \textbf{bind$_{\teff}$} to be
used.\footnote{
In this formalization, \textbf{bind} and \textbf{return} appear explicitly in
source programs. When using our implementation, however, the user need not call
\textbf{bind} and \textbf{return}; rather, they write programs in a direct
style, and \textbf{let}-bindings are turned into \textbf{bind}s as needed.
\autoref{sec:dm4f-impl} provides some details on the interpretation and
elaboration of concrete \fstar terms as \deflang terms.}
Functions can only take non-effectful terms as arguments, but may return
a monadic result.

%


The set of \deflang types is divided into \emph{$A$ types}, \emph{$H$
  types}, and \emph{$C$ types}, ranged over by $A$, $C$, and $H$,
respectively. They are given by the grammar:
\[\begin{array}{lll}
    \vphantom{\tarr\narr}
    A & ::= & X
                \mid b
                \mid A \narr A
                \mid A + A
                \mid A \times A \\
    \vphantom{\tarr\narr}
    H & ::= & A \mid C \\
    \vphantom{\tarr\narr}
    C & ::= & H \tarr A
                \mid H \narr C
                \mid C \times C \\
\end{array}\]

Here $X$ ranges over type variables (needed to define monads) and $b$ are base types. The $\teff$-arrows
represent functions with a monadic result, and our translations will
provide WPs for these arrows.
$A$ types are referred to as ``$\teff$-free'', since they contain no monadic
operations. $C$ types are inherently computational in the sense that
they cannot be eliminated into an $A$ type: every possible elimination
will lead to a monadic term. They are referred to as ``computational
types''. $H$ types are the union of both, and are called ``hypothesis''
types, as they represent the types of possible functional arguments.
%
%
As an example, the state monad is represented as the type $\smash{S
\tarr (X \times S)}$, where $X$ is a type variable and $S$ is some type
representing the state. We will exemplify our main results for terms of
this type, thus covering every stateful computation definable in \deflang.

\deflang types do not include ``mixed'' $A
\times C$ pairs, computational sums $C + H$, functions of type $\smash{C \narr
A}$, or types with right-nested $\teff$-arrows. We do allow nesting
$\teff$-arrows to the left, providing the generality needed for the
continuation monad, and others.  These restrictions are crafted to
carefully match \emf. Without them, our translations, would generate
ill-typed or logically unrelated \emf terms, and these restrictions
do not appear to be severe in practice, as evidenced by the examples
in~\autoref{sec:examples}.
%
%
%
%
%

The syntax for terms is ($\kappa$ standing for constants):
\[
\begin{array}{lrl}
        e & ::=  & x \mid e~e \mid \lambda x@H.~e \mid \kappa(e,\ldots,e) \\
          & \mid & (e,e) \mid \fst{e} \mid \snd{e} \\
          & \mid & \inl{e} \mid \inr{e} \mid \mycases{e}{x@A.~e}{y@A.~e} \\
          & \mid & \returnT{e} \mid \bindT{e}{x}{e} \\
\end{array}
\]
Typing judgments have the forms $\DmG \vdash e : H \bang \neu$ and
$\DmG \vdash e : A \bang \teff$, where $\Delta$ is a finite sequence of type variables and
$\Gamma$ is a normal typing context, whose types only use type variables from $\Delta$. Here are some example rules:
\[\small
    \infer{\DmG \vdash \lambda x@H.~e : H \earr H' ! \neu}
          {\DmG, x@H \vdash e : H' ! \epsilon\vphantom{\earr}}
\quad
    \infer{\DmG \vdash fe : H' ! \epsilon \vphantom{\earr}}
          {\DmG \vdash f : H \earr H' ! \neu & \DmG \vdash e : H ! \neu}
\]
\[\small
    \infer{\DmG \vdash \returnT{e}: A ! \teff}
          {\DmG \vdash e : A ! \neu}
\qquad
    \infer{\DmG \vdash \bindT{e_1}{x}{e_2} : A' ! \teff}
          {\DmG \vdash e_1 : A ! \teff & \DmG, x:A \vdash e_2 : A' ! \teff}
\]
In these rules we implicitly assume that all appearing types are
well-formed with respect to the grammar, e.g., one cannot form a
function of type $\smash{C \narr A}$ by the abstraction rule.

As an example, 
$\mathrm{return}_\mathrm{ST} =
\lambda x@X.~\lambda s@S.~\returnT{(x,s)}$ has type $\smash{X \narr S \tarr (X
\times S)}$, using these rules. 

When defining effects and actions, one deals (at a top level) with
non-effectful $C$ types ($C\bang\neu$).

\subsection[The CPS Translation]{The \cpst}
\label{sec:cps}

The essence of the \cpst is to translate $\returnT{e}$ and
$\bindT{e_1}{x}{e_2}$ to the returns and binds of the continuation
monad.
%
%
%
%
We begin by defining a translation $\cps{H}$, that translates any $H$ type to
the type of its predicates by CPS'ing the $\teff$-arrows.
First, for any $\tau$-free type $A$, $\cps{A}$ is essentially the
identity, except we replace every arrow $\smash{\narr}$ by a
$\rightarrow$. Then, for computation types, we define:
\[
\begin{array}{lll@{\hspace{5em}}lll}
    \vphantom{\narr\tarr}   \cps{(H \narr C)}  &=& \cps{H} -> \cps{C} \\
    \vphantom{\narr\tarr}   \cps{(C \times C')} &=& \cps{C} \times \cps{C'} \\
    \cps{(H \tarr A)}   &=& \multicolumn{4}{l}{\cps{H} -> (\cps{A} -> \typez) -> \typez} \\
\end{array}
\]
%
%
Note that all arrows on the right hand side have a
\ls$Tot$ codomain, as per our notational convention.

In essence, the codomains of $\teff$-arrows are CPS'd into a WP, which
takes as argument a predicate on the result and produces a predicate
representing the ``precondition''. All other constructs are just
translated recursively: the real work is for the $\teff$-arrows.

For example, for the state monad $\smash{S \tarr (X \times S)}$, the \cpst{}
produces the \emf type $S -> (X \times S -> \typez) -> \typez$. It is
the type of predicates that map an initial state and a postcondition (on
both result and state) into a proposition. Modulo isomorphism (of the
order of the arguments and currying)\footnotemark{} this is exactly the
type of WPs in current \fstar's state monad (cf. \autoref{sec:intro}, \autoref{sec:ex-cps}).
\footnotetext{One can tweak our translation to generate WPs that have
  the usual postcondition to precondition shape. However we found
  the current shape to be generally easier to work with.}


The two main cases for the \cpst for well-typed \deflang terms are
shown below; every other case is simply a homomorphic application of
$\star$ on the sub-terms.
\[\nqquad\begin{array}{lcl}
    \cps{(\returnT{e})} \nqquad&\!\!=\!\!&\nqquad \lambda p@(\cps{A} -> \typez).~p~\cps{e} ~~
                        \mbox{\tiny{when}}~\DmG \vdash e : A!\neu \\
    \cps{(\bindT{e_1}{x}{e_2})} \nqquad&\!\!=\!\!&\nqquad \lambda p@(\cps{A'}->\typez).~\cps{e_1}~(\lambda x@A.~\cps{e_2}~p) \\
                       \multicolumn{3}{r}{\hspace{-3cm}\mbox{{\tiny{when}} $\DmG, x:A \vdash e_2 : A'!\teff$}}\\
\end{array}\]
Formally, the \cpst and elaboration are defined over a typing
derivation, as one needs more information than what is present in the
term.
%
%
%
%
The \cpsts of terms and types are related in the following sense, where
we define the environments $\un\Delta$ as $X_1 : \typez, \ldots, X_n
: \typez$ when $\Delta = X_1,\ldots,X_n$ ; and $\cps{\Gamma}$ as $x_1
: \cps{t_1}, \ldots x_n : \cps{t_n}$ when $\Gamma = x_1 : t_1, \ldots
x_n : t_n$ (we assume that variables and type variables are also \emf
variables).
%

\begin{theorem}[well-typing of \cpst]~\\
\label{thm:cps-typing}
$\Delta \mid \Gamma \vdash e : C\bang\neu$ implies
$\un\Delta, \cps{\Gamma} \vdash \cps{e} : \cps{C}$.
\end{theorem}

After translating a closed term $e$, one can abstract over the variables
in $\un\Delta$ to introduce the needed polymorphism in \emf. This will
also be the case for elaboration.

As an example, for the previous definition of
$\mathrm{return}_\mathrm{ST}$ we get the translation $\lambda
x{:}X.~\lambda s{:}S.~\lambda p{:}(X \times S -> \typez).~ p (x,s)$, which
has the required transformer type: $X -> S -> (X \times S -> \typez) ->
\typez$ (both with $X$ as a free type variable). It is
what one would expect: to prove a postcondition $p$ about
the result of running $\mathrm{return}_\mathrm{ST}~x$, one needs to
prove $p(x,s)$ where $s$ is the initial state.

\subsection{Elaboration}
\label{sec:elab}

\newcommand{\mc}[2]{\multicolumn{#1}{l}{#2}}

\begin{figure*}[!t]
\[\begin{array}{llclllcl}
(1) & \un{x} & = & x &
(5) & \un{\fst{e}} & = & \mfst{\un{e}} \\
(2) & \un{\kappa(e_1,\ldots,e_n)} & = &  \kappa~\un{e_1}~\ldots~\un{e_n}  &
(6) & \un{\snd{e}} & = & \msnd{\un{e}} \\
(3) & \un{\lambda x@A.\, e} & = &  \lambda x@\un{A}.\, \un{e} &
(7) & \un{\inl{e}} & = & \minl{\un{e}} \\
(4) & \un{\lambda x@C.\, e} & = &  \lambda x^w@\cps{C}.\, \lambda x@\F{C}{x^w}.\, \un{e} &
(8) & \un{\inr{e}} & = & \minr{\un{e}}\\

(9) & \un{e_1e_2} & = &  \mc{3}{\un{e_1}~\un{e_2}} & \mc{2}{(\Delta\mid\Gamma \vdash  e_2: A \bang \neu)}\\

(10) & \un{e_1e_2} & = &  \mc{3}{\un{e_1}~(\cps{e_2}~\sG)~\un{e_2}} & \mc{2}{(\Delta\mid\Gamma \vdash  e_2: C \bang \neu)}\\

(11) & \un{(e_1,e_2)} & = & (\un{e_1},\un{e_2}) \\

(12) & \un{\mycases{e} {x@A_1.\, e_1}  {y@ A_2.\, e_2} } & = &
        \mc{3}{\caseimp{\un{e}}
            {x.\un{e_1}}{y.\un{e_2}}}
                & \mc{2}{(\Delta\mid\Gamma,x@A_1 \vdash e_1 : A\bang\varepsilon)}\\

(13) & \un{\mycases{e} {x@A_1.\, e_1}  {y@ A_2.\, e_2} } & = &
        \mc{3}{\casefull{\un{e}}{z}
            {x.\un{e_1}} {y.\un{e_2}}
            {\F{C}{\caseimp{z}{x.(\cps{e_1}~\sG)}{y.(\cps{e_2}~\sG)}}}}
                & \mc{2}{(\Delta\mid\Gamma,x@A_1 \vdash e_1 : C\bang\neu)}\\

(14) & \un{\returnT{e}} & = & \mc{3}{\pureret{\un{A}}{\un{e}}} & \mc{2}{(\Delta\mid\Gamma \vdash e: A\bang\teff)} \\

(15) & \un{\bindT{e_1}{x@A}{e_2}} & = & \mc{3}{\purebind{\un{A}}{\un{A'}}{(\cps{e_1}~\sG)}{\un{e_1}}{(\lambda x@\cps{A}.\, \cps{e_2}~\sG)}{x.\un{e_2}}}
            & \mc{2}{(\Delta\mid\Gamma,x:A \vdash e_2: A'\bang\teff)}\\
\end{array}\]
\caption{The elaboration of \deflang terms to \emf}
\label{fig:elab}
\end{figure*}

Elaboration is merely a massaging of the source term to make it
properly typed in \emf.
During elaboration, monadic operations are translated to those of the
identity monad in \emf, namely \ls$Pure$.

\paragraph*{Elaboration of types}
We define two elaboration translations for \deflang types, which
produce the \emf types of the elaborated expression-level terms.
The first translation $\un{A}$ maps an $A$ type to a simple \emf type, while
the second one $\F{C}{\wp}$ maps a $C$ type and a {\em specification}
$\wp$ of type $\cps{C}$ into an \emf computation type containing
\ls$Tot$ and \ls$Pure$ arrows.
The $\un{A}$ translation is the same as the CPS one, \IE $\un{A} = \cps{A}$.
%
%
%

%
The $\F{C}{\wp}$ (where $\wp : \cps{C}$) translation is defined by:
\[
\begin{array}{llll}
(1) & \vphantom{\earr \F{\earr}{}}    \F{C \times C'}{\wp} &\eqdef& \F{C}{(\mfst{\wp})} \times \F{C'}{(\msnd{\wp})} \\
(2) & \vphantom{\earr \F{\earr}{}}    \F{C \earr H}{\wp}   &\eqdef& \product{w'@\cps{C}}{\F{C}{w'} -> \G{\epsilon}{H}{\wp~w'}} \\
(3) & \vphantom{\earr \F{\earr}{}}    \F{A \earr H}{\wp}   &\eqdef& \product{x@\un{A}}{\G{\epsilon}{H}{\wp~x}} \\
\end{array}
\]
%
%
Here we define $\G{\neu}{C}{\wp} = \F{C}{\wp}$ and
$\G{\teff}{A}{\wp} = \kw{Pure}~\un{A}~\wp$.
%

The main idea is that if an \emf term $e$ has type $\F{C}{\wp}$, then $\wp$
is a proper specification of the final result. Putting pairs aside
for a moment, this means that if one applies enough arguments
$e_i$ to $e$ in order to eliminate it into a \ls$Pure$ computation,
then $e~\bar{e_i} : \Pure{A}{\wp~\bar{s_i}}$, where each $s_i$ is the
specification for each $e_i$. This naturally extends to pairs,
for which the specification is a pair of proper
specifications, as shown by case (1) above.

In case (2), the $w':\cps{C}$ arguments introduced by F are relevant for the
higher-order cases, and serve the following purpose, as illustrated
in \S\ref{sec:ex-cps} (for the translation of \ls$bind$ for
the \ls$ST$ monad) and
\S\ref{sec:ex-cont} (for the continuation monad): when taking
computations as arguments, we first require their specification in
order to be able to reason about them at the type level.
Taking these specification arguments is also the only way for being
WP-polymorphic in \emf. Note that, according to the dependencies, only the
$\cps{C}$ argument is used in the specifications, while we shall
see in the elaboration of terms
that only the $\F{C}{\wp}$ argument is used in terms.
When elaborating terms, we pass this specification as an extra
argument where needed.

In case (3), when elaborating functions taking an argument of $A$ type there is
no need to take a specification, since the argument is completely
non-effectful and can be used at both the expression and the type levels.
Informally, a non-effectful term is its own specification.

Returning to our state monad example, the result of $\F{S \tarr (X
\times S)}{\wp}$ is $s@S -> \Pure{(X \times S)}{\wp~s}$, i.e., the
type of a function $f$ such that for any postcondition $p$ and
states $s$ for which one can prove the precondition $\wp~s~p$, we
have that $f~s$ satisfies $p$.

\paragraph*{Elaboration of terms} is defined in \autoref{fig:elab}
and is, as expected, mostly determined by the translation of
types. The translation is formally defined over typing derivations,
however, for brevity, we present each translation rule simply on
the terms, with the important side-conditions we rely on from the
derivation shown in parenthesis. We describe only the most interesting
cases.

\paragraph*{Computational abstractions and applications (cases 4 and 10)}
Case (4) translates a function with a computational argument $x@C$ to
a function that expects two arguments, a specification $x^w@\cps{C}$ and $x$
itself, related to $x^w$ at a suitably translated type. We track the
association between $x$ and $x^w$ using a substitution $\sG$, which
maps every computational hypothesis $x:C$ in $\Gamma$ to $x^w$ (of
type $\cps{C}$) in $\un\Gamma$, In case (10), when passing a
computation argument $e_2$, we need to eliminate the double
abstraction introduced in case (4), passing both $\cps{e_2}\ s_\Gamma$,
\emph{i.e.} the specification of $e_2$
where we substitute the free computation variables, and $\un{e_2}$ itself.

\paragraph*{Return and bind (cases 14 and 15)} The last two rules show
the translation of return and bind for $\tau$ to return and bind
for \ls$Pure$ in \emf. This is one of the key points: in the
elaboration, we interpret the $\tau$ as the identity monad in \emf,
whereas in the \cpst, we interpret $\tau$ as the continuation
monad. \autoref{thm:elab-typing}, our main theorem, shows that
\emf's WP computation in the \ls$Pure$ monad for $\un{e}$ produces
a WP that is logically related to the \cpst of $e$, i.e., WPs and
the CPS coincide formally, at arbitrary order.

\begin{theorem}[Logical relations lemma]
\label{thm:elab-typing}
    \[\begin{array}{llll}
    \mbox{1.} & \Delta \mid \Gamma \vdash e : C\bang\neu &
        ==> & \un\Delta, \un\Gamma \vdash \un{e} : \F{C}{(\cps{e} \sG)} \\

    \mbox{2.} & \Delta \mid \Gamma \vdash e : A\bang\teff &
        ==> & \un\Delta, \un\Gamma \vdash \un{e} : \Pure{\un{A}}{\cps{e} \sG}
    \end{array}\]
\end{theorem}

Where $\un\Gamma$ is defined by mapping any ``$x : A$'' binding in
$\Gamma$ to ``$x : \un{A}$'' and any ``$y : C$'' binding to ``$y^w :
\cps{C}, y : \F{C}{y^w}$''.
Instantiating (1) for an empty $\Gamma$, we get as corollary that
$\un\Delta \vdash \un{e} : \F{C}{\cps{e}}$, representing the fact that
$\cps{e}$ is a proper specification for $\un{e}$.
Following the \ls$ST$ monad example, this implies that for any source
term $e$ such that
$\smash{X \mid \cdot \vdash e : S \tarr (X \times S)}$
holds, then
$X : \typez \vdash \un{e} : \product{s_0@S}{\Pure{(X \times S)}{\cps{e}~s_0}}$,
will hold in \emf, as intuitively expected.



%


\subsection{Monotonicity and Conjunctivity}
\label{sec:mon-conj}

A key property of WPs is monotonicity: weaker postconditions should map
to weaker preconditions. This is also an important \fstar invariant
that allows for logical optimizations of WPs.
Similarly, WPs are
conjunctive: they distribute over conjunction and universal quantification
in the postcondition.
We show that any \emf term obtained from the \cpst{} is monotonic and
conjunctive, for higher-order generalizations of the usual definitions
of these properties \cite{Dijkstra97}.

\newcommand{\eqtwo}[0]{==}

We first introduce a hereditarily-defined relation between \emf terms
$t_1 \stg_t t_2$, read ``$t_1$ stronger than $t_2$ at type $t$''
and producing an \emf formula in $\typez$, by recursion on the
structure of $t$:
\[\small
\arraycolsep=1pt
\begin{array}{lll}
    x \stg_{\typez} y &\eqdef& x => y \\
    x \stg_{b} y      &\eqdef& x \eqtwo y \\
    x \stg_{X} y      &\eqdef& x \eqtwo y \\
    f \stg_{t_1 -> t_2} g  &\eqdef&
            \forall x,y:t_1.\,
                x \stg_{t_1} x \land
                x \stg_{t_1} y \land
                y \stg_{t_1} y =>
                f~x \stg_{t_2} g~y \\
    x \stg_{t_1 \times t_2} y &\eqdef& \mfst{x} \stg_{t_1} \mfst{y} \land \msnd{x} \stg_{t_2} \msnd{y} \\
    x \stg_{t_1 + t_2} y &\eqdef& (\exists v_1,v_2:t_1,\, x \eqtwo \minl{v_1} \land y \eqtwo \minl{v_2} \land v_1 \stg_{t_1} v_2)~\lor \\
                        &      & (\exists v_1,v_2:t_2,\, x \eqtwo \minr{v_1} \land y \eqtwo \minr{v_2} \land v_1 \stg_{t_2} v_2) \\
  \end{array}
\]
where $b$ represents any \emf base type (\IE a type constant in $\typez$)
 and $X$ any type variable\footnotemark.
The symbol $\eqtwo$ represents \emf's squashed propositional equality.
%
%
\footnotetext{We can get a stronger result if we don't restrict
the relation on type variables to equality and treat it abstractly instead.
For our purposes this is not needed as we plan to instantiate type variables
with predicate-free types.}
%
%
The $\stg$ relation is only defined for the subset of \emf types that are
all-$\mTot$ and non-dependent. All types resulting from the \cpst are
in this subset, so this not a limitation for our purposes.
A type $t$ in this subset is called \emph{predicate-free} when it does
not mention $\typez$.
For any predicate-free type $t$ the relation $\stg_t$ reduces
to extensional equality.

The $\stg$ relation is not reflexive.
We say that an \emf term $e$ of type $t$ is {\em monotonic} when $e
\stg_{t} e$.
Note that monotonicity is preserved by application.
For first-order WPs this coincides with the standard definitions,
and for higher-order predicates it gives a reasonable extension.
Since the relation reduces to equality on predicate-free types,
every term of such a type is trivially monotonic.
The reader can also check that every term of a type $t = d_1 ->
\cdots -> d_n -> \typez$ (where each $d_i$ is predicate-free) is monotonic;
it is only at higher-order that monotonicity becomes interesting.

For a first-order example, let's take the type of WPs for
programs in the \ls$ST$ monad: $S -> (X \times S -> \typez) -> \typez$,
making use of the previous simplification:
%
\[\arraycolsep=1pt
 \begin{array}{cl}
    & f \stg_{S -> (X \times S -> \typez) -> \typez} f \\
 \equiv & \forall s_1, s_2.\, s_1 = s_1 \land s_1 = s_2 \land s_2 = s_2 => f~s_1 \stg f~s_2\\
 \iff & \forall s.\, f~s \stg_{(X \times S -> \typez) -> \typez} f~s\\
 \equiv & \forall s, p_1, p_2.\,
                  p_1 \stg p_2 => f~s~p_1 \stg_{\typez} f~s~p_2 \\
 \iff & \forall s, p_1, p_2.\,
             (\forall x,s'.\, p_1~(x,s') => p_2~(x,s')) => (f~s~p_1 => f~s~p_2)
\end{array}
\]

%
This is exactly the usual notion of monotonicity for imperative
programs \cite{Dijkstra97}: ``if $p_2$ is weaker
than $p_1$, then $f~s~p_2$ is weaker than $f~s~p_1$ for any $s$''.

Now, for a higher-order example, consider the continuation monad in
\deflang: $\smash{\mCont~X = (X \tarr R) \tarr R}$, where $X$ is the type
variable and $R$ some other variable representing the end result of the
computation. The type of WPs for this type is
\[\mCont_{\ii{wp}}~X = (X {->} (R {->} \typez) {->} \typez) -> (R {->} \typez) {->} \typez\]
Modulo argument swapping, this maps a postcondition on $R$ to a precondition
on the specification of the continuation function.
The condition $\ii{wp} \stg_{\mCont_{\ii{wp}}~X} \ii{wp}$ reduces and simplifies to:
\[
\begin{array}{c}
       \ii{kw}_1 \stg \ii{kw}_1 \;\land\;
       \ii{kw}_1 \stg \ii{kw}_2 \;\land\;
       \ii{kw}_2 \stg \ii{kw}_2 \;\land\;
       p_1 \stg p_2 \\
\implies \ii{wp}~\ii{kw}_1~p_1 \implies \ii{wp}~\ii{kw}_2~p_2 \\
\end{array}
\]
for any $\ii{kw}_1, \ii{kw}_2, p_1, p_2$ of appropriate types. Intuitively,
this means that $\ii{wp}$ behaves monotonically on both arguments, but
requiring that the first one is monotonic.
In particular, this implies that for any monotonic $\ii{kw}$, $\ii{wp}~\ii{kw}$ is
monotonic at type $(R -> \typez) -> \typez$.

%
%
%

We proved that the \cpst of any well-typed source term $e :
C\bang\neu$ gives a monotonic $\cps{e}$ at the type $\cps{C}$. This
result is more general than it appears at a first glance: not only
does it mean that the WPs of \emph{any} defined return and bind
are monotonic, but also those of any action or function are. Also, lifts
between monads and other higher-level computations will preserve this
monotonicity. Furthermore, the relation $\vDash$ in the conclusion of
the theorem below is \emf's validity judgment, i.e., we show that these
properties are actually provable within \fstar without relying
on meta-level reasoning.

\begin{theorem}[Monotonicity of \cpst]~\\
For any $e$ and $C$, $\Delta \mid \cdot \vdash e :
C \bang \neu$ implies $\un\Delta \vDash \cps{e} \le_{\cps{C}} \cps{e}$.
\end{theorem}



We give a similar higher-order definition of conjunctivity, and prove
similar results ensuring the \cpst produces conjunctive WPs. The definition
for conjunctivity is given below, where $a$ describes the
predicate-free types (including variables).
\newcommand{\C}[2]{\mathbb{C}_{#1}(#2)}
\[\small
\arraycolsep=1pt
\begin{array}{lll}
    \C{\dneg{a}}{w}        &\eqdef& \forall p_1, p_2.\, w~p_1 \land w~p_2 = w~(\lambda x. p_1~x \land p_2~x)  \\
    \C{a}{x}               &\eqdef& \True \\
    \C{t_1 -> t_2}{f}      &\eqdef& \forall x:t_1.\, \C{t_1}{x} => \C{t_2}{fx} \\
    \C{t_1 \times t_2}{p}  &\eqdef& \C{t_1}{\mfst{p}} \land \C{t_2}{\msnd{p}} \\
  \end{array}
\]

Again, the relation is not defined on all types, but it does include
the image of the type-level \cpst, so it is enough for our purposes. This
relation is trivially preserved by application, which allows us to prove
the following theorem:

\begin{theorem}[Conjunctivity of \cpst]~\\
For any $e$ and $C$, $\Delta \mid \cdot \vdash e :
C \bang \neu$ implies $\un\Delta \vDash \C{\cps{C}}{\cps{e}}$
\end{theorem}

For the \ls$ST$ monad, this implies that for any $e$ such that
$\smash{e : S \tarr X \times S}$
we know, again within \emf, that
$\cps{e}~s~p_1 \land \cps{e}~s~p_2 = \cps{e}~s~(\lambda x. p_1~x \land p_2~x)$
for any $s, p_1, p_2$. This is the usual notion of conjunctivity for
WPs of this type.


\subsection[The CPS Translation Preserves Equality and Monad Laws]{The \cpst Preserves Equality and Monad Laws}
\label{sec:eqs}

We define an equality judgment on \deflang terms that is basically
$\beta\eta$-equivalence, augmented with the monad laws for the abstract
$\tau$ monad. We show that the \cpst{} preserves this equality.

\begin{theorem}[Preservation of equality by CPS]\label{eq-star}~\\
If $\Delta\mid\cdot \vdash e_1 = e_2 : H \bang\varepsilon$ then
$\un\Delta \vDash \cps{e_1} \eqtwo \cps{e_2}$ .
\end{theorem}

Since the monad laws are equalities themselves,
any source monad will be translated to a specification-level monad
of WPs. This also applies to lifts: source monad morphisms are mapped to
monad morphisms between Dijkstra monads.

\ifcheckpagebudget\clearpage\fi
\subsection[Implementing the Translations in Fstar]{Implementing
  the Translations in \fstar}
\label{sec:dm4f-impl}

We devised a prototype implementation of the two translations in \fstar{}.
Users define their monadic effects as \fstar{} terms in direct style,
as done in \autoref{sec:examples}, and these definitions get
automatically rewritten into \deflang.
As explained in \autoref{sec:examples}, instead of $\teff$-arrows
(\smash{$H \tarr A$}), we use a distinguished \fstar{} effect \ls$tau$ to
indicate where the CPS should occur.
The effect \ls$tau$ is defined to be an alias for \fstar's \ls$Tot$
effect, which allows the programmer to reason extrinsically about the
definitions and prove that they satisfy various properties within
\fstar, \EG the monad laws.
Once the definitions have been type-checked in \fstar, another
minimalist type-checker kicks in, which has a twofold role.
First, it ensures that the definitions indeed belong to \deflang, \EG
distinguishing $A$ types from $C$ types.
Second, it performs bidirectional inference to distinguish monadic
computations from pure computations, starting from top-level
annotations, and uses this type information to automatically introduce
$\textbf{return}_\teff$ and $\textbf{bind}_\teff$ as needed.
For instance, in the \ls$st$ example from \autoref{sec:ex-cps}, the
type-checker rewrites \ls$x, s0$ into $\returnT {(x, s_0)}$; and
\ls$let x, s1 = f s0 in ...$ into $\bindT{f\ s_0}{x, s_1}{\dots}$; and
\ls$g x s1$ into $\returnT {(g\ x\ s_1)}$.
The elaboration maps let-bindings in \deflang to let-bindings in \fstar; the
general inference mechanism in \fstar takes care of synthesizing the WPs,
meaning that the elaboration, really, is only concerned about extra arguments
for abstractions and applications.

Once the effect definition is rewritten to \deflang, our tool uses the
\cpst and elaboration to generate the WP transformers for the Dijkstra
monad, which previously would be written by hand.
Moreover,
several other WP combinators are derived from the WP type and
used internally by the \fstar{} type-checker; again previously these
had to be written by hand.


\ifcheckpagebudget\clearpage\fi
\section[EMFstar with Primitive State]{\emf with Primitive State}
\label{sec:stateful}

As we have seen in \autoref{sec:emf}, \emf encodes all its effects using pure
functions. However, one would like to be able to run \fstar programs
efficiently using primitively implemented effects.
In this section, we show how \emf's pure monads apply to \fstar's
existing compilation strategy, which provides primitive support for
state via compilation to OCaml, which, of course, has state
natively.\footnote{\fstar also compiles exceptions natively to OCaml,
however we focus only on state here, leaving a formalization of
primitive exceptions to the future---we expect it to be similar to the
development here.} The main theorem of \autoref{sec:relatingemfST} states that
well-typed \emf programs using the state monad abstractly (i.e., not
breaking the abstraction of the state monad with arbitrary uses
of \ls$reify$ and \ls$reflect$) are related by a simulation to \emfST
programs that execute with a primitive notion of state. This result
exposes a basic tension: although very useful for proofs, \ls$reify$
and \ls$reflect$ can break the abstractions needed for efficient
compilation. However, as noted in \autoref{sec:state-with-references}, 
this restriction on the use of \ls$reify$ and \ls$reflect$ only 
applies to the \emph{executable} part of a program---fragments of 
a program that are computationally irrelevant are 
erased by the \fstar compiler and are free to use these operators.


\subsection{\emfST: A Sub-Language of \emf with Primitive State}
\label{sec:emfST}

The syntax of \emfST corresponds to \emf, except, we configure it to
just use the \ls$ST$ monad. Other effects that may be added to \emf
can be expanded into their encodings in its
primitive \ls$Pure$ monad---as such, we think of \emfST as modeling a
compiler target for \emf programs, with \ls$ST$ implemented
primitively, and other arbitrary effects implemented purely.
We thus exclude \ls$reify$ and \ls$reflect$ from \emfST, 
also dropping type and WP arguments of return, bind and lift operators,
since these are no longer relevant here.

The operational semantics of \emfST is a small-step, call-by-value
reduction relation between pairs $(s, e)$ of a state $s$ and a term
$e$. The relation includes the pure reduction steps of \emf that simply
carry the state along (we only show ST-beta), and three primitive
reduction rules for \ls$ST$, shown below. The only irreducible \ls$ST$
computation is \ls$ST.return v$. Since the state is primitive in \emfST, the 
term $\kw{ST.bind}~e~x.e'$ reduces
without needing an enclosing \ls$reify$.
%
\[\begin{array}{ll}
(s, (\lam{x@t}{e}) v) \leadsto (s, e[v/x]) & \mbox{ST-beta} \\
(s, \kw{ST.bind}~(\kw{ST.return}~v)~x.e) \leadsto (s, e[v/x])  & \mbox{ST-bind} \\
(s, \kw{ST.get}~()) \leadsto (s, \kw{ST.return}~s)   & \mbox{ST-get} \\
(s, \kw{ST.put}~s') \leadsto (s', \kw{ST.return}~()) & \mbox{ST-put}
\end{array}\]

\subsection{Relating \emf to \emfST}
\label{sec:relatingemfST}

We relate \emf to \emfST by defining 
a (partial) translation
from the former to the latter, and show that one or more steps of reduction
in \emfST are matched by one or more steps in \emf.
This result guarantees that it is
sound to verify a program in \emf and execute it in \emfST: the
verification holds for all \emf reduction sequences, and \emfST
evaluation corresponds to one such reduction.

The main intuition behind our proof is that the reduction
of \ls$reflect$-free \emf programs maintains terms in a very specific
structure---a stateful redex (an \ls$ST$ computation wrapped in \ls$reify$) reduces in a context structured like a
telescope of binds, with the state threaded sequentially as the
telescope evolves.
We describe this invariant structure as an \emf context, $K$,
parameterized by a state $s$. In the definition, $\hat{E}$ is a
single-hole, \ls$reify$-and-\ls$reflect$-free \emf context, a
refinement of the evaluation contexts of~\autoref{sec:emf}, to be
filled by a \ls$reify$-and-\ls$reflect$ free \emf term, $f$.
Additionally, we separate the $\hat{E}$ contexts by their effect into several
sorts: $\hat{E}:\kw{Tot}$ and $\hat{E}:\kw{Pure}$ are contexts which when
filled by a suitably typed term produce in \emf
a \ls$Tot$ or \ls$Pure$ term, respectively; the case $\hat{E} : \kw{Inert}$ is for an
un-reified stateful \emf term.
The last two cases are the most interesting: they represent the base and
inductive case of the telescope of a stateful term ``caught in the
act'' of reducing---we refer to them as the \ls$Active$ contexts.
We omit the sort of a context when it is irrelevant.
\[
\hspace{-0.45pc}
\begin{array}{lcl}
   K~s  & \!\!\!\!\!\!::=\!\!\!\!\!\!\! & \hat{E} : \kw{Tot}
          \mid  \hat{E} : \kw{Pure}
          \mid  \hat{E} : \kw{Inert}
          \mid  \kw{reify}~\hat{E}~s : \kw{Active} \\
         &\mid& \kw{Pure.bind}~(K~s)~p.((\lam{x}{\kw{reify}~f})~(\mFst~p)~(\mSnd~p)) : \kw{Active}
         \\
         && \hspace{12.5pc} (\text{if } K~s : \kw{Active})
\end{array}\]

\newcommand\sttrans[1]{\ensuremath{\{\!\![#1]\!\!\}}}

\noindent Next, we define a simple translation $\sttrans{\cdot}$ from
contexts $K~s$ to \emfST.
\[\nquad\begin{array}{l}
\sttrans{\hat{E}} = \hat{E} \\
\sttrans{\kw{reify}~\hat{E}~s} = \hat{E} \\
\sttrans{\kw{Pure.bind}~(K~s)~p.((\lam{x}{\kw{reify}~f})~(\mFst~p)~(\mSnd~p))} 
\\
\hspace{13pc} = \kw{ST.bind}~\sttrans{K~s}~x.f
\end{array}\]

The definition of $\sttrans{\cdot}$ further illustrates 
why we  need to structure  the \ls$Active$ contexts as a telescope---because not every stateful  
computation that can reduce in \emf is of the form \ls$reify$~$e$. For 
example, the reduction rule {\sc{R-ReifyBind}} pushes \ls$reify$ 
inside the arguments of \ls$bind$. As a result, one needs to 
perform several ``administrative" steps of reduction to get the 
resulting term back to being of the form \ls$reify$~$e$. However, 
in order to show that \emfST can indeed be used as a compiler 
target for \emf, we crucially need to relate all such intermediate 
redexes to \ls$ST$ computations in \emfST---thus the telescope-like definition of the \ls$Active$ contexts.

Finally, we prove the simulation theorem for \emf and \emfST, 
which shows that one or more steps of reduction
in \emfST are matched by one or more steps in \emf, 
in a compatible way. 

\begin{theorem}[Simulation]
\label{thm:simulation}
For all well-typed, closed, filled contexts $K~s~f$, either $K~s$ is \ls$Inert$, or one of the
following is true:

\begin{enumerate}
\item $\exists K' s' f'.$
      $(s, \sttrans{K~s}~f) \leadsto^{+} (s', \sttrans{K'~s'}~f')$\\
      and $K~s~f \steps^{+} K'~s'~f'$
      and $sort\, (K~s) = sort\, (K'~s')$\\
      and if $K'~s'$ is not \ls$Active$ then $s=s'$.
\item $K~s$ is \ls$Active$
      and $\exists v~s'.\, (s,\sttrans{K~s}~f) \leadsto^{*} (s',\kw{ST.return}~v)$\\
      and $K~s~f \steps^{+} \kw{Pure.return}~(v, s')$.
\item $K~s$ is \ls$Pure$
      and $\exists v.\, \sttrans{K~s}~ f = K~s~f = \kw{Pure.return}~v$.
\item $K~s$ is \ls$Tot$
      and $\exists v.\, \sttrans{K~s}~f = K~s~f = v$.
\end{enumerate}
%
\end{theorem}

\ifcheckpagebudget\clearpage\fi
\section{Related Work}
\label{sec:related}






We have already discussed many elements of related work throughout the
paper. Here we focus on a few themes not covered fully elsewhere.

Our work builds on the many uses of monads for programming language
semantics found in the literature. \citet{Moggi89} was the first to
use monads to give semantics to call-by-value
reduction---our \autoref{thm:simulation} makes use of the monadic
structure of \emf to show that it can safely be executed in a strict
semantics with primitive
state. \citet{Moggi89}, \citet{Wadler1990Comprehending,
Wadler92}, \citet{Filinski94,Filinski99,Filinski10}, \citet{BentonHM00}
and others, use monads to introduce effects into a functional
language---our approach of adding user-defined effects to the
pure \emf calculus follows this well-trodden path.
\citet{Moggi89}, \citet{flanagan93anf}, \citet{Wadler1994Composable}
and others, have used monads to provide a foundation on which to
understand program transformations, notably CPS---we show that weakest
precondition semantics can be formally related to CPS via our main
logical relation theorem (\autoref{thm:elab-typing}).

\paragraph*{Representing monads} Our work also draws a lot
from~\citepos{Filinski94} monadic reflection methodology, for
representing and controlling the abstraction of monads.
In particular, our \deflang monad definition language is essentially
the language of \citep{Filinski94} with some restrictions on the
allowed types.
Beyond controlling abstraction, Filinski shows how monadic reflection enables
a universal implementation of monads using composable continuations
and a single mutable cell. We do not (yet) make use of that aspect of
his work, partly because deploying this technique in practice is
challenging, since it requires compiling programs to a runtime system
that provides composable continuations.~\citepos{Filinski99} work on
representing layered monads generalizes his technique to the setting
of multiple monads. We also support multiple monads, but instead of
layering monads, we define each monad purely, and relate them via
morphisms.
This style is better suited to our purpose, since one
of our primary uses of reification is purification, i.e., revealing
the pure representation of an effectful term for reasoning
purposes. With layering, multiple steps of reification may be
necessary, which may be inconvenient for purification.
Finally,~\citet{Filinski10}
gives an operational semantics that is extensible with monadic
actions, taking the view of effects as being primitive, rather than
encoded purely. We take a related, but slightly different view:
although effects are encoded purely in \emf, we see it as language in
which to analyze and describe the semantics of a primitively effectful
object language, \emfST, relating the two via a simulation.

\paragraph*{Dependent types and effects}
Nanevski \ETAL developed Hoare type theory (HTT)~\cite{nmb08htt}
and Ynot~\cite{ynot} as a way of extending Coq with effects. The
strategy there is to provide an axiomatic extension of Coq with a
single catch-all monad in which to encapsulate imperative code. Being
axiomatic, their approach lacks the ability to reason extrinsically
about effectful terms by computation. However, their approach
accommodates effects like non-termination, which \emf currently
lacks. Interestingly, the internal semantics of HTT is given using
predicate transformers, similar in spirit to \emf's WP semantics. It
would be interesting to explore whether or not our free proofs of
monotonicity and conjunctivity simplify the proof burden on HTT's
semantics.

Zombie~\citep{zombie-popl14} is a dependently typed language with
general recursion, which supports reasoning extrinsically about
potentially divergent code---this approach may be fruitful to apply
to \emf to extend its extrinsic reasoning to divergent code.

Another point in the spectrum between extrinsic and intrinsic
reasoning is~\citepos{Chargueraud2011CF} characteristic formulae,
which provide a precise formula in higher-order logic capturing the
semantics of a term, similar in spirit to our WPs. 
However, as opposed to WPs, characteristic formulae are used
interactively to prove program properties after definition, although
not via computation, but via logical reasoning.
Interestingly enough, characteristic formulae are structured in a way
that almost gives the illusion that they are the terms themselves.
CFML is tool in Coq based on these ideas, providing special tactics to
manipulate formulas structured this way.

\citet{idris,idris2} encodes algebraic effects with pre- and postconditions
in Idris in the style of~\citepos{atkey09parameterised} parameterized
monads. Rather than speaking about the computations themselves, the pre-
and postconditions refer to some implicit state of the world, e.g.,
whether or not a file is closed. In contrast, \fstar's WPs give a full
logical characterization of a computation. Additionally, the WP style
is better suited to computing verification conditions, instead of
explicitly chaining indices in the parameterized monad.

It would be interesting, and possibly clarifying, to link up with
recent work on the denotational semantics of effectful languages with
dependent types~\cite{ahman16fossacs}; in our case one would
investigate the semantics of \emf and \emfST, which has state, but
extended with recursion (and so with nontermination).
 
\paragraph*{Continuations and predicate transformers}
We are not the first to study the connection between continuations and
predicate transformers.
For example, \citet{Jensen78} and \citet{AudebaudZ99} both derive WPs
from a continuation semantics of first-order imperative programs.
While they only consider several primitive effects, we allow arbitrary
monadic definitions of effects.
Also while their work is limited to the first-order case, we formalize
the connection between WPs and CPS also for higher-order.
The connection between WPs and the continuation monad also 
appears in~\citet{Keimel15, KeimelP}.

\section{Looking Back, Looking Ahead}
\label{sec:conclusion}

While our work has yielded the pleasant combination of both a
significant simplification and boost in expressiveness for \fstar, we
believe it can also provide a useful foundation on which to add
user-defined effects to other dependently typed languages. All that is
required is the \ls$Pure$ monad upon which everything else can be
built, mostly for free.

On the practical side, going forward, we hope to make use of the new
extrinsic proving capabilities in \fstar to simplify specifications
and proofs in several ongoing program verification efforts that
use \fstar. We are particularly interested in furthering the
relational verification style, sketched in~\autoref{sec:ex-ifc}. We
also hope to scale \emf to be a definitive semantics of all
of \fstar---the main missing ingredients are recursion and its
semantic termination check, inductive types, universe polymorphism,
and the extensional treatment of equality. Beyond the features
currently supported by \fstar{}, we would like to investigate adding
indexed effects and effect polymorphism.

We would also like to generalize the current work to divergent
computations. For this we do not plan any changes to \deflang. However, we
plan to extend \emf with general recursion and a primitive \ls$Div$
effect (for divergence), following the current \fstar implementation
\cite{mumon}. Each monad in \deflang will be elaborated in two ways: first,
to \ls$Pure$ for total correctness, as in the current paper; and
second, to \ls$Div$, for partial correctness. The reify operator for a
partial correctness effect will produce a \ls$Div$ computation, not a
\ls$Pure$ one. With the addition of \ls$Div$, the dynamic semantics of
\emf will force a strict evaluation order for \ls$Div$ computations,
rather than the non-deterministic strong reduction that we allow for
\ls$Pure$ computations.

Along another axis, we have already mentioned our plans to
investigate translations of effect handlers
(\autoref{sec:ex-combining}).
We also hope to enhance \deflang in other ways, \EG relaxing the
stratification of types and adding inductive types.
The latter would allow us to define monads for some forms of
nondeterminism and probabilities, as well as many forms of I/O,
provided we can overcome the known difficulties with CPS'ing inductive
types~\cite{BartheU02}.
Enriching \deflang further,
one could also add dependent types, reducing the gap between it and \fstar,
and bringing within reach examples like~\citepos{ahman13update}
dependently typed update monads.


\ifanon\else
\section*{Acknowledgments}
We are grateful to Cl\'ement Pit-Claudel for all his help with
the \fstar interactive mode; to Pierre-Evariste Dagand and Michael
Hicks for interesting discussions; and to the anonymous reviewers for
their helpful feedback.
  This work was, in part, supported by the
  \href{https://erc.europa.eu/}{European Research Council}
  under \href{https://secure-compilation.github.io/}{ERC
    Starting Grant SECOMP (715753).}

\fi

\iffull
\clearpage
\appendix
\newcommand{\cpsa}[1]{{\cps{#1}}^1}
\newcommand{\cpsb}[1]{{\cps{#1}}^2}
\newcommand{\rulename}[1]{\textsc{(#1)}}

\newcommand{\Ty}{\typez}

\section{Appendix}
\label{appendix}

In this appendix we provide proofs and auxiliary results for the
theorems that appear in the body of the paper. We also show the full
type system for the source language.


\subsection{The Definitional Language \deflang}
\label{def}

In the typing judgment, the metavariable $\Delta$ represents a set
of type variables that remains fixed throughout typing. It is used to
introduce top-level let-polymorphism on all CPS'd/elaborated terms. A
type is well-formed in the context $\Delta$ if all of its
variables are in $\Delta$.
In rigor, all judgments from here onwards are subject to that
constraint, which we do not write down. A context $\Gamma$ is
well-formed if both (1) all of its types are well-formed according
to $\Delta$ (2) no variable names are repeated. This last condition
simplifies reasoning about substitution and does not limit the
language in any way.

We assume that every base type in \deflang is also a base type in \emf
(or that there exists a mapping from them, formally), and that source
constants are also present and with the same type (formally, also a
mapping for constants that respects the previous one).

\begin{figure*}[!t]
\[\begin{array}{ccc}
    \inferrule*[lab=ST-Var]
    {x:H \in \Gamma}
    {\Delta\mid \Gamma \vdash x: H \bang \neu}
    &

    \inferrule*[lab=ST-Const]
    {\Delta\mid \Gamma \vdash e_i:b_i \bang \neu \qquad \kappa:b_1,\ldots,b_n \rightarrow b}
    {\Delta\mid \Gamma \vdash \kappa(e_1, \ldots, e_n):b\bang \neu}

    &

    \inferrule*[lab=ST-Abs]
    {\Delta\mid \Gamma, x:H \vdash e:H' \bang \varepsilon}
    {\Delta\mid \Gamma \vdash \lambda x:H.\, e: H \earr H' \bang\neu}
    \\\\

    \inferrule*[lab=ST-App]
    {\Delta\mid \Gamma \vdash e: H \earr H' \bang \neu
     \qquad  \Delta\mid \Gamma \vdash e': H \bang \neu}
    {\Delta\mid \Gamma \vdash ee': H'  \bang \varepsilon}

    &

    \inferrule*[lab=ST-Pair]
    {\Delta\mid \Gamma \vdash e: H \bang \neu
     \qquad \Delta\mid \Gamma \vdash e': H' \bang \neu}
    {\Delta\mid \Gamma \vdash (e,e'): H \times H' \bang \neu}

    &

    \inferrule*[lab=ST-Fst]
    {\Delta\mid \Gamma \vdash e: H \times H' \bang \neu}
    {\Delta\mid \Gamma \vdash \fst{e}: H \bang \neu}
    \\\\

    \inferrule*[lab=ST-Inl]
    {\Delta\mid \Gamma \vdash e: A \bang \neu}
    {\Delta\mid \Gamma \vdash \inl{e}: A + A' \bang \neu}
    &
    \multicolumn{2}{c}{
        \inferrule*[lab=ST-Case]
        {\Delta\mid \Gamma \vdash e : A + A' \bang \neu
         \qquad \Delta\mid \Gamma, x:A  \vdash e_1 : H \bang \varepsilon
         \qquad \Delta\mid \Gamma, x:A' \vdash e_2 : H \bang \varepsilon}
        {\Delta\mid \Gamma \vdash \mycases{e} {x:A.\, e_1}  {y:A'.\, e_2} : H \bang \varepsilon}
    }
    \\\\

    \inferrule*[lab=ST-Ret]
    {\Delta\mid \Gamma \vdash e : A \bang \neu}
    {\Delta\mid \Gamma \vdash  \returnT{e} : A \bang \teff}
    &
    \multicolumn{2}{c}{
        \inferrule*[lab=ST-Bind]
        {\Delta\mid \Gamma \vdash e : A \bang \teff
         \qquad \Delta\mid \Gamma, x:A \vdash e :A' \bang \teff}
        {\Delta\mid \Gamma \vdash \bindT{e}{x:A}{e'} : A' \bang \teff}
    }

\end{array}\]
\caption{Typing rules of \deflang}
\label{fig:def-typing}
\end{figure*}

The typing judgment for \deflang is given in \autoref{fig:def-typing}.
We assume that the types appearing in the rules are
well-formed. For example, in the \rulename{ST-Pair} rule, either both $H$ and
$H'$ are in $A$ or both are in $C$ \ETC

\subsection{CPS Translation (WP Generation)}
\label{app:cps}

\begin{figure*}[!t]
\[
\begin{array}{lcl@{\hspace{7em}}lcl}
    \cps{x} &=& x &
    \cps{K(e_1,\ldots,e_n)} &=& K~\cps{e_1}~\ldots~\cps{e_n} \\
    \cps{(f~e)} &=& \cps{f}~\cps{e} &
    \cps{(\lambda x:H.~e)} &=& \lambda x:\cps{H}.~\cps{e} \\
    \cps{\fst{e}}       &=& \mfst{\cps{e}}             &
    \cps{\snd{e}}       &=& \msnd{\cps{e}}             \\
    \cps{\inl{e}}       &=& \minl{\cps{e}}             &
    \cps{\inr{e}}       &=& \minr{\cps{e}}             \\
    \cps{(e_1,e_2)}     &=& (\cps{e_1},\cps{e_2})     &

                       \cps{(\mycases{e_0}{x:A.~e_1}{y:A'.~e_2})} &=&
            \caseimp{\cps{e_0}}{x.\cps{e_1}}{y.\cps{e_2}} \\\\

    \cps{(\returnT{e})} &=& \multicolumn{2}{l}{\lambda p:\cps{A} -> \typez.~p~\cps{e}}
                       &\multicolumn{2}{l}{\hspace{-3cm}\mbox{(when $\DmG \vdash e : A!\neu$)}} \\
    \cps{(\bindT{e_1}{x}{e_2})} &=& \multicolumn{2}{l}{\lambda p:\cps{A'}->\typez.~\cps{e_1}~(\lambda x:A.~\cps{e_2}~p)}
                       &\multicolumn{2}{l}{\hspace{-3cm}\mbox{(when $\DmG, x:A \vdash e_2 : A'!\teff$)}}\\

\end{array}
\]
\caption{Definition of the \cpst{} for \deflang terms}
\label{fig:cpsdef}
\end{figure*}

The full \cpst for \deflang expressions is given in
\autoref{fig:cpsdef}. The one for types was previously defined. We
define translation on environments in the following way:

\[
    \frac{\Delta = X_1,  \dots , X_m}{\cps{\Delta} = X_1:\typez, \ldots, X_m:\typez}
\qquad
    \frac{\Gamma = x_1:H_1,\ldots, x_n: H_1}{\cps{\Gamma} = x_1:\cps{H_1},\ldots, x_n: \cps{H_1}}
\]

One can then prove the following:
\begin{lemma}[Well-typing of \cpst{}]
For any $\Gamma$, $e$, $A$ and $H$:
\[\begin{array}{ll}
    \Delta \mid \Gamma \vdash e: H \bang \neu &==> \cps{\Delta}, \cps{\Gamma} \vdash \cps{e}: \cps{H} \\
    \Delta \mid \Gamma \vdash e: A \bang \teff &==> \cps{\Delta}, \cps{\Gamma} \vdash \cps{e}: (\cps{A} \rightarrow \Ty) \rightarrow \Ty \\
\end{array}
\]
\label{lemma:cps-well-typed}
\end{lemma}
\begin{proof}
By induction on the typing derivation.
\end{proof}

In this lemma statement, and in those that follow, when writing
$\cps{e}$ we refer to the translation of $e$ using the typing derivation
from the premise.

\subsection{Elaboration}

The definitions of $\un{A}$ and the F relation were previously given.
For elaboration, we also translate environments, in the following manner:
\[
    \infer{\un{\Delta} =  X_1:\typez, \ldots, X_m:\typez}{\Delta = X_1,  \dots , X_m}
\]

\[\small
    \infer{\un{x:A} = x:\un{A}}{}
\qquad
    \infer{\un{x:C} = x^w:\cps{C}, x:\F{C}{x^w}}{}
\qquad
    \infer{\un{\Gamma} = \un{x_1:H_1},\ldots, \un{x_n:H_n}}{\Gamma = x_1:H_1,\ldots, x_n: H_n}
\]
Note that for any computational variable in the context, we introduce
two variables: one for its WP and one for its actual expression. The
$x^w$ variable, which is assumed to be fresh, is used only at the WP
level.
Also note that $\un{\Delta} = \cps{\Delta}$.

For any $\Gamma$, we define the substitution $\sG$ as
$[x_{i_1}^w/x_{i_1}, \ldots, x_{i_k}^w/x_{i_k}]$, for the computational
variables $x_{i_1},\ldots,x_{i_k} \in \Gamma$.

Similarly to \autoref{lemma:cps-well-typed} we show that:
\begin{lemma}[Well-typing of \cpst{} --- elaboration contexts]
\label{lemma:cps-well-typed-elab}
For any $\Gamma$, $e$, $A$ and $C$ we have:
\[\begin{array}{lcl}
    \Delta \mid \Gamma \vdash e: H \bang \neu  &==>& \un{\Delta}, \un{\Gamma} \vdash \cps{e}\sG: \cps{H} \\
    \Delta \mid \Gamma \vdash e: A \bang \teff &==>& \un{\Delta}, \un{\Gamma} \vdash \cps{e}\sG: (\cps{A} \rightarrow \Ty) \rightarrow \Ty \\
\end{array}\]
\end{lemma}

For expression elaboration we aim to show that:
\[
    \infer{\un{\Delta}, \un{\Gamma} \vdash \un{e}: \un{A}}{\Delta \mid \Gamma \vdash e:A \bang \neu}
\qquad
    \infer{\un{\Delta}, \un{\Gamma} \vdash \un{e}: \F{C}{(\cps{e}\sG)}}{\Delta \mid \Gamma \vdash e:C \bang \neu}
\qquad
    \infer{\un{\Delta}, \un{\Gamma} \vdash \un{e}: \Pure{\un{A}}{\cps{e}\sG}}{\Delta \mid \Gamma \vdash e:A \bang \teff}
\]


\subsection{Proof of the Logical Relation Lemma}
\label{sec:lrel-proofs}

We shall prove some intermediate lemmas before.

\begin{theorem}\label{pure-relation}
For any $A$, $\Delta \mid \Gamma \vdash e : A \bang \neu \implies \un{e}
= \cps{e}~\sG$. (That is, syntactic equality).

\end{theorem}
\begin{proof}
By induction on the typing derivation. The cases for \rulename{ST-Ret} and \rulename{ST-Bind}
do not apply.

\begin{enumerate}
\item \rulename{ST-Var}

Our goal is to show $x = x~\sG$. Since the type of $x$ is $A$ the
substitution does not affect $x$, thus they're trivially both $x$.

\item \rulename{ST-Const}

Say $\Delta \mid \Gamma \vdash \kappa(b_1,\ldots,b_n) : b \bang n$. By the
induction hypothesis we know that $\un{b_i} = \cps{b_i}\sG$ for each $i$. We
thus trivially get our goal by substitution of the arguments.

\item \rulename{ST-Abs}

Say we concluded $\Delta \mid \Gamma \vdash \lambda x:A.\, e : A \narr
A' \bang \neu$. Our premise is (note the substitution from the IH
does not affect $x$, as it has an $A$-type) the fact that $\un{e} =
\cps{e}~\sG$. We need to show that:
\[
    \un{\lambda x:A.~e} = \cps{(\lambda x:A.~e)}
\]
which is just
\[
    \lambda x:\un{A}.~\un{e} = \lambda x:\cps{A}.~\cps{e}
\]
which is trivial from our hypothesis and since $\un{A} \eqdef \cps{A}$.

\item \rulename{ST-App}

Say we concluded $\Delta \mid \Gamma \vdash f~e : A' \bang \neu$ by
the premises
\[
    \Delta \mid \Gamma \vdash f : A \narr A' \bang \neu
\qquad
    \Delta \mid \Gamma \vdash e : A \bang \neu
\]
(it cannot be the case that $e$ has some $C$ type, because of the type
restrictions).
Using the inductive hypotheses we have:
\[
    \begin{array}{l}
    \un{f~e} = \un{f}~\un{e} = (\cps{f}~\sG)~(\cps{e}~\sG) =
        \cps{(f~e)}~\sG
    \end{array}
\]
As required.

\item \rulename{ST-Fst}, \rulename{ST-Snd}, \rulename{ST-Pair}, \rulename{ST-Inl}, \rulename{ST-Inr}

All of these are trivial by applying the IH. For \rulename{ST-Pair} one needs to
note that the restrictions will ensure that the type of the pair will be
an $A$-type.

\item \rulename{ST-Case}

Say we concluded $\Delta \mid \Gamma \vdash \mycases{e}{x:A_0.\,
e_1}{y:A_1.\, e_2} : A_2 \bang \neu$.
As inductive hypothesis we have:
\[
      \un{e} = \cps{e}~\sG
\qquad
      \un{e_1} = \cps{e_1}~\sG
\qquad
      \un{e_2} = \cps{e_2}~\sG
\]
($e_1$ and $e_2$ are typed in $\Gamma$ extended with $x$ and $y$
respectively, however since they are $A$-typed the $\sG$ substitution is
the same)

The goal is:
\[\begin{array}{cl}
  & (\caseimp{\un{e}}{x.\un{e_1}}{y.\un{e_2}}) \\
= & (\caseimp{\cps{e}~\sG}{x.\cps{e_1}~\sG}{y.\cps{e_2}~\sG})
\end{array}
\]
We trivially get our goal from the IHs.

\end{enumerate}
\end{proof}

\begin{theorem}\label{a-elab-type}
    If $\Delta \mid \Gamma \vdash e : A \bang \neu$, then
    $\un\Delta,\un\Gamma \vdash \un{e} : \un{A}$.
\end{theorem}
\begin{proof}
    By induction on the typing derivation.

\begin{enumerate}
\item \rulename{ST-Var}

We have $\Delta \mid \Gamma \vdash x:A \bang \neu$, with $x \in \Gamma$.
By the translation for environments, we have $x : \un{A}$ in $\Gamma$,
so this is trivial.

\item \rulename{ST-Const}

For any constant $\kappa : (b_1,\ldots,b_n) \rightarrow b$ say we have
$\Delta \mid \Gamma \vdash \kappa(e_1,\ldots, e_n):b \bang \neu$ by \rulename{ST-Const}
(note that $b$ and all the $b_i$ are in $A$).
This means that for every $i$ we have as inductive hypothesis:
    \[ \un{\Delta}, \un{\Gamma} \vdash \un{e_i} : b_i \]
Since $\kappa$ is also a target constant of the same type, we thus have:
    \[ \un{\Delta}, \un{\Gamma} \vdash \kappa~\un{e_1}~\ldots~\un{e_n} : b \]
Which is exactly our goal as $\un{b} = b$.

\item \rulename{ST-Fst}, \rulename{ST-Snd}, \rulename{ST-Pair}, \rulename{ST-Inl}, \rulename{ST-Inr}

Trivial by using IH.

\item \rulename{ST-Case}

Say we concluded $\Delta \mid \Gamma \vdash \mycases{e} {x:A_0.\, e_1}
{y:A_1.\, e_2} : A_2 \bang \neu$ by \rulename{ST-Case}. Our IHs give us
\[\arraycolsep=1pt
    \begin{array}{lclcl}
    \un\Delta, \un\Gamma & \vdash & \un{e} &:& \un{A_0} + \un{A_1}     \\
    \un\Delta, \un\Gamma, x:\un{A_0} & \vdash & \un{e_1} &:& \un{A_2} \\
    \un\Delta, \un\Gamma, y:\un{A_1} & \vdash & \un{e_2} &:& \un{A_2}
\end{array}\]
By a non-dependent application of T-CaseTot we get
\[
    \un\Delta, \un\Gamma \vdash \caseimp{\un{e}}{x.\un{e_1}}{y.\un{e_2}} : \un{A_2}
\]
Which is our goal.

\item \rulename{ST-Abs}, \rulename{ST-App}

Both trivial from IHs.

\end{enumerate}
\end{proof}

Before jumping into the logical relation lemma, we will require the
following auxiliary lemma, of which we make heavy use.

\begin{lemma}[Invariancy of $\F{C}{w}$]
    If $\Gamma \vDash w_1 = w_2$, then $\Gamma \vdash \F{C}{w_1} <: \F{C}{w_2}$.
\end{lemma}
\begin{proof}
By induction on $C$.
\begin{enumerate}
    \item $C \tarr A$

        We need to show that
        \[
            \Gamma \vdash \F{C \tarr A}{w_1} <: \F{C \tarr A}{w_2}
        \]
        Which is
        \[\arraycolsep=1pt
        \begin{array}{rcl}
            \Gamma &\vdash& x^w:\cps{C} -> \F{C}{x^w} -> \Pure{\un{A}}{w_1~x^w} \\
                        && <: x^w:\cps{C} -> \F{C}{x^w} -> \Pure{\un{A}}{w_2~x^w}
        \end{array}
        \]
        After two applications of \rulename{ST-Prod} (and some trivial reflexivity
        discharges), the required premise to show is:
        \[\small
            \Gamma, x^w:\cps{C}, \_:\F{C}{x^w} \vdash \Pure{\un{A}}{w_1~x^w} <: \Pure{\un{A}}{w_2~x^w}
        \]
        By \rulename{S-Pure} we're required to show that $\un{A}$ is a
        subtype of itself (which is trivial by reflexivity of subtyping
        \rulename{S-Conv}) and that $w_2$ is stronger than $w_1$,
        which can be easily proven as they are equal.

    \item $A \tarr A$

        Very similar to the previous case, but simpler.

    \item $C \narr C'$

        We need to show that
        \[
            \Gamma \vdash \F{C \narr C'}{w_1} <: \F{C \narr C'}{w_2}
        \]
        Which is
        \[\arraycolsep=1pt
        \begin{array}{rcl}
            \Gamma &\vdash& x^w:\cps{C} -> \F{C}{x^w} -> \F{C'}{(w_1~x^w)} \\
                      && <: x^w:\cps{C} -> \F{C}{x^w} -> \F{C'}{(w_2~x^w)}
        \end{array}
        \]
        After two applications of \rulename{ST-Prod} (and some trivial reflexivity
        discharges), the required premise to show is:
        \[\small
            \Gamma, x^w:\cps{C}, \_:\F{C}{x^w} \vdash \F{C'}{(w_1~x^w)} <: \F{C'}{(w_2~x^w)}
        \]
        As in this context we can show $w_1~x^w = w_2~x^w$ we apply our
        IH to the type $C'$ and are done.

    \item $A \narr C'$

        Also very similar to the previous case, but simpler.

    \item $C \times C'$

        Trivial by IHs and concluding that $\mfst{w_1} = \mfst{w_2}$, and similarly
        for $\mSnd$.

\end{enumerate}

\end{proof}

\paragraph{Proof of \autoref{thm:elab-typing} (\nameref{thm:elab-typing})}

\begin{proof}

The two parts are proved by a joint structural induction.

\begin{enumerate}
\item \rulename{ST-Var}

We have $\Delta \mid \Gamma \vdash x:C \bang \neu$, with $x \in
\Gamma$.
By the translation for environments, we have $x^w : \cps{C}$ and $x :
\F{C}{x^w}$ in $\un\Gamma$.
Since $x$ is covered by the substitution $\sG$, what we need to prove
is $\un\Delta,\un\Gamma \vdash x : \F{C}{x^w}$, which is exactly what
we have in the environment.

\item \rulename{ST-Pair}

Suppose we proved $(e_1,e_2) : C_1 \times C_2 \bang \neu$ by
\rulename{ST-Pair}. We want to show: $\un{\Delta}, \un{\Gamma} \vdash
(\un{e_1},\un{e_2}):\F{C_1 \times C_2}{((\cps{e_1}, \cps{e_2})~\sG)}$, i.e.,
that (after reduction inside F):
\[\un{\Delta}, \un{\Gamma} \vdash (\un{e_1}, \un{e_2}): \F{C_1}{(\cps{e_1}\sG)} \times \F{C_2}{(\cps{e_2}\sG)}\]
This is trivial by applying both IHs.

\item \rulename{ST-Fst}, \rulename{ST-Snd}

Suppose we proved $\Delta \mid \Gamma \vdash \fst{e} : C_1 \bang \neu$
by \rulename{ST-Fst}. We need to then show
\[
    \un{\Delta}, \un{\Gamma} \vdash \mfst{\un{e}}:\F{C_1}{(\mfst{\cps{e}} \sG)}
\]
By our induction hypothesis we have $\un{\Delta}, \un{\Gamma} \vdash
\un{e}:\F{C_1 \times C_2}{(\cps{e} \sG)}$, which is
\[
    \un{\Delta}, \un{\Gamma} \vdash \un{e} : \F{C_1}{(\mfst{\cps{e} \sG})} \times \F{C_2}{(\msnd{\cps{e} \sG})}
\]
It is therefore easy to see that we have our goal.

\item \rulename{ST-Abs}

There are two cases:
\begin{itemize}
\item $A \earr H$

Suppose we concluded $\Delta \mid \Gamma \vdash \lambda x: A.\, e: A \earr H
\bang \neu$. Then we have $\Delta \mid \Gamma, x: A \vdash e: H \bang
\epsilon$ and so, by the induction hypothesis, in both cases for $\epsilon$ we have
    \[ \un{\Delta}, \un{\Gamma}, x: \un{A} \vdash \un{e}: \G{\epsilon}{H}{\cps{e}~\sG} \]
And we have to show:
    \[ \un{\Delta}, \un{\Gamma} \vdash \lambda x: \un{A}.\,\un{e}: \F{A \narr H}{((\lambda x:\cps{A}. \cps{e})~\sG)}\]
Which is
    \[ \un{\Delta}, \un{\Gamma} \vdash \lambda x: \un{A}.\,\un{e}: x:\un{A} -> \G{\epsilon}{H}{(\lambda x:\cps{A}. \cps{e})~\sG~x} \]
Since the substitution does not cover $x$, the argument to G is just
$\cps{e}~\sG$, thus we use our IH to conclude this easily.

\item $C \earr H$

Suppose we concluded $\Delta \mid \Gamma \vdash \lambda x: C.\, e: C
\earr H \bang \neu$. Then we have $\Delta \mid \Gamma, x: C \vdash e:
H \bang \epsilon$ and so, by the induction hypothesis we have, in either case for $\epsilon$:
    \[ \un{\Delta}, \un{\Gamma}, x^w : \cps{C}, x : \F{C}{x^w} \vdash \un{e}: \G{\epsilon}{H}{\cps{e}~\sG~[x^w/x]} \]
And we have to show:
    \[ \un{\Delta}, \un{\Gamma} \vdash \lambda x^w: \cps{C}.\, \lambda x:\F{C}{x^w}.\, \un{e}: \F{C \earr H}{((\lambda x:\cps{C}. \cps{e})~\sG)} \]
Which is
    \[\small
      \arraycolsep=1pt
      \begin{array}{lrl}
          \un{\Delta}, \un{\Gamma} & \vdash & \lambda x^w: \cps{C}.\, \lambda x:\F{C}{x^w}.\, \un{e} \\
                                   &    :   & x^w:\cps{C} -> \F{C}{x^w} -> \G{\epsilon}{H}{(\lambda x:\cps{C}. \cps{e})~\sG~x^w}
      \end{array}
    \]
Using T-Abs twice we can conclude this via
    \[ \un{\Delta}, \un{\Gamma}, x^w: \cps{C}, x:\F{C}{x^w} \vdash \un{e} : \G{\epsilon}{H}{(\lambda x:\cps{C}. \cps{e})~\sG~x^w} \]

Since the substitution does not cover $x$, the argument to F reduces to
$\cps{e}~\sG~[x^w/x]$, thus we use our IH to conclude this easily.
\end{itemize}

\item \rulename{ST-App}

Again, There are two possible cases:
\begin{itemize}

\item $A \earr H$

We concluded $\DmG \vdash fe : G \bang \epsilon$.
Our premises are
$ \DmG \vdash f : A \earr C \bang \neu $
and
$ \DmG \vdash e : A \bang \neu $.
The IH for $f$ is, expanding F:
    \[ \un\Delta, \un\Gamma \vdash \un{f} : x:\un{A} -> \G{\epsilon}{H}{(\cps{f}~\sG)~x} \]
By T-App, and since $\un{e} : \un{A}$, this is just:
    \[ \un\Delta, \un\Gamma \vdash \un{f}~\un{e} : \G{\epsilon}{H}{(\cps{f}~\sG)~\un{e}} \]
Since from a previous theorem we know we have $\un{e} = \cps{e}~\sG$
(syntactically), we can conclude:
    \[ \un\Delta, \un\Gamma \vdash \un{f}~\un{e} : \G{\epsilon}{H}{(\cps{f}~\sG)~(\cps{e}~\sG)} \]
This is exactly:
    \[ \un\Delta, \un\Gamma \vdash \un{f~e} : \G{\epsilon}{H}{\cps{(f~e)}~\sG} \]
which is our goal, in either the $C\bang\neu$ or the $A\bang\teff$ case.

\item $C \earr H$ \\

We concluded $\DmG \vdash fe : H \bang \epsilon$.
Our premises are
$ \DmG \vdash f : C \earr H \bang \neu $
and
$ \DmG \vdash e : C \bang \neu $
The IHs are, expanding F:
\[\begin{array}{lll}
    \un\Delta, \un\Gamma &\vdash& \un{f} : x^w:\cps{C} -> \F{C}{x^w} -> \G{\epsilon}{H}{(\cps{f}~\sG)~x^w}   \\
    \un\Delta, \un\Gamma &\vdash& \un{e} : \F{C}{(\cps{e}~\sG)}     \\
\end{array}\]
Thus by two uses of T-App (noting that it's well typed by our IH for
$e$), we can conclude:
\[
    \un\Delta, \un\Gamma \vdash \un{f}~(\cps{e}~\sG)~\un{e} : \G{\epsilon}{C'}{(\cps{f}~\sG)~(\cps{e}~\sG)}
\]
This is just, syntactically:
    \[ \un\Delta, \un\Gamma \vdash \un{f~e} : \G{\epsilon}{C'}{(\cps{f}~\sG)~(\cps{e}~\sG)} \]
which is our goal, in either the $C\bang\neu$ or the $A\bang\teff$ case.

\end{itemize}

\item \rulename{ST-Case}

There are two cases depending on wether we eliminate into $C\bang\neu$
or $A\bang\teff$. Both of these cases are quite dull, and deal mostly
with the typing judgment on the target. This may be skipped without
hindering any of the main ideas.

\begin{itemize}
\item $C\bang\neu$

Suppose that
$\Delta \mid \Gamma \vdash e: A + A' \bang \neu$,
$\Delta \mid \Gamma, x:A  \vdash  e_1: C \bang \neu$, and
$\Delta \mid \Gamma, y:A' \vdash  e_2: C \bang \neu$, so that
$\Delta \mid \Gamma \vdash \mycases{e} {x:A.\, e_1}  {y:A'.\, e_2} : C \bang \neu$.
We will go into detail only for $e_1$ as the typing and reasoning for
$e_2$ is exactly analogous.
As inductive hypothesis for $e_1$ we have
\[
    \un{\Delta}, \un{\Gamma}, x:\un{A} \vdash \un{e_1} : \F{C}{(\cps{e_1}~\sG)}
\]
We wish to show:
\[\small
\arraycolsep=1pt
  \begin{array}{lll}
    \un{\Delta},\un{\Gamma} & \vdash & \casefull{\un{e}}{z} {x.\un{e_1}} {y.\un{e_2}\\&&}
                     {\F{C}{(\caseimp{z}{\cps{e_1}~\sG}{\cps{e_2}~\sG})}} \\
  & : & \F{C}{((\caseimp{\cps{e}} {x.\cps{e_1}}  {y.\cps{e_2}})~\sG)}
\end{array}\]
Which is
\[\small
\arraycolsep=1pt
  \begin{array}{lll}
    \un{\Delta},\un{\Gamma} & \vdash & \casefull{\un{e}}{z} {x:\un{A}.\, \un{e_1}} {y: \un{A'}.\, \un{e_2}\\&&}
                     {\F{C}{(\caseimp{z}{x.\cps{e_1}~\sG}{y.\cps{e_2}~\sG})}} \\
  & : & \F{C}{(\caseimp{\cps{e}~\sG} {x.\cps{e_1}~\sG}  {y.\cps{e_2}~\sG})}
\end{array}\]
Since $\un{e} = \cps{e}~\sG$ we will prove this has type
\[
    \F{C}{(\caseimp{\un{e}} {x.\cps{e_1}~\sG} {y.\cps{e_2}~\sG})}
\]
By T-CaseTot, we should show:
\[\small
\arraycolsep=1pt
\begin{array}{lll}
    \un{\Delta}, \un{\Gamma} & \vdash & \un{e} : \un{A} + \un{A'} \\
    \un{\Delta}, \un{\Gamma}, x:\un{A} & \vdash & \un{e_1} :
                     \F{C}{(\caseimp{\minl{x}} {x.\cps{e_1}~\sG}  {y.\cps{e_2}~\sG})} \\
\end{array}
\]
(And the one for $e_2$). We get the first one trivially by theorem
\ref{a-elab-type}. The second is by reduction equivalent to:
\[
    \un{\Delta}, \un{\Gamma}, x:\un{A} \vdash \un{e_1} :
                \F{C}{(\cps{e_1}~\sG)}
\]
Which is exactly our IH for $e_1$, so we're done.

\item $A\bang\teff$

Our hypotheses are:
\[\small
\arraycolsep=1pt
\begin{array}{lll}
    \un{\Delta}, \un{\Gamma} & \vdash & \un{e} : \un{A_0} + \un{A_1} \\
    \un{\Delta}, \un{\Gamma}, x:\un{A_0} & \vdash & \un{e_1} : \Pure{\un{A_2}}{\cps{e_1}~\sG} \\
    \un{\Delta}, \un{\Gamma}, y:\un{A_1} & \vdash & \un{e_2} : \Pure{\un{A_2}}{\cps{e_2}~\sG} \\
\end{array}
\]
Applying T-Case (non-dependently) we get.
\[\small
\arraycolsep=1pt
\begin{array}{ll}
    \un{\Delta}, \un{\Gamma} \vdash & \caseimp{\un{e}} {x.\un{e_1}}{y.\un{e_2}} : \\
            & \Pure{\un{A_2}}{\caseimp{\un{e}}{x.\cps{e_1}~\sG}{y.\cps{e_2}~\sG}}
\end{array}
\]
Since we know $\un{e} = \cps{e}~\sG$ and since $\un{A} \eqdef \cps{A}$ this
is exactly:
\[\small
\arraycolsep=1pt
\begin{array}{ll}
    \un{\Delta}, \un{\Gamma} \vdash & \caseimp{\un{e}} {x.\un{e_1}}{y.\un{e_2}} : \\
            & \Pure{\un{A_2}}{(\caseimp{\cps{e}}{x.\cps{e_1}}{y.\cps{e_2}})~\sG}
\end{array}
\]
Which is our goal.

\end{itemize}

\item \rulename{ST-Ret}

We have $\Delta \mid \Gamma \vdash e : A \bang \neu$.
We need to show:
\[ \un{\Delta}, \un{\Gamma} \vdash \un{\returnT{e}} : \Pure{\un{A}}{\cps{\returnT{e}}~\sG} \]
i.e.
\[ \un{\Delta}, \un{\Gamma} \vdash \pureret{\un{A}}{\un{e}} : \Pure{\un{A}}{\lambda p:\cps{A}->\Ty.~(\cps{e}~\sG)} \]
This is a trivial consequence of \autoref{pure-relation} by
using the T-Ret rule of \emf, and the fact that $\un{A} \eqdef \cps{A}$.

\item \rulename{ST-Bind}

Suppose we have $\DmG \vdash e_1: A \bang \teff$, and
$\DmG, x : A \vdash e_2 : A' \bang \teff$, and so
$\DmG \vdash \bindT{e_1}{x:A}{e_2}: A' \bang \teff$. We have to show:
\[
    \un{\Delta},\un{\Gamma} \vdash \un{\bindT{e_1}{x:A}{e_2}} \,:\, \Pure{\un{A'}}{\cps{(\bindT{e_1}{x:A}{e_2})}~\sG}
\]
that is:
\[\small
\arraycolsep=1pt
\begin{array}{ll}
    \un{\Delta}, \un{\Gamma} \vdash & \purebind{\un{A}}{\un{A'}}{(\cps{e_1}~\sG)}{\un{e_1}}{(\lambda x:\cps{A}.\, \cps{e_2}~\sG)}{(\lambda x:\un{A}.\, \un{e_2})} : \\
     & \Pure{\un{A'}}{\lambda p:\cps{A'} -> Ty.~\cps{e_1}(\lambda x:\cps{A}.\, \cps{e_2}~p)}
\end{array}
\]
By our IHs we have:
\[\small
\arraycolsep=1pt
\begin{array}{lcl}
        \un{\Delta}, \un{\Gamma} & \vdash & \un{e_1} : \Pure{\un{A}}{\cps{e_1}~\sG} \\
        \un{\Delta}, \un{\Gamma}, x:\un{A} & \vdash & \un{e_2} : \Pure{\un{A'}}{\cps{e_2}~\sG} \\
\end{array}
\]
So we get:
\[
    \un{\Delta}, \un{\Gamma} \vdash \lambda x:\un{A}.\, \un{e_2} : x:\un{A} -> \Pure{\un{A'}}{\cps{e_2}~\sG}
\]
By the T-Bind rule we can conclude:
\[\small
\arraycolsep=1pt
\begin{array}{ll}
    \un{\Delta}, \un{\Gamma} \vdash & \purebind{\un{A}}{\un{A'}}{(\cps{e_1}~\sG)}{\un{e_1}}{(\lambda x:\cps{A}.\, \cps{e_2}~\sG)}{(\lambda x:\un{A}.\, \un{e_2})} : \\
     & \Pure{\un{A'}}{\lambda p:\un{A'} -> Ty.~\cps{e_1}(\lambda x:\un{A}.\, \cps{e_2}~p)}
\end{array}
\]
Since $\un{A} = \cps{A}$ and $\un{A'} = \cps{A'}$ this is exactly our goal.
\end{enumerate}
\end{proof}

Note: in this proof, we didn't use any specific fact about Pure,
except the relation between monad operations and their WPs, so this is
all generalizable to another target monad that's already defined and
satisfies the base conditions for return and bind.

\subsection{Equality Preservation}

We want to show that any source monad will give rise to
specification-level monads in the target. This will be a consequence on
the fact that equality is preserved by the \cpst. From equality preservation we
also get the property that lifts are monad morphisms without further effort.

First we define equality for the source language. It is basically standard
$\beta\eta$-equivalence adding the monad laws for the base monad $T$. We
keep the type and effect of each equality. It is an invariant that if we
can derive an equality, both sides are well-typed at the specified type
and effect.

\begin{figure*}[!t]
\[\begin{array}{cccccc}

    \multicolumn{3}{c}{
    \inferrule*[lab=Eq-Beta]
    { \Delta \mid \Gamma, x : H \vdash e_1 : H'\bang\varepsilon
     \qquad \Delta \mid \Gamma \vdash e_2 : H\bang\neu }
    {\Delta \mid \Gamma \vdash (\lambda x:H.\, e_1) e_2 = e_1[e_2/x] : H'\bang\varepsilon }
    }

    &

    \multicolumn{3}{c}{
    \inferrule*[lab=Eq-Eta]
    {\Delta\mid\Gamma \vdash e : H\earr H'\bang\neu
     \qquad x \notin FV(e) }
    {\Delta\mid\Gamma \vdash (\lambda x:H.\, e~x) = e : H\earr H' \bang \neu}
    }

    \\\\

    \multicolumn{3}{c}{
    \inferrule*[lab=Eq-App]
    {\Delta \mid \Gamma \vdash e_1 = e_1' : H \earr H' \bang \neu
     \qquad \Delta \mid \Gamma \vdash e_2 = e_2' : H \bang \neu}
    {\Delta \mid \Gamma \vdash e_1~e_2 = e_1'~e_2' : H' \bang \varepsilon}
    }

    &

    \multicolumn{3}{c}{
    \inferrule*[lab=Eq-Abs]
    {\Delta \mid \Gamma, x : H \vdash e = e' : H' \bang \varepsilon}
    {\Delta \mid \Gamma \vdash (\lambda x:H.\, e) = (\lambda x:H.\, e') : H \earr H' \bang \neu}
    }

    \\\\

    \inferrule*[lab=Eq-Refl]
    {\Delta \mid \Gamma \vdash e : H \bang\varepsilon}
    {\Delta \mid \Gamma \vdash e = e : H \bang\varepsilon}

    &

    \inferrule*[lab=Eq-Symm]
    {\Delta \mid \Gamma \vdash e_1 = e_2 : H \bang\varepsilon}
    {\Delta \mid \Gamma \vdash e_2 = e_1 : H \bang\varepsilon}

    &

    \multicolumn{4}{c}{
    \inferrule*[lab=Eq-Trans]
    {\Delta \mid \Gamma \vdash e_1 = e_2 : H \bang\varepsilon
     \qquad \Delta \mid \Gamma \vdash e_2 = e_3 : H \bang\varepsilon}
    {\Delta \mid \Gamma \vdash e_1 = e_3 : H \bang\varepsilon}
    }

    \\\\

    \multicolumn{2}{c}{
    \inferrule*[lab=Eq-Pair]
    {\Delta\mid\Gamma \vdash e : H \times H' \bang \neu}
    {\Delta\mid\Gamma \vdash (\fst{e}, \snd{e}) = e : H \times H' \bang \neu}
    }

    &

    \multicolumn{4}{c}{
    \inferrule*[lab=Eq-Case]
    {\Delta\mid\Gamma \vdash e : A + A' \bang \neu}
    {\Delta\mid\Gamma \vdash \mycases{e}{x.\inl{x}}{x.\inr{x}} = e : A + A' \bang \neu}
    }
    \\\\

    \multicolumn{3}{c}{
    \inferrule*[lab=Eq-M1]
    {\Delta\mid\Gamma \vdash m : A\bang\teff}
    {\Delta\mid\Gamma \vdash \bindT{m}{x}{(\returnT{x})} = m : A\bang\teff}
    }

    &

    \multicolumn{3}{c}{
    \inferrule*[lab=Eq-M2]
    {\Delta\mid\Gamma \vdash e : A\bang\neu
     \qquad \Delta\mid\Gamma \vdash f : A\tarr A'\bang\neu
     \qquad x \notin FV(f)}
    {\Delta\mid\Gamma \vdash \bindT{(\returnT{e})}{x}{f~x} = f~e : A'\bang\teff}
    }

    \\\\

    \multicolumn{6}{c}{
    \inferrule*[lab=Eq-M3]
    {\Delta\mid\Gamma \vdash m : A\bang\teff
     \qquad \Delta\mid\Gamma,x:A \vdash e_1 : A'\bang\teff
     \qquad \Delta\mid\Gamma,y:A'  \vdash e_2 : A''\bang\teff
     \qquad x \notin FV(e_2) }
    {\Delta\mid\Gamma \vdash \bindT{(\bindT{m}{x}{e_1})}{y}{e_2} = \bindT{m}{x}{(\bindT{e_1}{y}{e_2})} : A'' \bang \teff}
    }

    \\\\

\end{array}
\]
\caption{Equality rules for \deflang}
\label{fig:eqrules}
\end{figure*}

The definition of the equality judgment is in \autoref{fig:eqrules}.
Besides those rules, there is a congruence rule for every source
construct and computation rules for pairs and sums, as expected.

We then prove that:
\[
    \infer{\un\Delta,\un\Gamma \vDash \cps{e_1}~\sG \eqtwo \cps{e_2}~\sG }
          { \Delta\mid\Gamma \vdash e_1 = e_2 : H \bang\varepsilon}
\]
Where by $\vDash$ it is meant the validity judgment of \emf.

\begin{theorem}[Preservation of equality by CPS]\label{eq-star-again}
If $\Delta\mid\Gamma \vdash e_1 = e_2 : H \bang\varepsilon$
for any $\Delta,\Gamma,e_1,e_2,H,\varepsilon$, then one has
$\un\Delta,\un\Gamma \vDash \cps{e_1}~\sG \eqtwo \cps{e_2}~\sG$ .
\end{theorem}

\begin{proof}
By induction on the equality derivation.
Most of the cases are trivial, since \emf has very similar rules for
equality. The interesting cases are the monadic equalities, which we
show here:

\begin{enumerate}
\item \rulename{Eq-M1}

We concluded
        \[ \DmG \vdash \bindT{m}{x}{(\returnT{x})} = m : A \bang\teff \]
Thus we need to show that
        \[ \un\Delta, \un\Gamma \vDash \cps{(\bindT{m}{x}{(\returnT{x})})}~\sG \eqtwo \cps{m}~\sG \]
That is:
\[
    \un\Delta, \un\Gamma \vDash (\lambda p.~(\cps{m}~\sG)(\lambda x.~(\lambda p'.~p'~x)~p)) \eqtwo \cps{m}~\sG
\]
This is trivially provable by $\beta\eta$-reduction.

\item \rulename{Eq-M2}

We concluded
        \[ \Delta\mid\Gamma \vdash \bindT{(\returnT{e})}{x}{f~x} = f~e : A' \bang\teff \]
thus we need to show that
        \[ \un\Delta, \un\Gamma \vDash \cps{(\bindT{(\returnT{e})}{x}{f~x})}~\sG \eqtwo \cps{(f~e)}~\sG \]
That is:
\[
    \un\Delta, \un\Gamma \vDash (\lambda p.~(\lambda p'.~p'~\cps{e})~(\lambda x.~\cps{f}~x~p))~\sG= (\cps{f}~\sG) (\cps{e}~\sG)
\]
Note that since $x \notin FV(f) \implies x \notin FV(\cps{f}~\sG)$, this is
easily shown by $\beta\eta$-reduction as well.

\item \rulename{Eq-M3}

We concluded
\[\small
\arraycolsep=1pt
\begin{array}{ll}
    \DmG \vdash & \bindT{(\bindT{m}{x}{e_1})}{y}{e_2} \\
                & = \bindT{m}{x}{(\bindT{e_1}{y}{e_2})} : A'' \bang \teff
\end{array}
\]
thus we need to show that
\[\small
\arraycolsep=1pt
\begin{array}{ll}
    \un\Delta,\un\Gamma \vDash & \cps{(\bindT{(\bindT{m}{x}{e_1})}{y}{e_2})}~\sG \\
                               & \eqtwo \cps{(\bindT{m}{x}{(\bindT{e_1}{y}{e_2})})}~\sG
\end{array}
\]
That is:
\[\small
\arraycolsep=1pt
\begin{array}{ll}
          \un\Delta,\un\Gamma \vDash
            & (\lambda p.~(\lambda p'.~(\cps{m}~\sG)~(\lambda x.~(\cps{e_1}~\sG) p'))~(\lambda y.~(\cps{e_2}~\sG)~p)) \\
            & \eqtwo (\lambda p.~(\cps{m}~\sG)~(\lambda x.~(\lambda p'.~(\cps{e_1}~\sG)~(\lambda y.~(\cps{e_2}~\sG)~p'))~p))
\end{array}
\]
Note that since $x \notin FV(e_2) \implies x \notin FV(\cps{e_2})$, this
is also easily shown by $\beta\eta$-reduction.

\end{enumerate}

The proof above is easy, and that should not be surprising, as
we are translating our abstract monadic operations into a concrete
monad (continuations), thus our source equalities should be trivially
satisfied after translation.

\gm{Go into more detail in the eq proofs?}
\end{proof}

\subsection{Monotonicity}

We're interested in the monotonicity of WPs. Firstly, we need a
higher-order definition for this property.
Throughout this section we mostly ignore \deflang's typing restrictions
and work with a larger source language.
This gives us a stronger result than strictly necessary.

The types where the translation is defined are those non-dependent
and monad-free (meaning every arrow in them is a $\mTot$-arrow). No
ocurrence of $\mPure$ is allowed. The types of specifications are always
of this shape, so this is not a limitation.

For non-empty environments the theorem states:

\begin{theorem}[Monotonicity of \cpst --- environments]
\label{theorem:mono}
For any $\Delta, \Gamma, e, H, A$ one has:
\[\small
\begin{array}{llll}
    \mbox{1.} & \Delta \mid \Gamma \vdash e : H \bang \neu  & \implies & \un\Delta, \Gamma^{12} \vDash \cpsa{e} \stg_{\cps{H}} \cpsb{e} \\
    \mbox{2.} & \Delta \mid \Gamma \vdash e : A \bang \teff & \implies & \un\Delta, \Gamma^{12} \vDash \cpsa{e} \stg_{\dneg{\cps{A}}} \cpsb{e} \\
\end{array}
\]
\end{theorem}
Where we define
\[\arraycolsep=1pt
\begin{array}{rcl}
    \cdot^{12} &=& \cdot \\
    (\Gamma, x:t)^{12} &=& \Gamma^{12},
    x^1 : \cps{t}, x_1^2 : \cps{t},
    [x^1 \stg_{\cps{t}} x^1 \land
     x^1 \stg_{\cps{t}} x^2 \land
     x^2 \stg_{\cps{t}} x^2] \\
\end{array}
\]
which essentially duplicates each variable and asserts both monotonicity
for each of them and their ordering ($[\phi]$ is notation for $h:\phi$,
where $h$ does not appear free anywhere). We then define the $-^1$
substitution as $[x_1^1/x_1, \ldots, x_n^1/x_n]$ and similarly for
$-^2$.
This trivially implies the previous monotonicity theorem.

Before jumping into the proof, we define and prove the following lemma:
\begin{lemma}
    For any $\phi$, $\Gamma^{12} \vDash \phi$ implies $\Gamma^{12} \vDash \phi[2->1]$.
    Where $[2->1]$ is the substitution mapping $x^2$ to $x^1$ for every $x$ in
    $\Gamma$. Analogously $\Gamma^{12} \vDash \phi[1->2]$.

\end{lemma}
\begin{proof}
    By induction on $\Gamma$. This is trivial for an empty gamma. For
    $\Gamma' = \Gamma, x:t$, assume $(\Gamma, x:t)^{12} \vDash \phi$
    holds. By applying \rulename{V-$\forall$i} three times, we get:
    \[
        \Gamma^{12} \vDash \forall x^1, x^2.
             [x^1 \stg x^1 \land
              x^1 \stg x^2 \land
              x^2 \stg x^2] \implies \phi
    \]
    By our IH we get:
    \[
        \Gamma^{12} \vDash \forall x^1, x^2.
             [x^1 \stg x^1 \land
              x^1 \stg x^2 \land
              x^2 \stg x^2] \implies \phi[2->1]
    \]
    By weakening (note that $x^1$ and $x^2$ are not free in the RHS)
    we get that:
    \[
        (\Gamma, x:t)^{12} \vDash \forall x^1, x^2.
             [x^1 \stg x^1 \land
              x^1 \stg x^2 \land
              x^2 \stg x^2] \implies \phi[2->1]
    \]
    We can instantiate this (via \rulename{V-$\forall$e}) with
    $x^1$ on both variables to get:
    \[
        (\Gamma, x:t)^{12} \vDash
             [x^1 \stg x^1 \land
              x^1 \stg x^1 \land
              x^1 \stg x^1] \implies \phi[2->1][x^1/x^2]
    \]
    Furthermore, it is trivial to show this antecedent in the context
    $(\Gamma, x:t)^{12}$ and so we get our goal of:
    \[
        (\Gamma, x:t)^{12} \vDash \phi[2->1][x^1/x^2]
    \]
\end{proof}

\begin{lemma}
\label{lemma:subst-12}
    For any $e$, if $\Gamma^{12} \vDash \cpsa{e} \stg \cpsb{e}$, then
    $\Gamma^{12} \vDash \cpsa{e} \stg \cpsa{e}
                  \land \cpsa{e} \stg \cpsb{e}
                  \land \cpsb{e} \stg \cpsb{e}$.
\begin{proof}
    Trivial from previous lemma by noting that $\cpsa{e}[1->2]
    = \cpsb{e}$ and likewise for $[2->1]$, and then using
    \rulename{V-AndIntro}.
\end{proof}
\end{lemma}

\paragraph{Proof of \autoref{theorem:mono} (\nameref{theorem:mono})}

\begin{proof}

We prove these two propositions by induction on the typing derivation
for $e$. Throughout the proof $\Delta$ plays no special role, so we just
drop it from the reasoning, keeping in mind that it has to be there for
having well-formed types (but nothing else).

Note that during the proof we treat $\stg_X$ abstractly, so any
instantiation with a proper type (not necessarily those where $\stg$
reduces to equality) would be OK.

Throughout this proof we sometimes skip the subindices for $\stg$ in
favor of compactness. Hopefully, they should be clear from the context.

\begin{enumerate}
\item \rulename{ST-Var}

    We need to show $\Gamma^{12} \vDash x_i^1 \stg_{\cps{t_i}} x_i^2$.
    This is trivial from the context and by using the
    \rulename{V-Assume} and \rulename{V-AndElim$i$} rules.

\item \rulename{ST-Const}

    The constants only deal with base types, so all inductive hypotheses
    for the arguments reduce to an equality, as does our goal. Our goal
    is is then trivially provable by applications of \rulename{V-EqP}.

\item \rulename{ST-Abs}

    Say we concluded $\Gamma, x:t \vdash e : s \bang \varepsilon$
    As IH we have:
        \[ (\Gamma, x:t)^{12} \vDash \cpsa{e}[x^1/x] \stg_{s'} \cpsb{e}[x^2/x] \]
    Where $s'$ is either $\cps{s}$ or $\dneg{\cps{s}}$ depending on
    $\varepsilon$. The proof is independent of this.
    What we need to prove is:
        \[ \Gamma^{12} \vDash (\lambda x:\cps{t}. \cpsa{e}) \stg_{\cps{t} -> s'} (\lambda x:\cps{t}. \cpsb{e}) \]
    Which by definition is:
        \[\arraycolsep=1pt\begin{array}{lcr}
        \Gamma^{12} &\vDash& \forall x^1, x^2:\cps{t}.\,
        x^1 \stg_{\cps{t}} x^1 \land
        x^1 \stg_{\cps{t}} x^2 \land
        x^2 \stg_{\cps{t}} x^2 \implies \\
            && (\lambda x:\cps{t}. \cpsa{e})~x^1 \stg_{s'} (\lambda x:\cps{t}. \cpsb{e})~x^2
        \end{array}
        \]
    By reduction (\rulename{V-EqRed} + \rulename{V-Eq*}), this is equivalent to:
        \[\arraycolsep=1pt\begin{array}{lcr}
        \Gamma^{12} &\vDash& \forall x^1, x^2:\cps{t}.\,
        x^1 \stg_{\cps{t}} x^1 \land
        x^1 \stg_{\cps{t}} x^2 \land
        x^2 \stg_{\cps{t}} x^2 \implies \\
            && \cpsa{e}[x^1/x] \stg_{s'} \cpsb{e}[x^2/x]
        \end{array}
        \]
    Which we can conclude from three applications of
    \rulename{V-$\forall$i}, and our IH.

\item \rulename{ST-App}

    Say $\Gamma \vdash f : t \earr s \bang \neu$ and $\Gamma \vdash e :
    a \bang \neu$. As inductive hypothesis we get:
        \[ \Gamma^{12} \vDash \cpsa{f} \stg_{\cps{t} -> s'} \cpsb{f}
        \qquad
           \Gamma^{12} \vDash \cpsa{e} \stg_{\cps{t}} \cpsb{e} \]
    Where $s'$ is either $\cps{s}$ or $\dneg{\cps{s}}$ depending on
    $\varepsilon$. Again, the proof is independent of this.
    Applying \autoref{lemma:subst-12} for our second IH we get that:
    \[
        \Gamma^{12} \vDash \cpsa{e} \stg_{\cps{t}} \cpsa{e}
                     \land \cpsa{e} \stg_{\cps{t}} \cpsb{e}
                     \land \cpsb{e} \stg_{\cps{t}} \cpsb{e}
    \]
    Expanding the definition of $\stg$ on the IH for $f$ we get:
        \[\arraycolsep=1pt\begin{array}{lcr}
        \Gamma^{12} &\vDash& \forall x^1, x^2:\cps{t}.\,
        x^1 \stg_{\cps{t}} x^1 \land
        x^1 \stg_{\cps{t}} x^2 \land
        x^2 \stg_{\cps{t}} x^2 \implies \\
         && f^{*1}~x^1 \stg_{s'} f^{*2}~x^2
        \end{array}
        \]
    We instantiate (using \rulename{V-$\forall$e}) $x^1,x^2$ with
    $\cpsa{e}, \cpsb{e}$, and apply \rulename{V-MP} with our proof about
    $\cpsa{e}$ and $\cpsb{e}$ to get:
        \[ \Gamma^{12} \vDash f^{*1}~\cpsa{e} \stg_{s'} f^{*2}~\cpsb{e} \]
    Which is exactly our goal in any ($\narr, \tarr$) case.

\item \rulename{ST-Ret}

    Say $\Gamma \vdash \returnT{e} : t \bang \teff$. Our IH gives us:
        \[\small \Gamma^{12} \vDash \cpsa{e} \stg_{\cps{t}} \cpsb{e} \]
    And we need to show that:
        \[\small \Gamma^{12} \vDash (\lambda p.\, p~\cpsa{e}) \stg_{\dneg{\cps{t}}} (\lambda p.\, p~\cpsb{e}) \]
    That is:
        \[\arraycolsep=1pt\begin{array}{lcr}
        \Gamma^{12} &\vDash& \forall p^1, p^2.\,
        p^1 \stg p^1 \land
        p^1 \stg p^2 \land
        p^2 \stg p^2 \implies \\
            && (\lambda p.\, p~\cpsa{e})~p_1 \stg_{s'} (\lambda p.\, p~\cpsb{e})~p_2
        \end{array}
        \]
    By reduction this is:
        \[\arraycolsep=1pt\begin{array}{lcr}
        \Gamma^{12} &\vDash& \forall p^1, p^2.\,
        p^1 \stg p^1 \land
        p^1 \stg p^2 \land
        p^2 \stg p^2 \implies \\
            && p_1~\cpsa{e} \stg_{s'} p_2~\cpsb{e}
        \end{array}
        \]
    By applying \autoref{lemma:subst-12} to $e$ we get:
    \[
        \small \Gamma^{12} \vDash
        \cpsa{e} \stg_{\cps{t}} \cpsa{e} \land
        \cpsa{e} \stg_{\cps{t}} \cpsb{e} \land
        \cpsb{e} \stg_{\cps{t}} \cpsb{e}
    \]
    With this, we can easily prove our goal by applying
    \rulename{V-$\forall$i}, and then the assumption of $p_1 \stg p_2$
    applied to this last proof.

\item \rulename{ST-Bind}

    Say $\Gamma \vdash m : a \bang \teff$ and $\Gamma, x:a \vdash e :
    b \bang \teff$, so we get $\Gamma \vdash \bindT{m}{x}{e} : b \bang
    \teff$. Our IHs are:
    \[\small
      \arraycolsep=1pt
      \begin{array}{rl}
          \Gamma^{12} & \vDash \cpsa{m} \stg_{\dneg{a}} \cpsb{m} \\
          (\Gamma, x:a)^{12} & \vDash \cpsa{e}~[x^1/x] \stg_{\dneg{b}} \cpsb{e}~[x^2/x]
      \end{array}
    \]
    We need to show that:
        \[\small \Gamma^{12} \vDash (\lambda p.~\cpsa{m}~(\lambda x.\, \cpsa{e}~p)) \stg_{\dneg{\cps{b}}} (\lambda p.~\cpsb{m}~(\lambda x.\, \cpsb{e}~p)) \]
    Which by expanding the definition, applying
    \rulename{V-$\forall$i}, and reducing can be simplified to:
        \[\arraycolsep=1pt\begin{array}{l}
        \Gamma^{12}, p^1, p^2,
        [ p^1 \stg p^1 \land
          p^1 \stg p^2 \land
            p^2 \stg p^2 ] \vDash \\
            \qquad \qquad \qquad
            \cpsa{m}~(\lambda x.\, \cpsa{e}~p^1)
            \stg
            \cpsb{m}~(\lambda x.\, \cpsb{e}~p^2)
        \end{array}
        \]
    By our IH we know that $\cpsa{m} \stg \cpsb{m}$, so it would be enough
    to show that:
        \[\arraycolsep=1pt\begin{array}{l}
        \Gamma^{12}, p^1, p^2,
        [ p^1 \stg p^1 \land
          p^1 \stg p^2 \land
            p^2 \stg p^2 ] \vDash \\
            \qquad \qquad \qquad \qquad \qquad
            (\lambda x.\, \cpsa{e}~p^1)
            \stg
            (\lambda x.\, \cpsb{e}~p^2)
        \end{array}
        \]
    and then use \autoref{lemma:subst-12}. Expanding the definitions,
    applying \rulename{V-$\forall$i} and reducing this can be shown by:
        \[\arraycolsep=1pt\begin{array}{l}
        \Gamma^{12},
          p^1, p^2,
        [ p^1 \stg p^1 \land
          p^1 \stg p^2 \land
          p^2 \stg p^2 ] \hphantom{\vDash}\\
          \hfill
          x^1, x^2,
        [ x^1 \stg x^1 \land
          x^1 \stg x^2 \land
          x^2 \stg x^2 ] \vDash \\
            \hfill
            \cpsa{e}[x^1/x]~p^1
            \stg
            \cpsb{e}[x^2/x]~p^2
            \hfill
        \end{array}
        \]
    Because of our assumptions for $p^1$ and $p^2$, this can be shown by
    proving $\cpsa{e}[x^1/x] \stg \cpsb{e}[x^2/x]$.
    This is trivial by our IH for $e$, weakening it into this extended
    environment that includes $p^i$.

\item \rulename{ST-Pair}, \rulename{ST-Fst}, \rulename{ST-Inl}

    All trivial from IHs.

\item \rulename{ST-Case}

    By case analysis on the IH for the sum type, and reduction.

\end{enumerate}
\end{proof}

Having this proof implies that any well-typed term will be given a
monotonic specification. And, as a consequence, functions preserve
monotonicity.

\subsection{Conjunctivity}

The definition of conjunctivity on \emf predicate types was given
previously. The full theorem which we prove is this:

\gm{NOTE: For WPs on simple types conjunctivity implies monotonicity, by
noting that if $p_1 => p_2$ then $p_1 \land p_2 = p_1$ and unfolding.
Thus a reader might wonder why we prove both. The answer is that (1) at
a higher-order that's probably not true (by our definitions) and (2) our
conjunctivity is only on the postcondition, or the final argument, not
on every argument as monotonicity}

\newcommand{\GammaC}{\Gamma_\mathbb{C}}

\begin{theorem}[Conjunctivity of \cpst --- environments]
For any $\Delta, \Gamma, e, H, A$ one has:
\[\small
\begin{array}{llll}
    \mbox{1.} & \Delta \mid \Gamma \vdash e : C \bang \neu  & \implies & \un\Delta, \GammaC \vDash \C{\cps{C}}{\cps{e}} \\
    \mbox{2.} & \Delta \mid \Gamma \vdash e : A \bang \teff & \implies & \un\Delta, \GammaC \vDash \C{\dneg{\cps{A}}}{\cps{e}} \\
\end{array}
\]

Where when $\Gamma = x_1:t_1,\ldots$, we define $\GammaC =
x_1:\cps{t_1}, [\C{\cps{t_1}}{x_1}],\ldots$. This trivially implies the
previously stated theorem by taking $\Gamma = \cdot$.
\end{theorem}

\begin{proof}
    By induction on the typing derivations. Once again, $\Delta$ does
    not play a big role and we omit it.

\begin{enumerate}
\item \rulename{ST-Var}

    Trivial from context, for any type.

\item \rulename{ST-Const}

    Does not apply as no constant gives a type $C\bang\neu$ nor $A\bang\teff$

\item \rulename{ST-Abs}

    Say we concluded $\Gamma, x:t \vdash e:s\bang\varepsilon$ (where
    that might be $C\bang\neu$ or $A\bang\teff$, we treat both cases
    uniformly). From the IH we get
        \[
            \GammaC, x:\cps{t}, [\C{\cps{t}}{x}] \vDash \C{s'}{e}
        \]
    Where $s'$ is $\cps{s}$ or $\dneg{\cps{s}}$ according to
    $(s,\varepsilon)$. By applying \rulename{V-$\forall$i} twice we get:
        \[
            \GammaC \vDash \forall x:\cps{t}. \C{\cps{t}}{x} => \C{s'}{e}
        \]
    Which is the same, by reduction, as:
        \[
            \GammaC \vDash \forall x:\cps{t}. \C{\cps{t}}{x} => \C{s'}{(\lambda x.~e~x)~x}
        \]
    Thus by definition of $\mathbb{C}$:
        \[
            \GammaC \vDash \C{t -> s'}{\lambda x.~e~x}
        \]
    As required for both cases.

\item \rulename{ST-App}

    Trivial by the preservation of $\mathbb{C}$ by application, in both
    cases (applies \rulename{V-MP}).

\item \rulename{ST-Ret}

    Say we concluded $\Gamma \vdash \returnT{e} : A\bang\teff$. Our goal is then:
        \[
            \GammaC \vDash \C{\dneg{\cps{A}}}{\lambda p.~p~\cps{e}}
        \]
    Which is:
        \[\arraycolsep=0pt\begin{array}{ll}
            \GammaC \vDash \forall p_1, p_2. & ~(\lambda p.~p~\cps{e})p_1 \land (\lambda p.~p~\cps{e})p_2 \\
                                             & = (\lambda p.~p~\cps{e})(\lambda x. p_1~x\land p_2~x) \\
        \end{array}\]
    By reduction that's equivalent to:
        \[
            \GammaC \vDash \forall p_1, p_2. ~p_1~\cps{e} \land p_2~\cps{e} = p_1~\cps{e} \land p_2~\cps{e} \\
        \]
    Which is trivially true (without use of any IH) by \rulename{V-Refl}.

\item \rulename{ST-Bind}

    Say we concluded $\Gamma \vdash \bindT{e_1}{x}{e_2} : A' \bang \teff$, where $e_1 : A\bang\teff$. Our IHs are:
        \[\arraycolsep=1pt\begin{array}{rl}
            \GammaC \vDash & \C{\dneg{\cps{A}}}{\cps{e_1}} \\
            \GammaC, x:\cps{A}, [\C{\cps{A}}{x}] \vDash & \C{\dneg{\cps{A'}}}{\cps{e_2}} \\
          \end{array}\]
    We need to show:
        \[
            \GammaC \vDash \C{\dneg{\cps{A'}}}{\lambda p. \cps{e_1}(\lambda x. \cps{e_2}~p)}
        \]
    Expanding the definition, this is:
        \[\arraycolsep=0pt\begin{array}{ll}
            \GammaC \vDash \forall p_1, p_2.~~ & (\lambda p. \cps{e_1}(\lambda x. \cps{e_2}~p))~p_1 \land (\lambda p. \cps{e_1}(\lambda x. \cps{e_2}~p))~p_2 \\
                                               & = (\lambda p. \cps{e_1}(\lambda x. \cps{e_2}~p))(\lambda x. p_1~x \land p_2~x) \\
        \end{array}\]
    By reduction, this is equivalent to:
        \[\arraycolsep=0pt\begin{array}{ll}
            \GammaC \vDash \forall p_1, p_2.~~ & \cps{e_1}(\lambda x. \cps{e_2}~p_1) \land \cps{e_1}(\lambda x. \cps{e_2}~p_2) \\
                                               & = \cps{e_1}(\lambda x. \cps{e_2}~(\lambda x. p_1~x \land p_2~x)) \\
        \end{array}\]
    By the IH for $e_2$ we know
    $\forall x.~\cps{e_2}~(\lambda x. p_1~x \land p_2~x) = \cps{e_2}p_1 \land \cps{e_2}p_2$.
    By reduction and \rulename{V-Ext} this means
    $(\lambda x. \cps{e_2}~(\lambda x. p_1~x \land p_2~x)) = (\lambda x. \cps{e_2}p_1 \land \cps{e_2}p_2)$
    Thus we replace on the RHS (via \rulename{V-Subst}) and get:
        \[\arraycolsep=0pt\begin{array}{ll}
            \GammaC \vDash \forall p_1, p_2.~~ & \cps{e_1}(\lambda x. \cps{e_2}~p_1) \land \cps{e_1}(\lambda x. \cps{e_2}~p_2) \\
                                               & = \cps{e_1}(\lambda x. \cps{e_2}p_1 \land \cps{e_2}p_2) \\
        \end{array}\]
    By some $\eta$-expansion and the IH for $e_1$ we can turn this to:
        \[\arraycolsep=0pt\begin{array}{ll}
            \GammaC \vDash \forall p_1, p_2.~~ & \cps{e_1}(\lambda x. \cps{e_2}~p_1) \land \cps{e_1}(\lambda x. \cps{e_2}~p_2) \\
                                               & = \cps{e_1}(\lambda x. \cps{e_2}p_1) \land \cps{e_1}(\lambda x. \cps{e_2}p_2) \\
        \end{array}\]
    Which is trivially provable by \rulename{V-Refl}.

\item \rulename{ST-Pair}, \rulename{ST-Fst}

    All trivial by IHs.

\item \rulename{ST-Inl}

    Does not apply for the cases we consider.

\item \rulename{ST-Case}

    Trivial by \rulename{V-SumInd} and the IHs.

\end{enumerate}
\end{proof}

Thus, any term obtained by the \cpst (return, bind, actions, lifts,
...) will be conjunctive in this sense, which means they also preserve
the property through application.

With a completely analogous definition and proof we get the expected
result of conjunctivity over (non-empty) universal quantification.
The non-empty requirement is not actually stressed during that proof,
but it's the wanted result as WPs (which can be taken as arguments)
might not distribute over empty universals (in particular,
non-satisfiable WPs do not).

\ch{Not clear what you mean by range}
\gm{Basically that the forall is at least over one element (you can
search the Coq proof for 'inhabited'), or the property doesn't hold (as
it should not). If you do an empty forall you always get True on one
side and the other can be just anything. For anyone familiar with the
property it should be clear, but I could put some more words}
\gm{Though of it a bit more, our WPs distribute over both but we're
interested in non-empty, as that's the case for real F*}

\fi

\bibliographystyle{abbrvnaturl}
\bibliography{fstar}

\end{document}